%% file: main_TwoColoumn.tex
\documentclass[10pt,journal,twocolumn]{IEEEtran}
\usepackage{cite,array}
\usepackage{amsmath,amssymb,amsfonts}
\usepackage{algorithmic}
\usepackage{textcomp}
\usepackage[table]{xcolor}
\usepackage{makecell}
\usepackage{psfrag}
\usepackage{arydshln}
\usepackage{url}
\usepackage{soul}
\usepackage{tabularx}
\usepackage{comment}
\usepackage{graphicx,color}
\usepackage{epstopdf}
\usepackage[nolist]{acronym}
\usepackage{multirow}
\usepackage{longtable}
\usepackage{adjustbox}

\usepackage{etoolbox}
\usepackage{wrapfig}
\usepackage{epsfig}
\usepackage{color}
\usepackage{amsmath}
\usepackage{amssymb}
\usepackage{amsthm}
\newtheorem{theorem}{Theorem}[]

\usepackage{mathabx}
\usepackage{enumerate}
\usepackage{multirow}
\usepackage{bbm}
\usepackage{algorithm}
\usepackage{algorithmic}
\usepackage{lipsum,amsmath,multicol}
\usepackage{stfloats}
\usepackage{arydshln} 
\usepackage{amsfonts}
\usepackage{dsfont}
\usepackage{upgreek}
\usepackage{tikz}
\usepackage{pgfplots}
\usepackage{pgfplotstable}
\pgfplotsset{compat=newest}
 \usetikzlibrary{plotmarks}
  \usetikzlibrary{arrows.meta}
  \usepgfplotslibrary{patchplots}
   \usepackage{subcaption}
\usepackage{verbatim}
\usepackage{booktabs} 
\usepackage{stfloats}
 
\usepackage[capitalise]{cleveref}
\Crefname{equation}{Eq.\!}{Eqs.\!}
\Crefname{figure}{Fig.\!}{Figs.\!}
\Crefname{tabular}{Tab.\!}{Tabs.\!}
\Crefname{section}{Section\!}{Sections.\!}

\definecolor{Linen}{rgb}{0.9803,0.9411,0.9019}
\definecolor{White}{rgb}{1,1,1}
\definecolor{Lightred}{rgb}{1,0.3803,0.3803}
\definecolor{Coral}{rgb}{1,0.4980,0.3137}
\definecolor{Grayblue}{rgb}{0.9411,0.9411,0.9803}
\definecolor{DarkLinen}{rgb}{0.729,0.7176,0.635}
\begin{document}
\input{Definitions.tex}
\title{The Rendezvous Between Extreme Value Theory and Next-generation Networks}
\author{\IEEEauthorblockN{Srinivas Sagar, Athira Subhash, Chen-Feng Liu, Ahmed Elzanaty,\\Yazan H. Al-Badarneh, Sheetal Kalyani, Mohamed-Slim Alouini, Mehdi Bennis, and Lajos Hanzo}
\thanks{Srinivas Sagar and Sheetal Kalyani are with the Dept. of Electrical Engg., IIT Madras, India. (email: \{ee21d051@smail, 
skalyani@ee\}.iitm.ac.in)}
\thanks{Athira Subhash (email: \{athira3003@gmail.com\})}
\thanks{Chen-Feng Liu is with the Department of Informatics, New Jersey Institute of Technology, Newark, NJ 07102 USA (e-mail: chenfeng.liu@njit.edu).}
\thanks{A. Elzanaty is with the 5GIC \& 6GIC, Institute for Communication Systems (ICS), University of Surrey, Guildford, GU2 7XH, United Kingdom (e-mail: a.elzanaty@surrey.ac.uk)}
\thanks{Y. H. Al-Badarneh is with the Department of Electrical Engineering, The University of Jordan, Amman, 11942 (email: yalbadarneh@ju.edu.jo).}
\thanks{M.-S. Alouini is with King Abdullah University of Science and Technology (KAUST), Thuwal, Makkah Province, Saudi
Arabia. (e-mail:slim.alouini@kaust.edu.sa).}
\thanks{Mehdi Bennis is with the Centre for Wireless Communications, University of Oulu, 90014 Oulu, Finland (e-mail: mehdi.bennis@oulu.fi).}
\thanks{L. Hanzo is with the School of Electronics and Computer Science, University of Southampton, Southampton SO17 1BJ, U.K. (e-mail:lh@ecs.soton.ac.uk).}
}
\maketitle

	\acresetall 
	
\tableofcontents

\input{Acronyms.tex}

\begin{abstract}
	\acresetall 

Promising technologies such as massive multiple-input and multiple-output, reconfigurable intelligent reflecting surfaces, non-terrestrial networks, millimetre wave communication, ultra-reliable low-latency communication are envisioned as the enablers for next-generation (NG)  networks. In contrast to conventional communication systems meeting specific average performance requirements, NG systems are expected to meet quality-of-service requirements in extreme scenarios as well. In this regard, extreme value theory (EVT) provides a powerful framework for the design of communication systems. In this paper, we provide a comprehensive survey of advances in communication that utilize EVT to characterize the extreme order statistics of interest. We first give an overview of the history of EVT and then elaborate on the fundamental theorems. Subsequently, we discuss different problems of interest in NG communication systems and how EVT can be utilized for their analysis. We finally point out the open challenges and future directions of EVT in NG communication systems.
\end{abstract}

\section{Introduction}
\subsection{Overview of EVT}
The statistics of extreme events has been of interest in various domains like climate studies, risk analysis in finance, 
strength determination in mechanical systems, component failure in the electronics industry, best/worst signal power determination in wireless communications, etc. \ac{EVT} is a branch of statistics that deals with the statistics of the extremes in sequences of \ac{RVs}. The field of \ac{EVT} was pioneered by Leonard Tippett (1902–1985), who worked at the British Cotton Industry Research Association. During his study of the strengths of cotton threads, he observed that their strength was determined by the weakest fibres within it. Collaborating with R. A. Fisher, Tippett derived three asymptotic limits that describe the distributions of the extremes of sequences of \ac{RVs}, assuming independence among the variables. This result was later further augmented by Gumbel in his work ``Statistics of Extremes'' in 1958 \cite{coles2001introduction}. 
\par
Like the \ac{CLT} that unifies the distribution of the average of sequences of \ac{RVs}, EVT claims that the distribution of the maximum/minimum of a sequence of \ac{i.i.d.} \ac{RVs} will always be one of the three \ac{EVDs}: Gumbel, Frechet, and Weibull, under certain uniformity assumptions. This result holds for a wide range of \ac{RVs}, even if the original distribution function is not an \ac{EVDs}. The versatility of the \ac{CLT} has led to its extensive use in deriving statistical inferences in numerous domains.
While the \ac{CLT} deals with central order statistics, events related to low probabilities fall within the realm of extreme order statistics.  
Similar to the \ac{CLT}, once the limiting distribution and the corresponding normalization in \ac{EVT} are established, we can employ the properties of the \ac{EVDs} to derive statistical inferences
\begin{figure*}[t!]
        \centering
\includegraphics[scale=0.6]{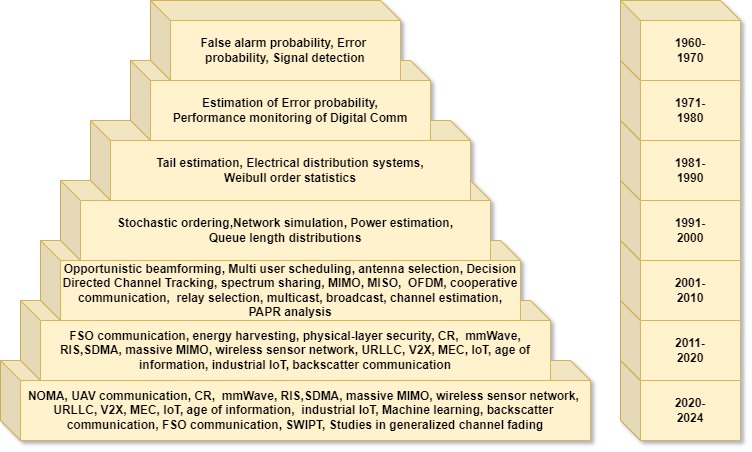}
        \caption{\ac{EVT} evolution and Communication Pyramid}
        \label{fig:pyramid}
    \end{figure*}
\subsection{Overview of \ac{EVT} in Communications}
The study of rare and extreme events holds great importance in the realm of communication systems due to their potential to inflict significant disruptions in the communication channel. Thus, comprehending the behavior and impact of these events on system performance becomes essential.
As a branch of probability theory, \ac{EVT} offers a mathematical framework for modeling and analyzing the statistics of such extreme events. This makes \ac{EVT} eminently suitable for system design and performance analysis in the rare and extreme scenarios encountered in wireless communications. 

\par By modeling the distribution of extreme events and conducting statistical inferences, \ac{EVT} facilitates predictions regarding the likelihood of future extreme events. This crucial information aids decision-making in system design and operational optimization within wireless communication systems. Furthermore, \ac{EVT} provides a means to model the relationship between communication system parameters and the occurrence of extreme events. This modeling information, furnished by \ac{EVT}, proves indispensable in optimizing the design and operation of communication systems. Overall, the incorporation of \ac{EVT} in wireless communications enhances our understanding and management of system behaviour under extreme conditions, resulting in reliable and robust designs.

\subsection{History of \ac{EVT} in Communications}\label{Sec: EVT history}
    The history of \ac{EVT} in communications dates back to the late 20-th century. In the early stages, \ac{EVT} was used to detect binary signals in the presence of \ac{AWGN} \cite{fine1962extreme},\cite{milstein1969robust}. First, \ac{AWGN} was passed through the signal detector to find the distribution of the output signal. Considering the output samples as \ac{i.i.d.} \ac{RVs}, the distribution of maximum was modeled as Gumbel distribution using \ac{EVT}. Furthermore, false alarm probability was derived to detect the presence of a signal using the Gumbel distribution \cite{fine1962extreme}. It was also used in deriving the error probabilities of binary communication receivers \cite{ashlock1966error}. Over the years, as shown in Fig.~\ref{fig:pyramid}, \ac{EVT} has evolved from being a tool for analyzing the reliability of communication systems to a framework for characterizing the performance of communication systems under extreme conditions in an ever-increasing plethora of applications.
\par {The use of \ac{EVT} in wireless communications has led to advances in various areas, including performance evaluation \cite{pun2010performance,kalyani2012analysis, 7883953, biswas2015performance, subhash2019asymptotic, sagar2023multi, subhash2022asymptotic}, channel estimation \cite{4151173, kalyani2006extreme, kalyani2007mitigation, kalyani2006leverage}, system optimization \cite{low2010optimized, abdel2019optimized, subhash2020transmit}, channel modeling \cite{mehrnia2021wireless, mehrnia2023multivariate}, etc. \ac{EVT} has been applied to a wide range of topics in communication engineering, including multi-user diversity and antenna selection  \cite{narasimhamurthy2011multi,kalyani2012analysis, al2017asymptotic,subhash2019asymptotic}, anomaly 
 detection \cite{vignotto2020extreme} , decision directed channel tracking \cite{kalyani2006extreme, kalyani2007mitigation, kalyani2006leverage}, selection diversity in \ac{CR} \cite{xia2013spectrum,hong2011throughput,subhash2020transmit}, relays and \ac{RIS} \cite{subhash2021cooperative,oyman2007opportunism,oyman2008power,biswas2015performance,dimitropoulou2020k, sagar2023multi}, \ac{MIMO} and cell-free massive \ac{MIMO} systems \cite{pun2010performance,moon2011sum,kazemi2020analysis, gao2017massive,park2009outage}, \ac{SWIPT}  \cite{subhash2021cooperative, liu2019cooperative},  \ac{URLLC} \cite{mouradian2016extreme,mehrnia2021wireless,mehrnia2022extreme, LiuBenPoo17,liu2019dynamic,zhou2020learning, zhu2021reliability, samarakoon2019distributed, HsuLiuWeiBen22}, and \ac{ML} \cite{scheirer2014probability,zhang2016sparse,bendale2016towards,oza2019c2ae,yang2021conditional,ni2021open,yu2021deep,rudd2017extreme,vignotto2020extreme,chen2023open}. In the following sub-sections, we explore the key contributions in the major applications of \ac{EVT}.} 
 
    \textbf{Multi-User Diversity and Antenna Selection:} \ac{EVT} has found practical employment in addressing multi-user diversity  \cite{song2006asymptotic, narasimhamurthy2011multi,al2017asymptotic} and antenna selection \cite{gao2017massive,subhash2019asymptotic} problems, which are crucial aspects of wireless communication systems. \ac{EVT} is exploited to characterize the best channel link for multi-user diversity and to characterize the performance under selection combining. It also provides a valuable framework for examining the distribution of the maximum/minimum received signal strength in wireless communications.

    \textbf{Selection Diversity in \ac{CR}:} \ac{EVT} has also been used in \ac{CR} systems \cite{xia2013spectrum,hong2011throughput,subhash2020transmit}, \cite{al2018asymptotic, al2019asymptotic} to analyze the performance of spectrum sensing and selection diversity  \cite{duan2019asymptotic}. Note that in CR applications characterizing the worst primary user/\ac{SU} channel and the strongest \ac{SU} channel are important aspects of system design and \ac{EVT} has proved to be an efficient tool in this. 
    
    \textbf{Relays and \ac{RIS}s:} {\ac{EVT} has been used in the analysis of relay-assisted communication systems \cite{oyman2007opportunism,oyman2008power,biswas2015performance,dimitropoulou2020k, liu2019cooperative, subhash2021cooperative} and \ac{RIS}s \cite{sagar2023multi}.} It has been shown that \ac{EVT} can provide a useful framework for analyzing the distribution of the received signal strength under the best relay/\ac{RIS} selection scheme in multi-relay/\ac{RIS} assisted communication systems. Such analyses are useful for practical system design exemplified by deciding on the optimal placement of relays and \ac{RIS}s, transmit power, etc.
    
    \textbf{Massive \ac{MIMO} and Cell-Free Massive \ac{MIMO}:} \ac{EVT} has also been used in the analysis of massive \ac{MIMO} \cite{pun2010performance,moon2011sum,kazemi2020analysis, gao2017massive,park2009outage} and cell-free massive \ac{MIMO} systems. Furthermore it can be used in characterizing the best beamformer, best antenna, and best combining method to evaluate the system's performance.

   \textbf{\ac{URLLC} Applications:} \ac{EVT} has proven to be highly valuable in analyzing the performance of \ac{URLLC} applications \cite{liu2018ultra}. In \ac{URLLC} systems \cite{mouradian2016extreme,mehrnia2021wireless,mehrnia2022extreme, LiuBenPoo17,liu2019dynamic,zhou2020learning, zhu2021reliability, samarakoon2019distributed, HsuLiuWeiBen22}, meeting stringent delay requirements presents significant challenges for system design and performance analysis. \ac{EVT} offers a powerful tool of characterizing the tail distribution of delay and outage probability. By harvesting \ac{EVT}, researchers and engineers can gain insights into the extreme events associated with delays and outages in \ac{URLLC} systems. This characterization is crucial for accurately assessing system's performance and ensuring that the stringent latency requirements are met. Furthermore, \ac{EVT} facilitates the development of robust \ac{URLLC} system designs and efficient resource allocation strategies to achieve reliable and low-latency communications.

   \textbf{Decision Directed Channel Tracking and Channel Estimation:} \ac{EVT} has also been used for mitigating the error propagation in the systems employing decision directed channel tracking \cite{4151173, kalyani2006extreme, kalyani2006leverage, kalyani2007mitigation}. In the above works decision directed channel estimation has been treated as an outlier contaminated estimation problem and EVT was employed to detect the outliers and hence mitigate error propagation thus enabling robust channel tracking/estimation. 
      
   \textbf{\ac{ML}:} As a further advance \ac{EVT} has been used to analyze communication systems with \ac{ML} \cite{yang2021conditional,ni2021open,yu2021deep}. It has been shown that \ac{EVT} can provide a potent framework for analyzing the distribution of learning errors in \ac{ML} systems, allowing for optimizing learning algorithms. The application of \ac{EVT} in \ac{ML} allows for identifying outliers, anomalies, or rare events that may significantly impact the learning process. This knowledge helps in developing robust and efficient learning algorithms that can handle diverse challenging communication scenarios. The combination of \ac{EVT} and \ac{ML} techniques offers great potential for enhancing the performance and reliability of communication systems.

    {The authors of \cite{kalyani2012asymptotic} harnessed \ac{EVT} to find the asymptotic \ac{PDF} of the maxima of $n$ \ac{i.i.d.} sums of \ac{i.n.i.d.} or correlated Gamma \ac{RVs} and proved that it follows the Gumbel \ac{CDF}.
\par In addition to the significant applications mentioned earlier, \ac{EVT} has also found favours in various other areas within wireless communications. For instance, \ac{EVT} has been employed in the design of energy-efficient networks, enabling researchers to analyze and optimize the energy consumption in wireless communication systems \cite{bahl2015asymptotic, liu2019cooperative, subhash2021cooperative}.  \ac{EVT} has also been utilized in signal detection \cite{milstein1969robust}, where it aids in identifying and characterizing rare or extreme signals buried in noise. In the context of \ac{OFDM}, \ac{EVT} has been used for  the optimization of system performance \cite{wei2002modern, jiang2008derivation, kalyani2006extreme, kalyani2007mitigation, 7883953}. Furthermore, it has been applied in user scheduling to determine the optimal scheduling \cite{al2012asymptotic, kalyani2012analysis, liu2019cooperative, kalyani2012asymptotic} strategy that maximizes system capacity or fairness. \ac{EVT} has also been instrumental in overhead threshold selection \cite{m2021efficient}, providing a framework for selecting appropriate thresholds in various wireless communications algorithms. \ac{EVT} has also been utilized in wireless localization \cite{YanGuoGuoZhaZha20}, aiding in the accurate estimation of mobile device positions by analyzing the statistical distribution of received signal strength. These are just a few \textcolor{black}{representative} examples of the diverse range of applications, where \ac{EVT} has demonstrated its effectiveness in enhancing wireless communication systems.
\begin{table*}   
        \centering
         \caption{Positioning of this work against other relevant papers surveying \ac{EVT}. }
 	  \label{tab:comp_surveys}    
        \resizebox{\textwidth}{!}{%
        \begin{tabular}{|l|c|c|ccc|c|}
        \hline
        \rowcolor[HTML]{FD6864} 
        \multicolumn{1}{|c|}{\cellcolor[HTML]{FD6864}} &
          \cellcolor[HTML]{FD6864} &
          \cellcolor[HTML]{FD6864} &
          \multicolumn{3}{c|}{\cellcolor[HTML]{FD6864}\textbf{Perspective}} &
          \cellcolor[HTML]{FD6864} \\ \cline{4-6}
        \rowcolor[HTML]{FD6864} 
        \multicolumn{1}{|c|}{\multirow{-2}{*}{\cellcolor[HTML]{FD6864}\textbf{Paper}}} &
          \multirow{-2}{*}{\cellcolor[HTML]{FD6864}\textbf{Year}} &
          \multirow{-2}{*}{\cellcolor[HTML]{FD6864}\textbf{\begin{tabular}[c]{@{}c@{}}Introduction to \\  EVT theory\end{tabular}}} &
          \multicolumn{1}{c|}{\cellcolor[HTML]{FD6864}\textbf{Mathematical}} &
          \multicolumn{1}{c|}{\cellcolor[HTML]{FD6864}\textbf{Finance}} &
          \textbf{\begin{tabular}[c]{@{}c@{}}Communication\\   Engineering\end{tabular}} &
          \multirow{-2}{*}{\cellcolor[HTML]{FD6864}\textbf{\begin{tabular}[c]{@{}c@{}}Limitations and future\\  directions in communications\end{tabular}}} \\ \hline
        \rowcolor[HTML]{F5E6E6} 
        \cite{rocco2014extreme} &
          2014 &
          Covered &
          \multicolumn{1}{c|}{\cellcolor[HTML]{F5E6E6}Yes} &
          \multicolumn{1}{c|}{\cellcolor[HTML]{F5E6E6}Yes} &
          No &
          No \\ \hline
        \rowcolor[HTML]{F5E6E6} 
       \cite{gomes2015extreme} &
          2015 &
          Covered &
          \multicolumn{1}{c|}{\cellcolor[HTML]{F5E6E6}Yes} &
          \multicolumn{1}{c|}{\cellcolor[HTML]{F5E6E6}No} &
          No &
          No \\ \hline
        \rowcolor[HTML]{F5E6E6} 
       \cite{hansen2020three} &
          2020 &
          \begin{tabular}[c]{@{}c@{}}Partially\\ Covered\end{tabular} &
          \multicolumn{1}{c|}{\cellcolor[HTML]{F5E6E6}Yes} &
          \multicolumn{1}{c|}{\cellcolor[HTML]{F5E6E6}No} &
          No &
          No \\ \hline
        \rowcolor[HTML]{F5E6E6} 
        This work &
          2024 &
          Covered &
          \multicolumn{1}{c|}{\cellcolor[HTML]{F5E6E6}Yes} &
          \multicolumn{1}{c|}{\cellcolor[HTML]{F5E6E6}No} &
          \begin{tabular}[c]{@{}c@{}}Present the detailed review\\ in the context of\\ communications applications\\ like diversity, relays,\\ spectrum sharing,\\ MIMO, URLLC, and ML.\end{tabular} &
          Yes \\ \hline
        \end{tabular}%
        }
\end{table*}
   \par \ac{EVT} can also be used to statistically characterize the peaks exceeding a threshold as a rare event. Such an analysis is usualy referred to as the peak-over-threshold approach and the authors of \cite{kumar2017modeling} used it for characterizing the statistical properties of the peaks of an \ac{OFDM} signal over a given threshold. This characterization is then used to derive approximate expressions for the symbol error probability (SEP) in \ac{OFDM} systems in the presence of clipping effects due to \textcolor{black}{insufficient} power amplifier back-off and has been demonstrated to be significantly more accurate and tighter with respect to other approximations available in the literature. 
  \color{black}
    
    {Despite its wide range of applications, \ac{EVT} has its limitations, when it comes to modeling and analyzing communication systems. Note that \ac{EVT} may not be suitable for analyzing systems having complex dependencies between various components, if these lead to correlated/dependent \ac{RVs} as most EVT results assume independent \ac{RVs}. Additionally, \ac{EVT} may not be accurate for characterizing maximum/minimum order statistics over very small sample sizes, since all the results in EVT are asymptotic in nature.}
    
    The future directions of using \ac{EVT} in \ac{NG} systems are likely to involve system performance analysis and optimization. This could include the development of new \ac{EVT}-based algorithms for wireless resource allocation, energy efficiency, optimization, and wireless localization. Additionally, \ac{EVT} could be used to analyze the performance of new \ac{NG} communication technologies, such as terahertz communication, holographic \ac{MIMO}, and autonomous networks. Finally, there is scope for further research into the limitations of \ac{EVT} and the development of new methods for overcoming these limitations, to make \ac{EVT} an even more powerful tool for the analysis and optimization of \ac{NG} systems.

\subsection{Related Review Articles}   
\begin{table*}[b]  
\centering
 \caption{List of acronyms}\label{Table: Acronyms}
\begin{tabular}{|l| l| r| r|}
\hline
 Acronym & Definition &  Acronym & Definition  
\\\hline
\hline
AF&Amplify-and-forward
&AI & Artificial intelligence 
\\\hline
AoI&Age of information
&AP&Access point
\\\hline
AV&Activation vector
&AWGN & Additive white Gaussian noise
\\\hline
BAC-NOMA&Backscatter communication-assisted non-orthogonal multiple access
&BEP&Bit error probability
\\\hline
BER&Bit error rate
&BS&Base station
\\\hline
CAP&Compact abating probability
& CDF & Cumulative distribution function
\\\hline
CFAR&Constant false alarm rate
&CIP&Class-inclusion probability
\\\hline
CLT & Central limit theorem 
&CNN&Convolutional neural network
\\\hline
COSFD&Cross-domain OSFD
&CR & Cognitive radio
\\\hline
CSI&Channel state information
&DC&Direct current
\\\hline
DDRN&Deep discriminative representation network
&DF&Decode-and-forward
\\\hline
EGT&Equal gain transmission
& EMF & Electromagnetic field
\\\hline
ESC&Ergodic secrecy capacity
&EVD & Extreme value distribution
\\\hline
EVM&Extreme value machine
& EVT & Extreme value theory
\\\hline
E2E&End-to-end
&FSO&Free-space optical communication
\\\hline
GAN&Generative adversarial network
&GEV&Generalized extreme value
\\\hline
GPD & Generalized Pareto distribution
& i.i.d. & Independent and identically distributed
\\\hline
i.n.i.d. & Independent and non-identically distributed
&IoT&Internet of things
\\\hline
ISAC& Integrated sensing and communication 
&MAB&Multi-armed bandits
\\\hline
MAV&Mean activation vector
&MC&Multi-task based feature extractor and classifiers
\\\hline
MEC&Multi-access edge computing
&MGF&Moment-generating function
\\\hline
MIMO & Multiple-input multiple-output
&MISO&Multiple-input single-output
\\\hline
ML & Machine learning
&MRC&Maximum-ratio combining
\\\hline
MRT&Maximum ratio transmission
&MS&Mobile station
\\\hline
NG & Next-generation
&NN&Neural network
\\\hline
OC&Optimal combining
&OFDM & Orthogonal frequency-division multiplexing
\\\hline
OMS&Opportunistic multicast scheduling
&OSC&Open-set classifier
\\\hline
OSFD&Open set fault diagnosis
&OSmIL&Open set model with incremental learning
\\\hline
OSR&Open set recognition
&PAC&Probably approximately correct
\\\hline
PAPR&Peak-to-average power ratio
&PDF&Probability density function
\\\hline
PER & Packet error rate
&PMEPR&Peak-to-mean envelope power ratio
\\\hline
PR&Primary receiver
&PT&Primary transmitter
\\\hline
QoS&Quality-of-service
&RBF&Random beamforming
\\\hline
RIS& Reconfigurable intelligent surface
&RRO&Resource redistributive opportunistic
\\\hline
RSU&Roadside unit
& RV & Random variable 
\\\hline
SC&Selection combining
&SDMA&Space division multiple access
\\\hline
SINR&Signal-to-interference-plus-noise ratio
&SIR&Signal-to-interference ratio
\\\hline
SISO&Single-input single-output
&SNDR&Signal-to-noise-plus-distortion ratio
\\\hline
SNR & Signal-to-noise ratio
&SOSFD&Shared-domain OSFD
\\\hline
SR&Secondary receiver
&SRC&Sparse representation-based classification
\\\hline
ST&Secondary transmitter
&SU&Secondary user
\\\hline
SVM&Support vector machine
&SWIPT&Simultaneous wireless information and power transfer
\\\hline
TAS&Transmit antenna selection
&UAV&Unmanned aerial vehicle
\\\hline
URLLC & Ultra-reliable and low-latency communication
&VUE&Vehicular user equipment
\\\hline
V2V&Vehicle-to-vehicle
&WPS&Wireless powered communication system
\\\hline
 WPT-NOMA&WPT-assisted non-orthogonal multiple access
&ZFBF&Zero-forcing beamforming
\\\hline
3GPP & 3rd Generation Partnership Project &&
\\ \hline
 \end{tabular} 
\end {table*}
There are only a few articles that review \ac{EVT} and survey its applications, namely \cite{hansen2020three,gomes2015extreme,rocco2014extreme}. A concise introductory review of \ac{EVT} can be found in \cite{hansen2020three}, while \cite{gomes2015extreme} presents a review specifically focused on univariate extremes. Rocco's work \cite{rocco2014extreme} surveys the applications of \ac{EVT} in finance. 
However, to the best of the authors' knowledge, there is currently no comprehensive article providing a consolidated survey on the application of \ac{EVT} in wireless communications, despite its notable importance in this field, as summarized before. 

Motivated by this knowledge gap, the primary objective of this survey paper is to provide an overview of how \ac{EVT} results can be applied to various applications, such as multi-user diversity, \ac{URLLC}, \ac{MIMO}, and \ac{ML} for communication systems. In this paper, we aim to address several research questions and provide insights into the crucial performance metrics of communication systems, such as outage probability, sum-rate, scaling laws, delay outage, and throughput analysis. Furthermore, we explore the question of designing efficient communication systems while considering these performance metrics.
In contrast to the existing survey articles \cite{hansen2020three,gomes2015extreme,rocco2014extreme}, Table \ref{tab:comp_surveys} highlights the unique contributions of our paper. 

\par The contributions of this paper are as follows.
\begin{enumerate}
    \item 
    We present an overview of the history of applying \ac{EVT} in wireless communications. We highlight the importance of \ac{EVT} in characterizing the extreme order statistics. Starting with the analogy with \ac{CLT}, we present intuitions behind the application of \ac{EVT} in diverse communications fields. We also summarize the recent articles related to \ac{EVT}.
     \item
    We present the mathematical preliminaries required for understanding the \ac{EVT}. First, we introduce order statistics followed by the fundamental theorems in \ac{EVT} for  \ac{i.i.d.} \ac{RVs}. We also present how the scenarios differ in the case of  \ac{i.n.i.d.} \ac{RVs} and explain the fundamental theorems. We also emphasize the need for peak over threshold using \ac{EVT} and provide the mathematical setting in this case.

    \item
   We survey the applications of \ac{EVT} in different fields of wireless communications like multi-user diversity, selection diversity in \ac{CR}, relays and \ac{RIS}s, \ac{MIMO}, \ac{URLLC}, and \ac{ML}. For each field, we describe the system model, highlight the advantages of utilizing \ac{EVT}, showcase the applicability of \ac{EVT}, and review the state-of-the-art literature. We also briefly cover the application of \ac{EVT} in miscellaneous topics like energy efficiency, signal detection, \ac{OFDM}, user scheduling, threshold selection, and localization .

    \item
   We provide numerical results to validate the effectiveness of applying \ac{EVT} in the communication fields discussed.
   \item \textcolor{black}{Finally, we present the limitations and future directions of \ac{EVT} research as the take-away message.}
\end{enumerate}
\subsection{Notations}
\par The following notations are used in the paper. If $X$ is a random variable, then $f_{X}(.)$,  $F_{x}(.)$, and $\mathbb{E}[X]$ represent the \ac{PDF}, \ac{CDF}, and expectation of $X$.    The probability of event $A$ is denoted by $\mathbb{P}[A]$. Exponential functions are represented either by $e^{(.)}$ or $\exp(.)$. The $\ell_2$-norm of vector $\mathbf{y}$ and   transpose of $\mathbf{y}$ are denoted by $\left \| \mathbf{y} \right \|$, and $\mathbf{y}^{T}$ respectively. The maximum of $\left \{ X_1, X_2,\cdots,X_K \right \}$ is denoted by
 $\max \left \{ X_1, X_2,\cdots,X_K \right \}$ and $\arg \max \left \{ X_1, X_2,\cdots,X_K \right \}$ represents the index of the maximum. Also, $\left | i \right |$ denotes absolute value of $i$. {Furthermore, for the sake of readability, all acronyms and their definitions in this paper are listed in Table \ref{Table: Acronyms}.} 
\subsection{Organization}
In this work, we commence with an \ac{EVT} overview of various fields and also present the overview of \ac{EVT} in communications. After that, we present the detailed history of \ac{EVT} in wireless communications. In Section II, we present the mathematical preliminaries required for understanding \ac{EVT} and also the procedure of establishing order statistics with \ac{EVT}. Section III deals with the potential applications of \ac{EVT} in various fields of communications. Each application is described by a common system model and presented with the state of the art literature. Finally, its limitations and future directions are presented in Section IV.

\section{Mathematics behind Extreme Value Theory}
 \begin{figure*}
        \centering
        \tikzset{every picture/.style={line width=0.75pt}} 

        \begin{tikzpicture}[x=0.75pt,y=0.75pt,yscale=-1,xscale=1]
        
        \draw  [fill={rgb, 255:red, 240; green, 217; blue, 217 }  ,fill opacity=1 ] (60.67,85.33) -- (429,85.33) -- (429,125.33) -- (60.67,125.33) -- cycle ;
        \draw  [fill={rgb, 255:red, 219; green, 215; blue, 215 }  ,fill opacity=1 ] (61,175.83) -- (150,175.83) -- (150,234.83) -- (61,234.83) -- cycle ;
        \draw  [fill={rgb, 255:red, 207; green, 204; blue, 204 }  ,fill opacity=1 ] (183.67,175.33) -- (272.67,175.33) -- (272.67,234.33) -- (183.67,234.33) -- cycle ;
        \draw  [fill={rgb, 255:red, 201; green, 197; blue, 197 }  ,fill opacity=1 ] (304.67,175) -- (431,175) -- (431,234) -- (304.67,234) -- cycle ;
        \draw  [fill={rgb, 255:red, 218; green, 232; blue, 204 }  ,fill opacity=1 ] (61,285.67) -- (430,285.67) -- (430,325.67) -- (61,325.67) -- cycle ;
        \draw   (94,238.33) -- (101.5,234.33) -- (109,238.33) -- (104,238.33) -- (104,280.33) -- (109,280.33) -- (101.5,284.33) -- (94,280.33) -- (99,280.33) -- (99,238.33) -- cycle ;
        \draw   (216,239.33) -- (223.5,235.33) -- (231,239.33) -- (226,239.33) -- (226,281.33) -- (231,281.33) -- (223.5,285.33) -- (216,281.33) -- (221,281.33) -- (221,239.33) -- cycle ;
        \draw   (353,239.33) -- (360.5,235.33) -- (368,239.33) -- (363,239.33) -- (363,281.33) -- (368,281.33) -- (360.5,285.33) -- (353,281.33) -- (358,281.33) -- (358,239.33) -- cycle ;
        \draw    (215,125) -- (110.81,173.9) ;
        \draw [shift={(109,174.75)}, rotate = 334.86] [color={rgb, 255:red, 0; green, 0; blue, 0 }  ][line width=0.75]    (10.93,-3.29) .. controls (6.95,-1.4) and (3.31,-0.3) .. (0,0) .. controls (3.31,0.3) and (6.95,1.4) .. (10.93,3.29)   ;
        \draw    (235,126.75) -- (234.52,173.25) ;
        \draw [shift={(234.5,175.25)}, rotate = 270.59] [color={rgb, 255:red, 0; green, 0; blue, 0 }  ][line width=0.75]    (10.93,-3.29) .. controls (6.95,-1.4) and (3.31,-0.3) .. (0,0) .. controls (3.31,0.3) and (6.95,1.4) .. (10.93,3.29)   ;
        \draw    (263,124.75) -- (366.19,173.89) ;
        \draw [shift={(368,174.75)}, rotate = 205.46] [color={rgb, 255:red, 0; green, 0; blue, 0 }  ][line width=0.75]    (10.93,-3.29) .. controls (6.95,-1.4) and (3.31,-0.3) .. (0,0) .. controls (3.31,0.3) and (6.95,1.4) .. (10.93,3.29)   ;
        
        \draw (219.5,93.5) node [anchor=north west][inner sep=0.75pt]   [align=left] {\textbf{EVT}};
        \draw (72,181.33) node [anchor=north west][inner sep=0.75pt]   [align=left] { \textbf{ Order}\\\textbf{Statistics}};
        \draw (187.67,179.33) node [anchor=north west][inner sep=0.75pt]   [align=left] { \textbf{Statistical}\\\textbf{ Inference}};
        \draw (305.67,181) node [anchor=north west][inner sep=0.75pt]   [align=left] { \textbf{Communication}\\\textbf{  \ \ \ \ Theory}};
        \draw (98,294.33) node [anchor=north west][inner sep=0.75pt]   [align=left] 
        {\textbf{EVT for Communication  Engineering }};  

        \end{tikzpicture}
        \caption{Various mathematical tools across different disciplines are required to utilize \ac{EVT} for  communications engineering.}
        \label{fig:generalEVTCommTools}
    \end{figure*}
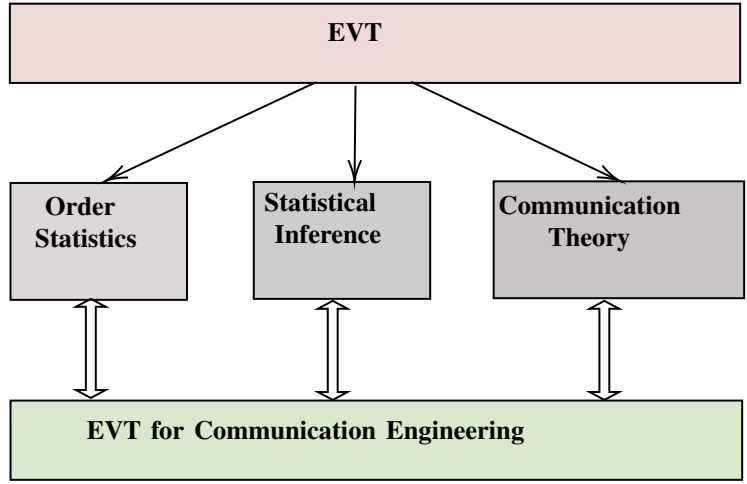

\ac{EVT} is a field of statistics dealing with large deviations from the median/central values. Such values with large deviations from the median are referred to as extreme values. Let $X_1,X_2,\cdots,X_K$ be a sequence of $K$ \ac{RVs}, where $X_{(1:K)} \leq X_{(2:K)}\leq \cdots X_{(K:K)}$ represents the corresponding ordered sequence. Here, $X_{(K-k+1:K)}$ is referred to as the $k$-th order statistic of the sequence of \ac{RVs}. It was observed that the statistical properties of the $k$-th order statistics are very different for different values of $p$ when $\frac{k}{K} \to p$ and $K \to \infty$. EVT deals with the scenario, where $k$ or {$K-k$} is fixed and  $K \to \infty$ \cite{david2004order}. One of the most widely used result from \ac{EVT} states that under an appropriate normalization, the distribution of the maximum of a sequence of \ac{i.i.d.} \ac{RVs} converges to one of three distributions, namely the Gumbel, the Frechet and the Weibull \ac{CDF}s \cite{david2004order}. Mathematically, this result can be stated as follows: Assume that there exist normalizing constants $a_K$ and $b_K$ such that 
\begin{equation}
F^K(a_Kx+b_K) \xrightarrow[K \to \infty]{}  G(x)
\label{limit_1}
\end{equation} for any point of continuity $x$ of $G(.)$, and let $X^K_{max} := \max \{X_1,\cdots,X_K\}$ for $X_i \sim F(.)$. Note that here, $F^K(.)$ is the CDF of the maximum of $K$ \ac{i.i.d.} RVS along with CDF $F(.)$. Then $G(.)$ is up to a location and scale shift an extreme value distribution. The class of distributions
$F$ satisfying (\ref{limit_1}) is termed as the maximum domain of attraction or simply the domain of attraction of $G$. \ac{EVDs} were mathematically characterized by Fisher, Tippet and Gumbel and one of the most widely used result is stated below \cite[Theorem 1.1.3]{de2006extreme}.

\begin{theorem}
The class of EVDs is $G_{\beta}\left(ax+b\right)$ with $a>0$ and $b \in \mathbb{R}$, where
\begin{equation}
G_{\beta}(x) = \exp\left( -\left(1+\beta x \right)^{\frac{-1}{\beta}}\right); \ \ 1+\beta x >0
\end{equation}with $\beta \in \mathbb
R$ and for $\beta=0$, the right-hand side is interpreted as $\exp\left[-\exp\left(-x\right)\right]$. 
\end{theorem}
The parameter $\beta$ is the extreme value index. Hence, the EVD, if it exists, forms a simple single-parameter family, apart from the scale and location parameters. Furthermore, after some algebraic manipulations of $G_\beta(x)$, the sub-cases for $\beta>0$, $\beta<0$ and $\beta=0$ can be identified to be the distribution functions of the Frechet, Weibull and Gumbel \ac{RVs} respectively. These distribution functions are given by,  
\begin{align}
\text{Gumbel}: & \quad    G(x)=\exp \left\{-\exp \left(-\left(\frac{x-b}{a}\right)\right)\right\},      \\
\text{Weibull}: &  \quad G(x)= \begin{cases}\exp \left\{-\left(-\left(\frac{x-b}{a}\right)\right)^{\frac{1}{\beta}}\right\}, & x<b, \\ 1, & x \geq b,\end{cases}      \\
\text{Frechet}: & \quad   G(x)= \begin{cases}0, & x \leq b, \\ \exp \left\{-\left(\frac{x-b}{a}\right)^{-\frac{1}{\beta}}\right\}, & x>b,\end{cases}    
\label{gev}
\end{align}  
for $\beta>0$. Fig. \ref{pdf_fig_1} shows the PDF of the EVD for the three choices of $\beta$ discussed above. 
\begin{figure}
        \centering
        \includegraphics[scale=0.6]{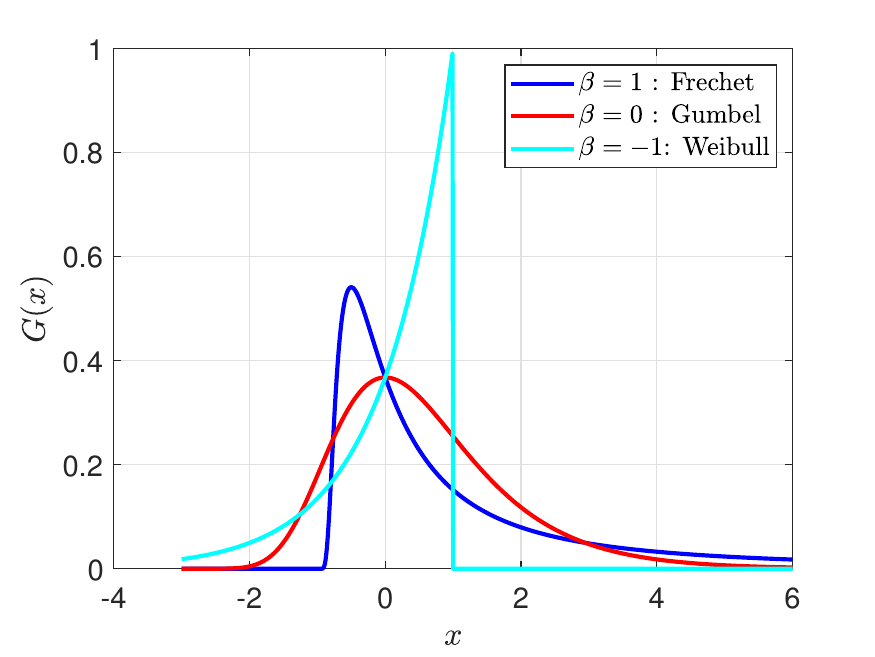}
    \caption{PDF of EVDs for $a=0$, $b=1$.}
    \label{pdf_fig_1}
    \end{figure}

Results from \ac{EVT} state that regardless of the parent distribution, the tail of a large class of distributions behaves alike under appropriate normalization. Recall that the characterization of the exact distribution of the extremes (even the maximum or the minimum) of a sequence of \ac{RVs} is not a trivial task for most of the commonly encountered distributions. For example, the \ac{CDF} of the maximum and the minimum of a sequence of $K$ \ac{RVs} are characterized in terms of the product of the \ac{CDF} of the respective \ac{RVs}, i.e., 
\begin{align}
& \mathbb{P} \left(X_{(K:K)} \leq x \right) = \left( F(x)\right)^K, \ \ \text{and} \\ & 1-\mathbb{P} \left(X_{(1:K))} \leq x \right) = \left( 1- F(x)\right)^K.
\end{align}
Such a product of $K$ functions will not result in a very simple function, especially for large values of $K$. Using results from \ac{EVT}, we can characterize the order statistics for many common distribution functions as one of the three \ac{EVDs} discussed above. There has been much interest in identifying the conditions under which such a normalization exists and in identifying the corresponding normalization constants \cite{falk2010laws}. 

\textcolor{black}{Let us consider a simple example, where we examine the limiting distributions of the maximum and minimum of $K$ \ac{i.i.d.} exponential \ac{RVs}, each with the \ac{CDF} $F\left( x\right)=1-\exp\left( -\frac{x}{a} \right)$. It can be observed that the maximum of $K$ \ac{i.i.d.} exponential \ac{RVs} follows a Gumbel distribution, with normalizing constants $a_K=a\log\left( K \right)$ and $b_K=a$. Fig. \ref{fig:image2} (a) and (b) display the simulated and theoretical limiting \ac{CDF} and \ac{PDF} for the maximum of exponential \ac{RVs}.
Similarly, the minimum of $K$ \ac{i.i.d.} exponential \ac{RVs} follows a Weibull distribution, with normalizing constants $a_K=-a\log\left( 1-\frac{1}{K} \right)$ and $b_K=0$. The corresponding results are presented in Fig. \ref{fig:image2} (c) and (d).}
\begin{figure*}
    \begin{subfigure}{0.5\textwidth}
        \includegraphics[width=0.9\linewidth, height=5cm]{{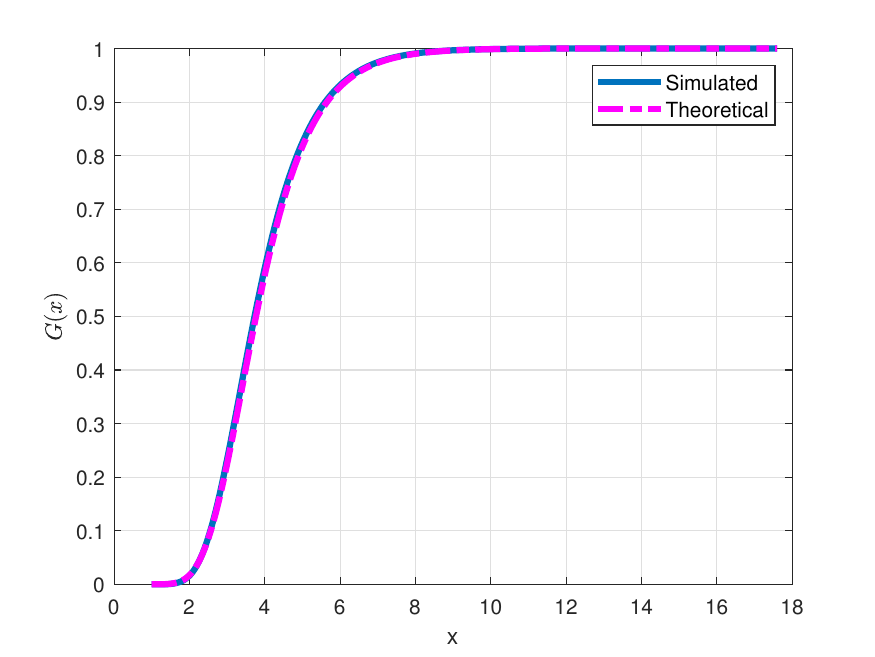}} 
        \caption{Gumbel CDF}
        \label{fig:zero}
    \end{subfigure}
    \begin{subfigure}{0.5\textwidth}
        \includegraphics[width=0.9\linewidth, height=5cm]{{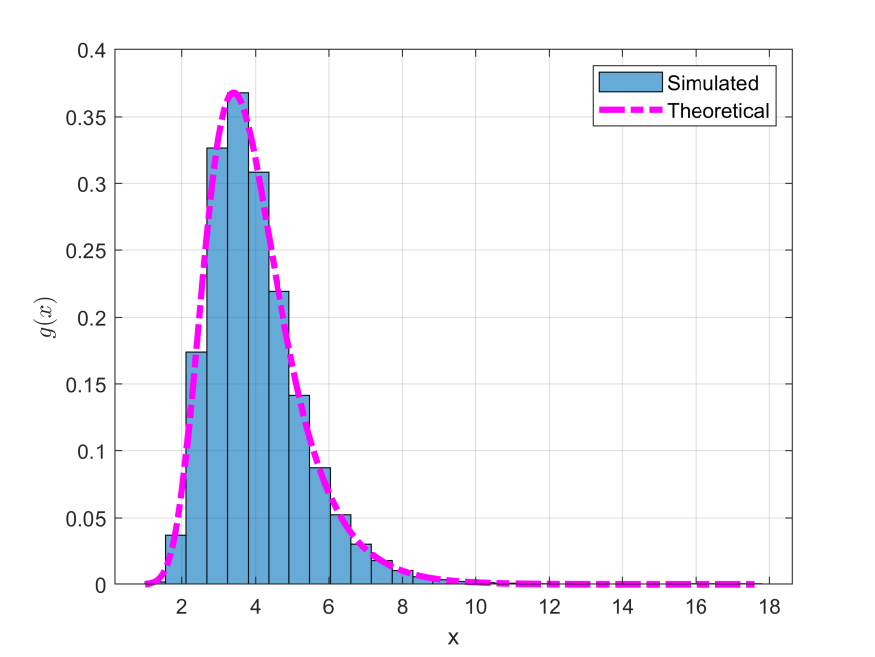}}
        \caption{Gumbel PDF}
        \label{fig:one}
    \end{subfigure}
    \begin{subfigure}{0.5\textwidth}
        \includegraphics[width=0.9\linewidth, height=5cm]{{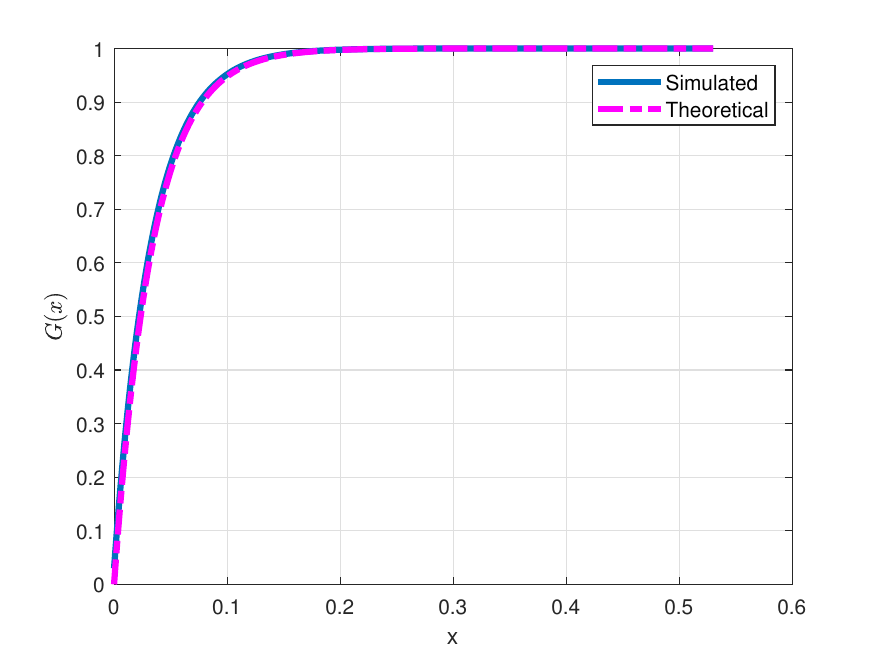}} 
        \caption{Weibull CDF}
        \label{fig:two}
    \end{subfigure}
    \begin{subfigure}{0.5\textwidth}
        \includegraphics[width=0.9\linewidth, height=5cm]{{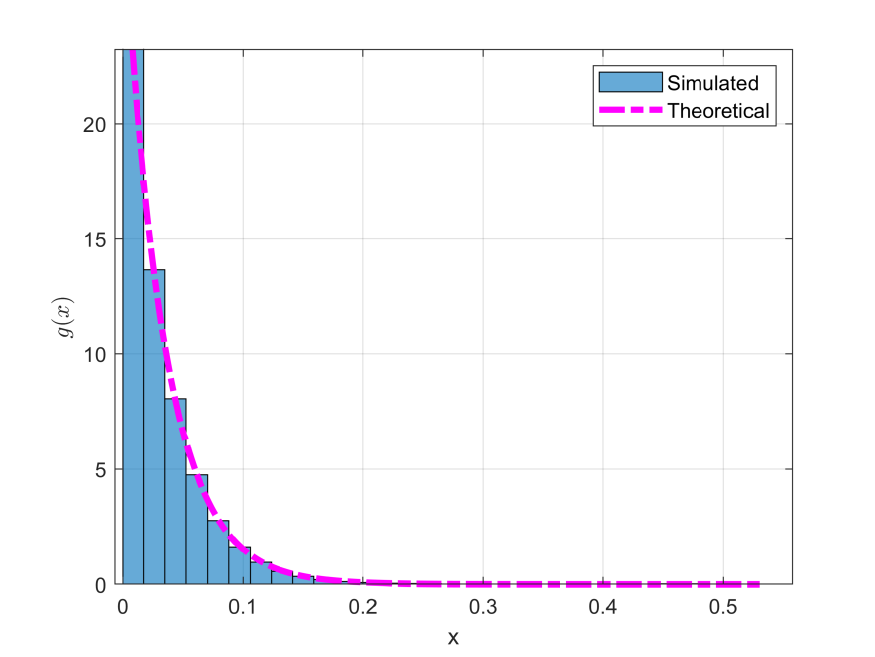}} 
        \caption{Weibull PDF}
        \label{fig:four}
    \end{subfigure}
\caption{\textcolor{black}{llustration of the limiting distributions of the maximum and minimum for an exponential random variable. (a) shows the Gumbel distribution as the limiting distribution for the maximum of $ K$ \ac{i.i.d.} exponential \ac{RVs}, while (b) presents the corresponding \ac{PDF}. (c) displays the Weibull distribution as the limiting distribution for the minimum of $ K$ \ac{i.i.d.} exponential \ac{RVs}, and  (d) shows the corresponding \ac{PDF}.}}
\label{fig:image2}
\end{figure*}

The case of extreme statistics over sequences of \ac{i.i.d.}~\ac{RVs} are very well studied, and we have rules to identify the limiting distribution and the corresponding normalization constants for most of the practical cases. The case of \ac{i.n.i.d.} \ac{RVs} is relatively new in the communication literature, and only a few contributions discuss the limiting distribution of the extremes in such scenarios \cite{subhash2021cooperative,subhash2022asymptotic}. The following sub-sections discuss the key results for the cases of  \ac{i.i.d.} and \ac{i.n.i.d.} \ac{RVs}.

\subsection{\ac{EVT} for \ac{i.i.d.} \ac{RVs}}
Leadbetter {\it et al.} \cite{leadbetter1983exceedances} gave a comprehensive account of the necessary and sufficient conditions for a distribution function $F$ to belong to the domain of attraction $G$ and the characterizations of $a_{K}$ and $b_{K}$ when $G$ is one of the three \ac{EVDs}. The conditions for establishing the limiting distributions and choices of normalization constants that can achieve the asymptotics are discussed in detail in textbooks like \cite{de2006extreme} and \cite{falk2010laws}. Given the expression for the \ac{CDF} $F$ of each of the RV in the sequence, the conditions to identify the domain of attraction can be verified using basic algebra. The mathematical complexities to establish these limits depend upon the distribution function $F$. Once the domain of attraction is established, one can find appropriate choices for $a_K$ and $b_K$ that satisfy the limit in (\ref{limit_1}) for the corresponding distribution. Note that the choices of normalization constants need not be unique, and the expressions in \cite{de2006extreme } and \cite{falk2010laws} serve as a convenient choice to study the asymptotics for any distribution function belonging to the domain of attraction in an EVD. There are also results discussing the convergence of the moments of the asymptotic order statistics \cite{de2007extreme} to the moments of the EVDs.

\subsection{\ac{EVT} for \ac{i.n.i.d.} \ac{RVs}}
In contrast to the case of \ac{i.i.d.} \ac{RVs}, the limiting distribution of the asymptotic order statistics over sequences of \ac{i.n.i.d.} \ac{RVs} is not always one of the three \ac{EVDs} discussed in the above section. Furthermore, we do not have any known general expression for the normalization constants, unlike in the case of \ac{i.i.d.} \ac{RVs}. Mejzler \cite{mejzler1969some} published the first study discussing the order statistics of sequences of \ac{i.n.i.d.} \ac{RVs}, 40 years after the first paper on the  
\ac{i.i.d.} \ac{RVs} by Fisher and Tippett in the year 1928. Many years later, Barakat \textit{et al}. in \cite{barakat2013limit} extended this result to the case of the random maximum of \ac{i.n.i.d.} random vectors, where the sample size is also assumed to be a random quantity. The results in \cite{barakat2013limit} are general and can be used for characterizing the maximum over independent sequences of \ac{RVs} encountered in communication systems.

\subsection{Peak over Threshold Using \ac{EVT}}
Characterizing the statistics of the peaks over a particular threshold is significant in different engineering arenas (characterizing extreme floods, the maximum load on structures, peak network traffic, financial risk studies, etc.). It is very challenging to mathematically characterize the exceedance over a threshold, when this event is rare. However, if the probability of such an event is non-zero, it is very important to consider such scenarios in system planning (to avoid the risks of the extreme event.). There are results in \ac{EVT} that study the statistics of such rare events that can be expressed as the exceedance of a threshold \cite{falk2010laws}. Here, we have results that hold asymptotically (in many practical scenarios, even in the non-asymptotic regime) and allow the characterization of these extreme events in terms of simple and tractable statistics. The key result from \ac{EVT} used for characterizing the peak over the threshold is given below:
\begin{theorem}\label{Theorem2}
   Let $\left(X_1, X_2, \ldots, X_K\right)$ be a sequence of \ac{i.i.d.} \ac{RVs}, and let $F_u$ be their conditional excess distribution function. Then, for a large class of underlying distribution functions $F$, and large $u, F_u$ is well approximated by the \ac{GPD}. That is: $F_u(x) \rightarrow G_{\beta, \delta}(x)$ as $u \rightarrow \infty$, where
$$
G_{\beta, \delta}(x)= \begin{cases}1-e^{-\beta x}, & \text { if } \gamma=0, x>0, \\ 1-(1+\gamma \beta x)^{-1 / \gamma}, & \text { if } \gamma \neq 0, x>0.\end{cases}$$ 
\end{theorem}
\begin{proof}
    Please see \cite{pickands1975statistical,balkema1974residual}.
\end{proof}
\begin{figure*} 
    \centering
   \tikzset{every picture/.style={line width=0.75pt}} 

        \begin{tikzpicture}[x=0.75pt,y=0.75pt,yscale=-1,xscale=1]
        
        \draw  [color={rgb, 255:red, 0; green, 0; blue, 0 }  ,draw opacity=1 ][fill={rgb, 255:red, 237; green, 241; blue, 226 }  ,fill opacity=1 ] (5.76,26) -- (636.49,26) -- (636.49,102.34) -- (5.76,102.34) -- cycle ; \draw  [color={rgb, 255:red, 0; green, 0; blue, 0 }  ,draw opacity=1 ] (17.75,38) -- (624.49,38) -- (624.49,90.34) -- (17.75,90.34) -- cycle ; \draw  [color={rgb, 255:red, 0; green, 0; blue, 0 }  ,draw opacity=1 ] (5.76,26) -- (17.75,38) ; \draw  [color={rgb, 255:red, 0; green, 0; blue, 0 }  ,draw opacity=1 ] (636.49,26) -- (624.49,38) ; \draw  [color={rgb, 255:red, 0; green, 0; blue, 0 }  ,draw opacity=1 ] (636.49,102.34) -- (624.49,90.34) ; \draw  [color={rgb, 255:red, 0; green, 0; blue, 0 }  ,draw opacity=1 ] (5.76,102.34) -- (17.75,90.34) ;
        \draw  [fill={rgb, 255:red, 213; green, 225; blue, 238 }  ,fill opacity=1 ] (4,334.14) .. controls (4,324.31) and (11.97,316.34) .. (21.8,316.34) -- (75.21,316.34) .. controls (85.05,316.34) and (93.02,324.31) .. (93.02,334.14) -- (93.02,447) .. controls (93.02,447) and (93.02,447) .. (93.02,447) -- (4,447) .. controls (4,447) and (4,447) .. (4,447) -- cycle ;
        \draw  [fill={rgb, 255:red, 248; green, 212; blue, 217 }  ,fill opacity=1 ] (114.1,334.62) .. controls (114.1,324.66) and (122.18,316.59) .. (132.14,316.59) -- (186.25,316.59) .. controls (196.21,316.59) and (204.29,324.66) .. (204.29,334.62) -- (204.29,446.5) .. controls (204.29,446.5) and (204.29,446.5) .. (204.29,446.5) -- (114.1,446.5) .. controls (114.1,446.5) and (114.1,446.5) .. (114.1,446.5) -- cycle ;
        \draw  [fill={rgb, 255:red, 241; green, 239; blue, 205 }  ,fill opacity=1 ] (218.34,334.84) .. controls (218.34,324.62) and (226.63,316.34) .. (236.85,316.34) -- (292.37,316.34) .. controls (302.59,316.34) and (310.87,324.62) .. (310.87,334.84) -- (310.87,446.5) .. controls (310.87,446.5) and (310.87,446.5) .. (310.87,446.5) -- (218.34,446.5) .. controls (218.34,446.5) and (218.34,446.5) .. (218.34,446.5) -- cycle ;
        \draw  [fill={rgb, 255:red, 239; green, 220; blue, 243 }  ,fill opacity=1 ] (327.27,335.62) .. controls (327.27,325.66) and (335.35,317.59) .. (345.31,317.59) -- (399.42,317.59) .. controls (409.38,317.59) and (417.46,325.66) .. (417.46,335.62) -- (417.46,446) .. controls (417.46,446) and (417.46,446) .. (417.46,446) -- (327.27,446) .. controls (327.27,446) and (327.27,446) .. (327.27,446) -- cycle ;
        \draw  [fill={rgb, 255:red, 217; green, 235; blue, 198 }  ,fill opacity=1 ] (436.2,335.08) .. controls (436.2,324.73) and (444.59,316.34) .. (454.94,316.34) -- (511.16,316.34) .. controls (521.51,316.34) and (529.9,324.73) .. (529.9,335.08) -- (529.9,447) .. controls (529.9,447) and (529.9,447) .. (529.9,447) -- (436.2,447) .. controls (436.2,447) and (436.2,447) .. (436.2,447) -- cycle ;
        \draw  [fill={rgb, 255:red, 203; green, 241; blue, 233 }  ,fill opacity=1 ] (546.3,335.08) .. controls (546.3,324.73) and (554.69,316.34) .. (565.04,316.34) -- (621.26,316.34) .. controls (631.61,316.34) and (640,324.73) .. (640,335.08) -- (640,446) .. controls (640,446) and (640,446) .. (640,446) -- (546.3,446) .. controls (546.3,446) and (546.3,446) .. (546.3,446) -- cycle ;
        \draw   (34.45,147.6) -- (42.65,147.6) -- (42.65,102.34) -- (54.36,102.34) -- (54.36,147.6) -- (62.56,147.6) -- (48.51,163.66) -- cycle ;
        \draw   (144.55,147.6) -- (152.75,147.6) -- (152.75,102.34) -- (164.46,102.34) -- (164.46,147.6) -- (172.66,147.6) -- (158.61,163.66) -- cycle ;
        \draw   (249.97,148.85) -- (258.17,148.85) -- (258.17,103.59) -- (269.88,103.59) -- (269.88,148.85) -- (278.08,148.85) -- (264.02,164.91) -- cycle ;
        \draw   (357.72,148.85) -- (365.92,148.85) -- (365.92,103.59) -- (377.64,103.59) -- (377.64,148.85) -- (385.83,148.85) -- (371.78,164.91) -- cycle ;
        \draw   (467.82,148.85) -- (476.02,148.85) -- (476.02,103.59) -- (487.73,103.59) -- (487.73,148.85) -- (495.93,148.85) -- (481.88,164.91) -- cycle ;
        \draw   (576.75,148.85) -- (584.95,148.85) -- (584.95,103.59) -- (596.66,103.59) -- (596.66,148.85) -- (604.86,148.85) -- (590.81,164.91) -- cycle ;
        \draw  [fill={rgb, 255:red, 195; green, 216; blue, 240 }  ,fill opacity=1 ] (6.2,208.87) .. controls (6.2,183.9) and (25.14,163.66) .. (48.51,163.66) .. controls (71.88,163.66) and (90.82,183.9) .. (90.82,208.87) .. controls (90.82,233.84) and (71.88,254.08) .. (48.51,254.08) .. controls (25.14,254.08) and (6.2,233.84) .. (6.2,208.87) -- cycle ;
        \draw  [fill={rgb, 255:red, 236; green, 187; blue, 193 }  ,fill opacity=1 ] (116.3,208.87) .. controls (116.3,183.9) and (135.24,163.66) .. (158.61,163.66) .. controls (181.98,163.66) and (200.92,183.9) .. (200.92,208.87) .. controls (200.92,233.84) and (181.98,254.08) .. (158.61,254.08) .. controls (135.24,254.08) and (116.3,233.84) .. (116.3,208.87) -- cycle ;
        \draw  [fill={rgb, 255:red, 221; green, 215; blue, 149 }  ,fill opacity=1 ] (221.71,210.12) .. controls (221.71,185.15) and (240.65,164.91) .. (264.02,164.91) .. controls (287.39,164.91) and (306.33,185.15) .. (306.33,210.12) .. controls (306.33,235.09) and (287.39,255.33) .. (264.02,255.33) .. controls (240.65,255.33) and (221.71,235.09) .. (221.71,210.12) -- cycle ;
        \draw  [fill={rgb, 255:red, 219; green, 156; blue, 230 }  ,fill opacity=1 ] (329.47,210.12) .. controls (329.47,185.15) and (348.41,164.91) .. (371.78,164.91) .. controls (395.15,164.91) and (414.09,185.15) .. (414.09,210.12) .. controls (414.09,235.09) and (395.15,255.33) .. (371.78,255.33) .. controls (348.41,255.33) and (329.47,235.09) .. (329.47,210.12) -- cycle ;
        \draw  [fill={rgb, 255:red, 155; green, 187; blue, 122 }  ,fill opacity=1 ] (439.57,210.12) .. controls (439.57,185.15) and (458.51,164.91) .. (481.88,164.91) .. controls (505.25,164.91) and (524.19,185.15) .. (524.19,210.12) .. controls (524.19,235.09) and (505.25,255.33) .. (481.88,255.33) .. controls (458.51,255.33) and (439.57,235.09) .. (439.57,210.12) -- cycle ;
        \draw  [fill={rgb, 255:red, 137; green, 218; blue, 199 }  ,fill opacity=1 ] (548.49,210.12) .. controls (548.49,185.15) and (567.44,164.91) .. (590.81,164.91) .. controls (614.17,164.91) and (633.12,185.15) .. (633.12,210.12) .. controls (633.12,235.09) and (614.17,255.33) .. (590.81,255.33) .. controls (567.44,255.33) and (548.49,235.09) .. (548.49,210.12) -- cycle ;
        \draw   (34.45,299.34) -- (42.65,299.34) -- (42.65,254.08) -- (54.36,254.08) -- (54.36,299.34) -- (62.56,299.34) -- (48.51,315.4) -- cycle ;
        \draw   (144.55,299.34) -- (152.75,299.34) -- (152.75,254.08) -- (164.46,254.08) -- (164.46,299.34) -- (172.66,299.34) -- (158.61,315.4) -- cycle ;
        \draw   (249.97,300.59) -- (258.17,300.59) -- (258.17,255.33) -- (269.88,255.33) -- (269.88,300.59) -- (278.08,300.59) -- (264.02,316.65) -- cycle ;
        \draw   (357.72,300.59) -- (365.92,300.59) -- (365.92,255.33) -- (377.64,255.33) -- (377.64,300.59) -- (385.83,300.59) -- (371.78,316.65) -- cycle ;
        \draw   (467.82,300.59) -- (476.02,300.59) -- (476.02,255.33) -- (487.73,255.33) -- (487.73,300.59) -- (495.93,300.59) -- (481.88,316.65) -- cycle ;
        \draw   (576.75,300.59) -- (584.95,300.59) -- (584.95,255.33) -- (596.66,255.33) -- (596.66,300.59) -- (604.86,300.59) -- (590.81,316.65) -- cycle ;
        
        \draw (320.24,64.17) node  [font=\normalsize] [align=left] {\textbf{{\Large Order statistics/\ac{EVT} for NG}}};
        \draw (15.94,183.96) node [anchor=north west][inner sep=0.75pt]   [align=left] {
        {\footnotesize \textbf{Multi-User}}\\{\footnotesize \textbf{Diversity}}};
        \draw (130.8,188.21) node [anchor=north west][inner sep=0.75pt]   [align=left] {{\footnotesize \textbf{Spectrum}}\\{\footnotesize \textbf{Sharing}}};
        \draw (241.19,189.21) node [anchor=north west][inner sep=0.75pt]   [align=left] {{\footnotesize \textbf{Relays/}}\\{\footnotesize \textbf{RIS}}};
        \draw (350.52,195.33) node [anchor=north west][inner sep=0.75pt]   [align=left] {{\footnotesize \textbf{MIMO}}};
        \draw (457.79,194.08) node [anchor=north west][inner sep=0.75pt]   [align=left] {{\footnotesize \textbf{URLLC}}};
        \draw (562.89,186.71) node [anchor=north west][inner sep=0.75pt]   [align=left] {{\footnotesize \textbf{Machine}}\\{\footnotesize \textbf{Learning}}};
        \draw (4,334.14) node [anchor=north west][inner sep=0.75pt]   [align=left]{\vspace{0.8mm}\textbf{{\scriptsize 1.Optical Comm. }}\\\textbf{{\scriptsize 2.Opportunistic }}\\\textbf{{\scriptsize  \ \ beamforming}}\\\textbf{{\scriptsize 3.Cognitive }}\textbf{{\scriptsize  Radio}}};
        \draw (113,331.04) node [anchor=north west][inner sep=0.75pt]   [align=left] {\textbf{{\scriptsize 1.Relaying}}\\\textbf{{\scriptsize 2.Antenna }}\textbf{{\scriptsize   selection}}\\\textbf{{\scriptsize 3.Multicasting}}\\\textbf{{\scriptsize 4.Cognitive  Radio}}};
        \draw (555.3,335.08) node [anchor=north west][inner sep=0.75pt]   [align=left] {\textbf{{\scriptsize 1.SVM}}\\\textbf{{\scriptsize 2.Auto encoder}}\\\textbf{{\scriptsize 3. EVM}}\\\textbf{{\scriptsize 4.MAB }}};
        \draw (327.27,335.62) node [anchor=north west][inner sep=0.75pt]   [align=left] {\textbf{{\scriptsize 1.MISO-RBF}}\\\textbf{{\scriptsize 2.MIMO-SDMA}}\\\textbf{{\scriptsize 3.Multicasting}}\\{\scriptsize \textbf{4.Massive-MIMO}}};
        \draw (434,334.04) node [anchor=north west][inner sep=0.75pt]   [align=left] {\textbf{{\scriptsize 1.V2V-Network}}\\\textbf{{\scriptsize 2.Edge computing }}\\\textbf{{\scriptsize 3.AoI}}\\{\scriptsize \textbf{4.Federated learning}}};
        \draw (220,334.04) node [anchor=north west][inner sep=0.75pt]   [align=left] {\textbf{{\scriptsize 1.WPS}}\\\textbf{{\scriptsize 2.Opportunistic }}\\\textbf{{\scriptsize  \ \ scheduling}}\\\textbf{{\scriptsize 3.mmWave net}}};
                
        \end{tikzpicture}
    \caption{Order statistics for NG}
    \label{fig:order6g}
\end{figure*}
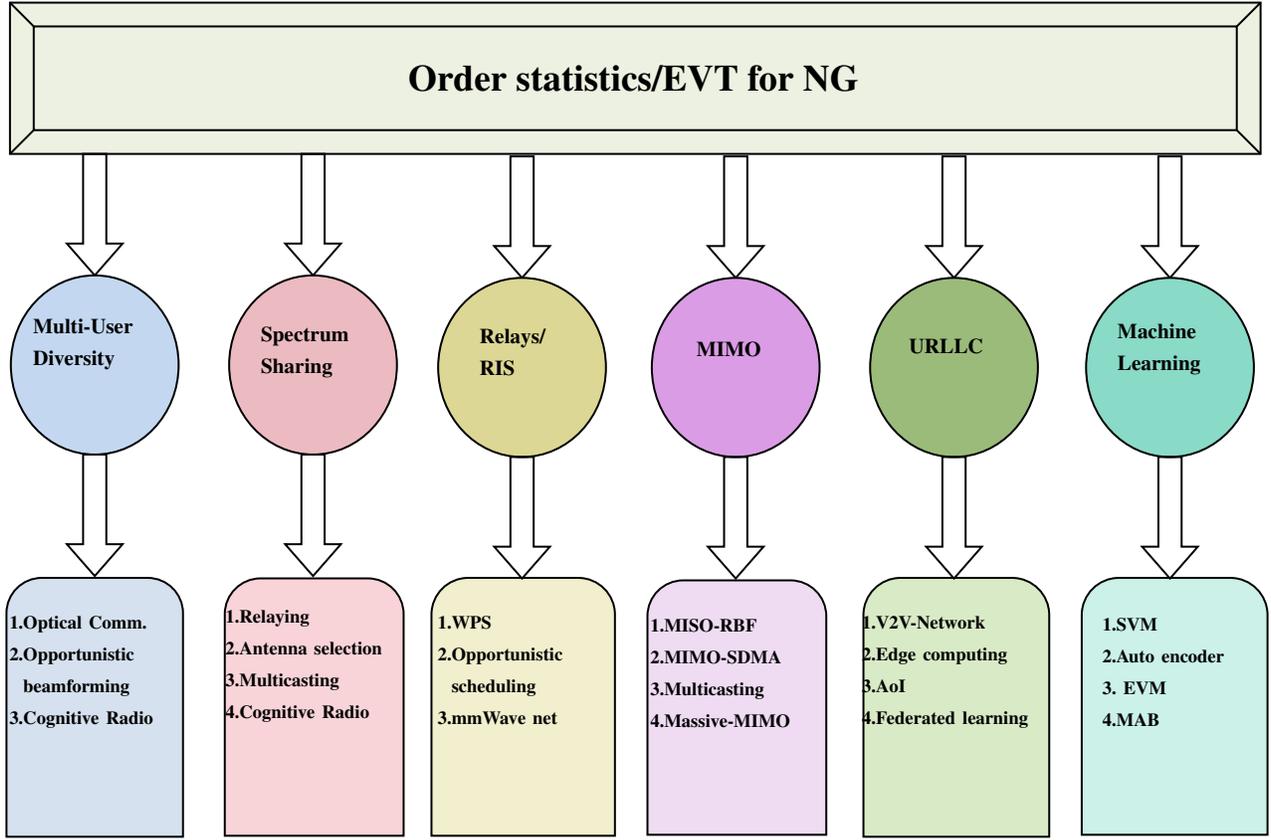
According to this result, if there exists a scaling function $g>0$ such that the conditional distribution of scaled excess value, i.e., $\left\{\left(X-u_{o p t}\right) / g\left(u_{o p t}\right) \mid X>u_{o p t}\right\}$ converges to a non-degenerate distribution, then that distribution is necessarily the \ac{GPD} \cite{kumar2017modeling}, i.e.,
$$
\lim _{u_{o p t} \rightarrow x_f} \mathbb{P}\left(\frac{X-u_{o p t}}{g\left(u_{o p t}\right)} \leq x \mid X>u_{o p t}\right)=G_{\beta, \delta}(x).
$$
Note that the accuracy of the above distribution depends upon the threshold, and hence we have to be careful while using the result to arrive at accurate inferences. In most practical cases, we do not know $g$, and the expertise in the problem domain and sample observations may be used to arrive at the parameters of the \ac{GPD} to model the statistics of interest. For example, $\beta>0$ leads to a bounded tail distribution, and hence for any signal with a bounded tail, we can safely assume that the exceedance can be modeled with $\beta$ positive. More details regarding applying these results can be found in \cite{pickands1975statistical,balkema1974residual}. 

\subsection{Rate of Convergence}
Note that, (\ref{limit_1}) guarantees the convergence of the distribution of $X^{K}_{max}$ to the distribution $G(x)$, but does not discuss the rate of convergence. In other words, it does not discuss how rapidly the distribution of $X^{K}_{max}$ converges to $G(x)$. The rate of convergence is not the same for all distributions in any domain of attraction. It is a function of the initial distribution parameters and depends on the equivalence of the tail of the initial distribution function to the tail of a \ac{GPD} \cite{de2007extreme}. The closer the tail behaviour of the initial distribution to the tail behaviour of a GPD, the faster its rate of convergence. A detailed discussion on the rate of convergence at the point of maximum possible deviation over the entire support of the maximum distribution of sequences of i.i.d. RVs is available in \cite{de2007extreme}. In many practical applications, characterizing the rate of convergence in terms of the parameters of the distributions of the RVs helps us arrive at meaningful inferences regarding the expected rate of convergence and the accuracy of the results for different choices of distribution parameters. Such inferences are particularly useful to appreciate the utility of the results derived using the asymptotic assumptions. 

\section{\ac{EVT} for \ac{NG} Applications}
In this section, we explore the utility of applying \ac{EVT} to the diverse applications found in wireless communications, including multi-user diversity, spectrum sharing, relays, \ac{MIMO} systems, \ac{URLLC}, and \ac{ML}, as depicted in Fig.~\ref{fig:order6g}. For each application, we present a unified structure which comprises a general system description, the advantages of applying \ac{EVT}, the illustrative applicability of \ac{EVT} in the given context, and a state-of-the-art literature review.

In the system description, we offer a brief introduction of the communication scenario to which \ac{EVT} is applied. Subsequently, we highlight the advantages and benefits of employing \ac{EVT} in the considered application, illustrating how \ac{EVT} can simplify the characterization of the system performance and hence help in addressing the unique challenges in the given context. Additionally, we provide practical examples or use cases to demonstrate the applicability and effectiveness of \ac{EVT} in these applications. Lastly, we review the state-of-the-art to showcase the relevant advances and research trends.

\subsection{Multi-User Diversity}
Diversity is widely used in wireless communications for improving link performance. The core concept is having multiple replicas of the same signal at the receiver, which experience different channel conditions. Various diversity techniques are available for combating multipath fading in wireless communications. Among those, the multi-user diversity technique selects one of the best channels among all the available channels for communication.
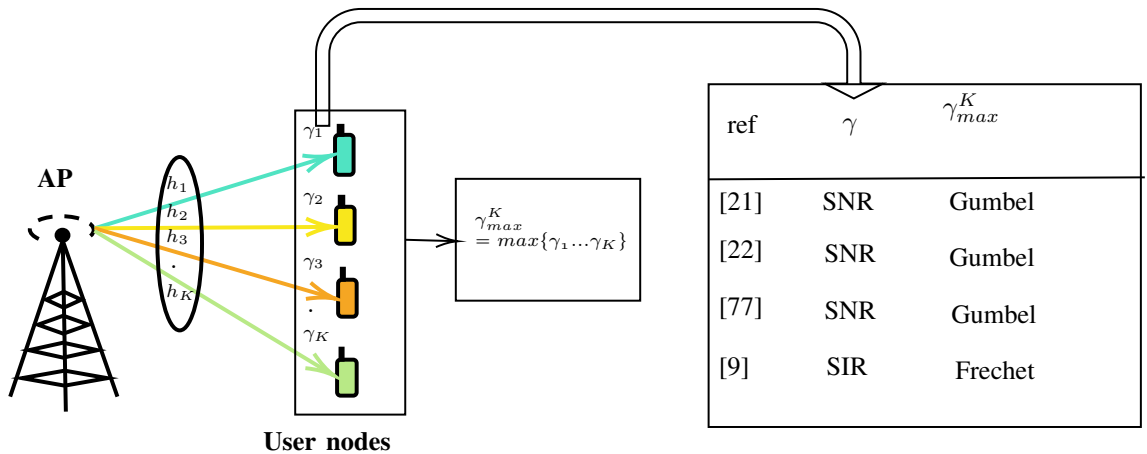
\begin{figure*}[b]
\centering
        \tikzset{every picture/.style={line width=0.75pt}} 
        
        \begin{tikzpicture}[x=0.75pt,y=0.75pt,yscale=-1,xscale=1]
        
        \draw [color={rgb, 255:red, 0; green, 0; blue, 0 }  ,draw opacity=1 ][line width=1.5]    (103.27,131) -- (76.75,209.68) ;
        \draw [color={rgb, 255:red, 0; green, 0; blue, 0 }  ,draw opacity=1 ][line width=1.5]    (103.27,131) -- (130.84,209.68) ;
        \draw [color={rgb, 255:red, 0; green, 0; blue, 0 }  ,draw opacity=1 ][line width=1.5]    (103.27,131) -- (104.83,217.25) ;
        \draw  [color={rgb, 255:red, 0; green, 0; blue, 0 }  ,draw opacity=1 ][line width=1.5]  (103.79,156.73) -- (112.64,160.76) -- (103.79,164.8) -- (94.95,160.76) -- cycle ;
        \draw  [color={rgb, 255:red, 0; green, 0; blue, 0 }  ,draw opacity=1 ][line width=1.5]  (103.53,168.83) -- (116.28,173.37) -- (103.53,177.91) -- (90.79,173.37) -- cycle ;
        \draw  [color={rgb, 255:red, 0; green, 0; blue, 0 }  ,draw opacity=1 ][line width=1.5]  (103.5,182.45) -- (120.75,186.23) -- (103.5,190.01) -- (86.25,186.23) -- cycle ;
        \draw  [color={rgb, 255:red, 0; green, 0; blue, 0 }  ,draw opacity=1 ][line width=1.5]  (104.25,196.57) -- (125.75,200.35) -- (104.25,204.14) -- (82.75,200.35) -- cycle ;
        \draw  [draw opacity=0][dash pattern={on 5.63pt off 4.5pt}][line width=1.5]  (91.18,130.09) .. controls (88.36,128.82) and (86.63,127.11) .. (86.65,125.25) .. controls (86.7,121.47) and (93.96,118.5) .. (102.86,118.61) .. controls (111.76,118.72) and (118.93,121.87) .. (118.88,125.65) .. controls (118.85,127.7) and (116.7,129.52) .. (113.31,130.73) -- (102.77,125.45) -- cycle ; \draw  [color={rgb, 255:red, 0; green, 0; blue, 0 }  ,draw opacity=1 ][dash pattern={on 5.63pt off 4.5pt}][line width=1.5]  (91.18,130.09) .. controls (88.36,128.82) and (86.63,127.11) .. (86.65,125.25) .. controls (86.7,121.47) and (93.96,118.5) .. (102.86,118.61) .. controls (111.76,118.72) and (118.93,121.87) .. (118.88,125.65) .. controls (118.85,127.7) and (116.7,129.52) .. (113.31,130.73) ;  
        \draw  [color={rgb, 255:red, 0; green, 0; blue, 0 }  ,draw opacity=1 ][fill={rgb, 255:red, 0; green, 0; blue, 0 }  ,fill opacity=1 ][line width=1.5]  (100.15,128.23) .. controls (100.15,126.7) and (101.49,125.45) .. (103.14,125.45) .. controls (104.8,125.45) and (106.14,126.7) .. (106.14,128.23) .. controls (106.14,129.76) and (104.8,131) .. (103.14,131) .. controls (101.49,131) and (100.15,129.76) .. (100.15,128.23) -- cycle ;
        
        \draw  [fill={rgb, 255:red, 80; green, 227; blue, 194 }  ,fill opacity=1 ][line width=1.5]  (240.58,79.41) .. controls (240.58,78.38) and (241.42,77.54) .. (242.45,77.54) -- (248.07,77.54) .. controls (249.1,77.54) and (249.94,78.38) .. (249.94,79.41) -- (249.94,95.84) .. controls (249.94,96.88) and (249.1,97.71) .. (248.07,97.71) -- (242.45,97.71) .. controls (241.42,97.71) and (240.58,96.88) .. (240.58,95.84) -- cycle ;
        \draw  [fill={rgb, 255:red, 0; green, 0; blue, 0 }  ,fill opacity=1 ][line width=1.5]  (242.45,73) -- (244.32,73) -- (244.32,77.54) -- (242.45,77.54) -- cycle ;
        
        \draw  [fill={rgb, 255:red, 248; green, 231; blue, 28 }  ,fill opacity=1 ][line width=1.5]  (241.1,115.73) .. controls (241.1,114.69) and (241.94,113.85) .. (242.97,113.85) -- (248.59,113.85) .. controls (249.62,113.85) and (250.46,114.69) .. (250.46,115.73) -- (250.46,130.64) .. controls (250.46,131.68) and (249.62,132.52) .. (248.59,132.52) -- (242.97,132.52) .. controls (241.94,132.52) and (241.1,131.68) .. (241.1,130.64) -- cycle ;
        \draw  [fill={rgb, 255:red, 0; green, 0; blue, 0 }  ,fill opacity=1 ][line width=1.5]  (242.97,108.81) -- (244.74,108.81) -- (244.74,113.85) -- (242.97,113.85) -- cycle ;
        
        \draw  [fill={rgb, 255:red, 245; green, 166; blue, 35 }  ,fill opacity=1 ][line width=1.5]  (242.14,152.55) .. controls (242.14,151.51) and (242.98,150.67) .. (244.01,150.67) -- (249.63,150.67) .. controls (250.66,150.67) and (251.5,151.51) .. (251.5,152.55) -- (251.5,167.46) .. controls (251.5,168.5) and (250.66,169.33) .. (249.63,169.33) -- (244.01,169.33) .. controls (242.98,169.33) and (242.14,168.5) .. (242.14,167.46) -- cycle ;
        \draw  [fill={rgb, 255:red, 0; green, 0; blue, 0 }  ,fill opacity=1 ][line width=1.5]  (244.01,145.13) -- (245.05,145.13) -- (245.05,150.67) -- (244.01,150.67) -- cycle ;
        
        \draw  [fill={rgb, 255:red, 184; green, 233; blue, 134 }  ,fill opacity=1 ][line width=1.5]  (241.62,192.77) .. controls (241.62,191.8) and (242.4,191.02) .. (243.36,191.02) -- (249.76,191.02) .. controls (250.72,191.02) and (251.5,191.8) .. (251.5,192.77) -- (251.5,207.44) .. controls (251.5,208.4) and (250.72,209.18) .. (249.76,209.18) -- (243.36,209.18) .. controls (242.4,209.18) and (241.62,208.4) .. (241.62,207.44) -- cycle ;
        \draw  [fill={rgb, 255:red, 0; green, 0; blue, 0 }  ,fill opacity=1 ][line width=1.5]  (243.36,185.47) -- (244.4,185.47) -- (244.4,191.02) -- (243.36,191.02) -- cycle ;
        
        \draw [color={rgb, 255:red, 80; green, 227; blue, 194 }  ,draw opacity=1 ][line width=1.5]    (119.36,124.71) -- (236.11,88.77) ;
        \draw [shift={(238.98,87.89)}, rotate = 162.89] [color={rgb, 255:red, 80; green, 227; blue, 194 }  ,draw opacity=1 ][line width=1.5]    (14.21,-4.28) .. controls (9.04,-1.82) and (4.3,-0.39) .. (0,0) .. controls (4.3,0.39) and (9.04,1.82) .. (14.21,4.28)   ;
        \draw [color={rgb, 255:red, 248; green, 231; blue, 28 }  ,draw opacity=1 ][line width=1.5]    (119.36,124.71) -- (237.25,124.02) ;
        \draw [shift={(240.25,124)}, rotate = 179.66] [color={rgb, 255:red, 248; green, 231; blue, 28 }  ,draw opacity=1 ][line width=1.5]    (14.21,-4.28) .. controls (9.04,-1.82) and (4.3,-0.39) .. (0,0) .. controls (4.3,0.39) and (9.04,1.82) .. (14.21,4.28)   ;
        \draw [color={rgb, 255:red, 245; green, 166; blue, 35 }  ,draw opacity=1 ][line width=1.5]    (119.36,124.71) -- (237.87,159.41) ;
        \draw [shift={(240.75,160.25)}, rotate = 196.32] [color={rgb, 255:red, 245; green, 166; blue, 35 }  ,draw opacity=1 ][line width=1.5]    (14.21,-4.28) .. controls (9.04,-1.82) and (4.3,-0.39) .. (0,0) .. controls (4.3,0.39) and (9.04,1.82) .. (14.21,4.28)   ;
        \draw [color={rgb, 255:red, 184; green, 233; blue, 134 }  ,draw opacity=1 ][line width=1.5]    (119.25,126) -- (237.43,197.45) ;
        \draw [shift={(240,199)}, rotate = 211.16] [color={rgb, 255:red, 184; green, 233; blue, 134 }  ,draw opacity=1 ][line width=1.5]    (14.21,-4.28) .. controls (9.04,-1.82) and (4.3,-0.39) .. (0,0) .. controls (4.3,0.39) and (9.04,1.82) .. (14.21,4.28)   ;
        \draw  [line width=1.5]  (151,132.5) .. controls (151,108.34) and (156.04,88.75) .. (162.25,88.75) .. controls (168.46,88.75) and (173.5,108.34) .. (173.5,132.5) .. controls (173.5,156.66) and (168.46,176.25) .. (162.25,176.25) .. controls (156.04,176.25) and (151,156.66) .. (151,132.5) -- cycle ;
        \draw   (220.5,65.25) -- (276,65.25) -- (276,218.75) -- (220.5,218.75) -- cycle ;
        \draw   (301.5,99.75) -- (394.5,99.75) -- (394.5,161.25) -- (301.5,161.25) -- cycle ;
        \draw    (276,130.5) -- (300.5,131.19) ;
        \draw [shift={(302.5,131.25)}, rotate = 181.62] [color={rgb, 255:red, 0; green, 0; blue, 0 }  ][line width=0.75]    (10.93,-3.29) .. controls (6.95,-1.4) and (3.31,-0.3) .. (0,0) .. controls (3.31,0.3) and (6.95,1.4) .. (10.93,3.29)   ;
        \draw   (429,52) -- (647,52) -- (647,225) -- (429,225) -- cycle ;
        \draw    (429,100) -- (649,99) ;
        \draw   (231,73) -- (231,36.6) .. controls (231,24.12) and (241.12,14) .. (253.6,14) -- (483.01,14) .. controls (495.49,14) and (505.6,24.12) .. (505.6,36.6) -- (505.6,51.66) -- (517,51.66) -- (502.25,59.19) -- (487.5,51.66) -- (498.9,51.66) -- (498.9,36.6) .. controls (498.9,27.82) and (491.78,20.7) .. (483.01,20.7) -- (253.6,20.7) .. controls (244.82,20.7) and (237.7,27.82) .. (237.7,36.6) -- (237.7,73) -- cycle ;
        
        \draw (154,97) node [anchor=north west][inner sep=0.75pt]  [font=\scriptsize]  {$h_{1}$};
        \draw (154,111) node [anchor=north west][inner sep=0.75pt]  [font=\scriptsize]  {$h_{2}$};
        \draw (154,123.5) node [anchor=north west][inner sep=0.75pt]  [font=\scriptsize]  {$h_{3}$};
        \draw (154,151.5) node [anchor=north west][inner sep=0.75pt]  [font=\scriptsize]  {$h_{K}$};
        \draw (156,142) node [anchor=north west][inner sep=0.75pt]    {$.$};
        \draw (223,72) node [anchor=north west][inner sep=0.75pt]  [font=\scriptsize]  {$\gamma _{1}$};
        \draw (223,105.5) node [anchor=north west][inner sep=0.75pt]  [font=\scriptsize]  {$\gamma _{2}$};
        \draw (223,138) node [anchor=north west][inner sep=0.75pt]  [font=\scriptsize]  {$\gamma _{3}$};
        \draw (223,174.5) node [anchor=north west][inner sep=0.75pt]  [font=\scriptsize]  {$\gamma _{K}$};
        \draw (225,164) node [anchor=north west][inner sep=0.75pt]    {$.$};
        \draw (303,113) node [anchor=north west][inner sep=0.75pt]  [font=\scriptsize]  {$ \begin{array}{l}
        \mathnormal{\gamma _{max}^{K}}\\
        \mathnormal{=max\{\gamma _{1} ...\gamma _{K}\}}
        \end{array}$};
        \draw (485,105) node [anchor=north west][inner sep=0.75pt]   [align=left] {SNR};
        \draw (495,68) node [anchor=north west][inner sep=0.75pt]    {$\gamma $};
        \draw (437,66) node [anchor=north west][inner sep=0.75pt]   [align=left] {ref};
        \draw (511,55) node [anchor=north west][inner sep=0.75pt]    {$ \begin{array}{l}
        \ \ \ \ \ \ \gamma _{max}^{K} 
        \end{array}$};
        \draw (428,105) node [anchor=north west][inner sep=0.75pt]   [align=left] {\cite{narasimhamurthy2011multi}};
        \draw (428,129) node [anchor=north west][inner sep=0.75pt]   [align=left] {\cite{al2017asymptotic}};
        \draw (428,158) node [anchor=north west][inner sep=0.75pt]   [align=left] { \cite{xia2008opportunistic}};
        \draw (428,188) node [anchor=north west][inner sep=0.75pt]   [align=left] { \cite{subhash2019asymptotic}};
        \draw (486,131) node [anchor=north west][inner sep=0.75pt]   [align=left] {SNR};
        \draw (486,160) node [anchor=north west][inner sep=0.75pt]   [align=left] {SNR};
        \draw (487,188) node [anchor=north west][inner sep=0.75pt]   [align=left] {SIR};
        \draw (549,105) node [anchor=north west][inner sep=0.75pt]   [align=left] {Gumbel};
        \draw (549,133) node [anchor=north west][inner sep=0.75pt]   [align=left] {Gumbel};
        \draw (550,162) node [anchor=north west][inner sep=0.75pt]   [align=left] {Gumbel};
        \draw (552,190) node [anchor=north west][inner sep=0.75pt]   [align=left] {Frechet};
        \draw (89,93) node [anchor=north west][inner sep=0.75pt]   [align=left] {\textbf{AP}};
        \draw (203,226) node [anchor=north west][inner sep=0.75pt]   [align=left] {\textbf{User nodes}};                
        \end{tikzpicture}
\caption{Multi-user diversity}
\label{fig:mud}
 \end{figure*}
\subsubsection*{System Description}
    Consider the system shown in Fig.~\ref{fig:mud} having an \ac{AP} and $K$ user nodes, where we assume that each user node and the \ac{AP} have a single antenna. While all user nodes experience independent channel fading, some users may have better channel conditions than others. In this regard, we denote the flat fading channel gains of the $K$ users as $\left | h_{1} \right |^{2}$,$\left | h_{2} \right |^{2}$, $\cdots$, $\left | h_{K} \right |^{2}$, respectively, and the \ac{SNR} of the $i$-th user at the \ac{AP} as $\gamma_{i}=\frac{P}{\sigma^{2}} \left | h_{i} \right |^{2}$. The best channel at \ac{AP} would be the user node link with the highest \ac{SNR}, i.e., $\gamma _{\max}^{K}=\max\left \{\gamma_{1},\gamma_{2},\cdots,\gamma_{K} \right \}$. Similar system models are used in \ac{MIMO} systems, free-space optical communication, opportunistic beamforming, and \ac{CR}.
\subsubsection*{Advantages of \ac{EVT}} 
     The \ac{CDF} of the maximum order statistics, i.e., $\gamma _{\max}^{K}$, involves the multiplication of the individual CDFs of all $\gamma_i$ when the \ac{RVs} are independent. If the \ac{CDF} of $\gamma_i$ itself has a complicated expression, finding the exact \ac{CDF} of $\gamma _{\max}^{K}$ will be a challenge. For example, considering $\kappa-\mu$ shadowed \cite{subhash2019asymptotic} fading environment, the CDF is as shown below
     \begin{align}
         &F_{\gamma }\left ( \gamma  \right )=\frac{\mu ^{\mu -1}m^{m}\left ( 1+\kappa \right )^{\mu }}{\Gamma \left ( \mu  \right )\left ( \mu \kappa+m \right )^m}\left ( \frac{\gamma}{\bar{\gamma}} \right )^{\mu }\nonumber\\
         &\Phi _{2}\left ( \mu -m,m;\mu +1;-\frac{\mu \left ( 1+\kappa \right )\gamma }{\bar{\gamma}},-\frac{\mu \left ( 1+\kappa \right ) }{\bar{\gamma}}\frac{m\gamma }{\mu \kappa+m} \right ),
     \end{align}
     where $\Gamma[.]$ is the gamma function and $\Phi_2(.,.;.;)$ is the bivariate confluent hypergeometric function. 
      Capacity analysis \cite{song2006asymptotic} in multi-user diversity typically involves special functions, and evaluating the \ac{CDF} for $K$ beyond $4$ or $5$ is mathematically intractable. By contrast, asymptotic analysis relying on \ac{EVT} is fairly simple. Similarly, it is shown in  \cite{al2017asymptotic} that finding the exact expression of average throughput for equal gain combining in free-space optical communication is very difficult. To tackle this difficulty, \ac{EVT} is advantageous for the asymptotic analysis instead of the exact analysis when the number of links is high. Deriving performance metrics like the outage probability and ergodic rate is also simpler by using the asymptotic \ac{CDF}. In some situations, the link having maximum \ac{SNR} may not be available, as some other users may occupy it. Hence, the next best link or $k$-th best link must be selected in such scenarios. The asymptotic analysis here involves $k$-th order statistics \cite{al2018asymptotic,al2019asymptotic} for $K-k$ finite and $K\to \infty$. 
\subsubsection*{Applicability Illustration}
   The distributions of the \ac{SNR}/\ac{SIR} have been extensively studied in the literature, starting with the simple Rayleigh fading channels to complex scattering environments associated with complex \ac{SNR} distributions such as $\alpha-\mu$ \cite{al2017asymptotic}, $\kappa-\mu$, $\eta-\mu$ \cite{subhash2019asymptotic}, etc. For different \ac{SNR} distributions such as Rayleigh distribution, $\alpha-\mu$ distribution, and chi-square distribution \cite{xia2008opportunistic} with two degrees of freedom, the distribution of $\gamma _{\max}^{K}$ tends to either the truncated Gumbel or the Gumbel distribution. \\
    \textit{Example}:
    Let us consider the scenario, where the source and interferers undergo $\kappa-\mu$ shadowed fading with parameters ${(\kappa,\mu,m,\bar{x})}$ and ${\{(\kappa_i,\mu_i,m_i,\bar{y_i});i=1,\cdots,N\}}$ respectively, where $N$ is the number of interferers and $\bar{x},\{\bar{y_i};i=1,\cdots,N\}$ are the expectations of the corresponding RVs. Then, the CDF of $\gamma_i$ is given by (\ref{cdf1}),
\begin{figure*}[b]
\hrule
\begin{equation}
\begin{aligned}
	F_{\gamma}(\gamma) = & 1 - K_1  \sum_{p=0}^{\infty}\left[ \frac{(m)_p\left(1-\frac{\theta}{\lambda}\right)^p\Gamma\left[\sum\limits_{i=1}^{N}\mu_i+\mu+p\right]}{(\mu)_pp!} F_D^{(2N)}(1-p-\mu,\mu_1-m_1,\cdots,\mu_N-m_N, \right. \\
	& \left. m_1,\cdots,m_N;1+\sum_{i=1}^{N}\mu_i;\frac{\theta}{\theta+\gamma\theta_1},\cdots,\frac{\theta}{\theta+\gamma\theta_N},\frac{\theta}{\theta+\gamma\lambda_1},\cdots,\frac{\theta}{\theta+\gamma\lambda_N}) \right] 
\end{aligned}
\label{cdf1}
\end{equation}
\end{figure*}
where we have ${K_1=\frac{\prod\limits_{i=1}^{N}\left(\left(\frac{\theta}{\theta+x\theta_i}\right)^{\mu_i-m_i}\left(\frac{\theta}{\theta+x\lambda_i}\right)^{m_i}\right)\theta^m}{\Gamma\left[\sum\limits_{i=1}^{N}\mu_i+1\right]\Gamma[\mu]\lambda^m}}$  \cite{kumar2017outage}. Here, ${(m)_p=\frac{\Gamma[m+p]}{\Gamma[m]}}$ is the Pochhammer symbol, $p!$ is the factorial of $p$, $F_D^{(2N)}(.,.,.;.)$ is the Lauricella function of fourth kind and 
\begin{align}
    &{\theta = \frac{\bar{x}}{\mu_(1+\kappa)}}, \ {\theta_i = \frac{\bar{y_i}}{\mu_i(1+\kappa_i)}}, \nonumber \\ & {\lambda = \frac{(\mu\kappa+m)\bar{x}}{\mu_(1+\kappa_)}, \lambda_i = \frac{(\mu_i\kappa_i+m_i)\bar{y_i}}{\mu_i(1+\kappa_i)m_i}}.
    \label{thetalambda}
\end{align}
 Hence, the distribution of the maximum SNR $\gamma _{max}^{K}$ can be written using order statistics as $F_{\gamma _{max}^{K}}\left ( \gamma \right )=\left ( F\left ( \gamma \right ) \right )^{K}$.
   As we can see, the distribution is complicated and does not provide insights into the system's performance in terms of diversity order. However, applying EVT, where the number of links tends to infinity, the limiting distribution of the maximum SNR is given by,
  \begin{equation}
   {F}_{ {\gamma}_{max}}(\gamma) =  \begin{cases}
		0, & \gamma \leq 0,  \\
		\exp\left(-\left( {\gamma}/{a_K}\right)^{-\beta}\right), & \gamma>0,
		\end{cases} 
		\label{asymp_cdf}
\end{equation} where we have $\beta = \sum\limits_{i=1}^N\mu_i$ and $ a_K = F_{\gamma}^{-1}(1-K^{-1})$. Using results from EVT, we can establish that the limiting distribution of $\gamma _{max}^{K}$ is a Frechet distribution with normalizing constants $b_{K}= 0$ and $a_{K}$, as shown in Fig.~\ref{fig:kms}. 

    \begin{figure}
        \centering
        \includegraphics[scale=0.6]{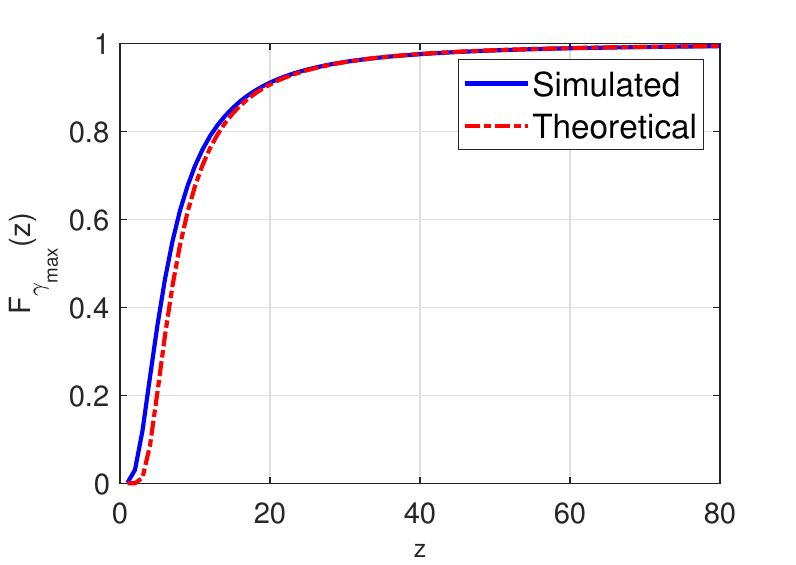}
        \caption{CDF of Maxima of $\kappa-\mu$ shadowed fading scenario for $K=64$.}
        \label{fig:kms}
    \end{figure}
    
    \subsubsection*{State of the Art}
       Next we present a brief literature on \ac{EVT} in the context of multi-user diversity. In 2011, Narasimhamurthy \textit{et al.} \cite{narasimhamurthy2011multi} used \ac{EVT} in a multi-user diversity context, where the number of user nodes is considered random instead of a fixed number. Considering the Poisson distribution for the number of user nodes and Rayleigh fading channels, the CDF of the channel gain of the best user obeys a truncated  Gumbel distribution. Asymptotic throughput analysis is carried out in  \cite{al2017asymptotic},  \cite{xia2008opportunistic}, \cite{song2006asymptotic} with the help of \ac{EVT}. The asymptotic throughput analysis of multi-user diversity under known channel conditions is presented in \cite{song2006asymptotic}. It is also shown that the exact throughput analysis is complicated, even for the Rayleigh fading scenario. It is proved that the limiting distribution for the maximum throughput obeys the Gumbel distribution for distributions like Rayleigh and Nakagami. 
       \par  Xia \textit{et al.}  \cite{xia2008opportunistic}  considered throughput analysis of opportunistic beamforming with the help of extreme order statistics. The system model consists of $M$ transmit antennas and $K$ users, each with a single receive antenna. The $M$ transmit antennas are multiplied by the respective beamforming coefficients, and the authors considered both single-beam and multibeam cases. The base station then selects the user having the highest \ac{SNR}. Assuming each \ac{SNR} $\gamma_{k}$ follows a central chi-square distribution with two degrees of freedom, the authors have derived the limiting distribution of maximum received \ac{SNR} and proved that it follows the Gumbel distribution. Furthermore, the authors have derived the throughput of both single-beam and multibeam opportunistic beamforming systems.
        \begin{table*}[b]  
        \centering
         \resizebox{\textwidth}{!}{%
             \begin{tabular}{|l|l|l|l|l|}
             \hline
              Ref & Application~ & SNR/SINR  & Min/Max   & Asymptotic Distribution  \\
               \hline
             \cite{narasimhamurthy2011multi} & Uplink multi-user diversity       & Channel gain     & Maximum   & Truncated Gumbel    \\
             \cite{al2017asymptotic} & Free space optical communication       & Channel gain         & Maximum   & Truncated Gumbel                  \\
             \cite{xia2008opportunistic} & Opportunistic beamforming        & \ac{SNR}           & Maximum   &  Gumbel \\
              \cite{subhash2019asymptotic} & Antenna selection       & \ac{SIR}           & Maximum   &  Frechet \\
              \cite{al2018asymptotic} & Cognitive radio       & \ac{SNR}             & $k$-th best  & Inverse Gamma RV \\
             \cite{al2019asymptotic} & Cognitive radio        & \ac{SIR}             & $k$-th best  & Inverse Gamma RV \\
              \hline
             \end{tabular}%
             }
             \caption{Applications of \ac{EVT} in multi-user diversity. }\label{Table: MUD}             
        \end{table*}
       \par  Multi-user diversity harnessed in \ac{FSO} systems was considered by Badarneh and Georghiades \cite{al2017asymptotic}. The authors analyzed the asymptotic throughput with the help of \ac{EVT}. The system model has an optical transmitter having $M$ transmit apertures, and multiple users equipped with $N$ receive apertures each. The channel fading coefficients between the $i$-th transmit aperture and $n$-th receive aperture of the $i$-th user is taken as gamma-gamma random variable. With the aid of equal gain combining at the receiver, the combined fading coefficient ($h_{i}$) becomes a $\alpha-\mu$ distributed random variable. Assuming $\gamma_{i}=h_{i}^{2}$, multi-user diversity selects the best link for transmission. Authors have proved that $\gamma _{max}^{K}$ follows the truncated Gumbel distribution. Furthermore, asymptotic average throughput expression is derived using \ac{EVT} for the \ac{FSO} system. 
       \par The authors of \cite{subhash2019asymptotic}, \cite{al2020throughput} considered $\gamma_k$ to be the \ac{SIR} and analyzed $\gamma _{max}^{K}$ in this case. Subhash {\it et al.} \cite{subhash2019asymptotic} used \ac{EVT} to find the limiting distribution of the maximum of $K$ i.i.d. \ac{SIR} \ac{RVs}. They assumed that source and interferers follow   $\kappa-\mu$ shadowed fading with non-identical parameters (i.n.i.d.). It is shown that the limiting distribution follows the Frechet distribution. Further analysis of the convergence rate of the actual maximum distribution towards the asymptotic distribution is presented. Badarneh and Georghiades \cite{al2020throughput} considered a \ac{WPS} with a source, destination, and a co-channel interferer in which the destination has multiple antennas. In this scenario, interference-based \ac{WPS} is considered without any dedicated power beacon. Analytical expressions are derived for the average throughput under delay constraints, assuming that the receiver selects the best \ac{SIR} link.       
       \par So far, we have seen multi-user diversity exploiting the best channel, i.e., the link selection with the highest \ac{SNR}/\ac{SIR}. Sometimes, the highest \ac{SNR} (maximum order statistic) link may not be available or occupied by another user. Hence, the next best link is needed, or $k$-th best link will be selected for communication. Various applications like selection diversity, relay selection, and \ac{CR} tend to rely on $k$-th best link selection.   
      \par Ikki and Ahmed \cite{ikki2010performance} used the order statistics in \ac{DF} and \ac{AF} cooperative-diversity systems, where  $k$-th best relay selection is used for transmission.
       The performance of selection diversity techniques is analyzed by Badarneh and Georghiades \cite{al2018asymptotic_2}. The instantaneous throughput distribution of the largest order statistic converges to the Gumbel distribution. Using this, closed-form asymptotic expressions are derived for the average throughput, effective throughput, and the average \ac{BEP} of the $k$-th best link under different channel models.
       \par Badarneh and Alouni \cite{al2018asymptotic} considered the selection of $k$-th best \ac{SU} in \ac{CR} systems. Considering a large number of \ac{SU}s with fixed $k$, it is shown that $k$-th highest \ac{SNR} converges uniformly in distribution to an inverse gamma distributed \ac{RVs}. In\cite{al2019asymptotic},  \ac{EVT} is exploited for analyzing the asymptotic performance of $k$-th best \ac{SU} selection in underlay \ac{CR} networks. It is proved that the $k$-th best \ac{SIR} converges uniformly in distribution to an inverse gamma \ac{RVs}. With this result, the asymptotic expression of the average throughput, effective throughput, average \ac{BER}, and outage probability for the $k$-th best \ac{SU} is derived.
       \par Badarneh and Alouni \cite{al2018secrecy} studied the secrecy performance of the $k$-th best user selection scheme for an interference-limited wireless network. They  derived the closed-form expressions of the \ac{ESC} for the $k$-th best user. \textcolor{black}{Table \ref{Table: MUD} illustrates the applications of \ac{EVT} in the context of multi-user diversity}. 
    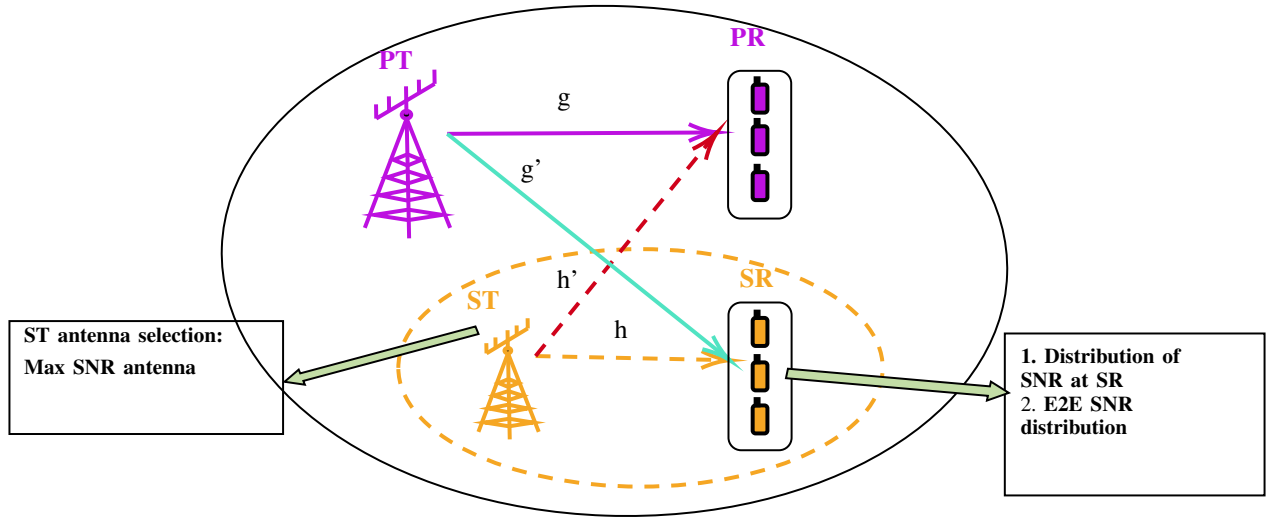
\begin{figure*}
    \centering
    \tikzset{every picture/.style={line width=0.75pt}} 
    \begin{tikzpicture}[x=0.75pt,y=0.75pt,yscale=-1,xscale=1]
        
        \draw  [color={rgb, 255:red, 0; green, 0; blue, 0 }  ,draw opacity=1 ][fill={rgb, 255:red, 189; green, 16; blue, 224 }  ,fill opacity=1 ][line width=1.5]  (391.83,59.12) .. controls (391.83,58.3) and (392.5,57.63) .. (393.32,57.63) -- (397.79,57.63) .. controls (398.61,57.63) and (399.28,58.3) .. (399.28,59.12) -- (399.28,70.29) .. controls (399.28,71.11) and (398.61,71.78) .. (397.79,71.78) -- (393.32,71.78) .. controls (392.5,71.78) and (391.83,71.11) .. (391.83,70.29) -- cycle ;
        \draw  [color={rgb, 255:red, 0; green, 0; blue, 0 }  ,draw opacity=1 ][fill={rgb, 255:red, 189; green, 16; blue, 224 }  ,fill opacity=1 ][line width=1.5]  (393.32,54.45) -- (394.81,54.45) -- (394.81,57.63) -- (393.32,57.63) -- cycle ;
        
        \draw [color={rgb, 255:red, 189; green, 16; blue, 224 }  ,draw opacity=1 ][line width=1.5]    (237.8,82.49) -- (370.87,81.81) ;
        \draw [shift={(373.87,81.79)}, rotate = 179.7] [color={rgb, 255:red, 189; green, 16; blue, 224 }  ,draw opacity=1 ][line width=1.5]    (14.21,-4.28) .. controls (9.04,-1.82) and (4.3,-0.39) .. (0,0) .. controls (4.3,0.39) and (9.04,1.82) .. (14.21,4.28)   ;
        \draw  [color={rgb, 255:red, 0; green, 0; blue, 0 }  ,draw opacity=1 ][fill={rgb, 255:red, 245; green, 166; blue, 35 }  ,fill opacity=1 ][line width=1.5]  (391.83,199.33) .. controls (391.83,198.5) and (392.5,197.84) .. (393.32,197.84) -- (397.79,197.84) .. controls (398.61,197.84) and (399.28,198.5) .. (399.28,199.33) -- (399.28,210.49) .. controls (399.28,211.31) and (398.61,211.98) .. (397.79,211.98) -- (393.32,211.98) .. controls (392.5,211.98) and (391.83,211.31) .. (391.83,210.49) -- cycle ;
        \draw  [color={rgb, 255:red, 0; green, 0; blue, 0 }  ,draw opacity=1 ][fill={rgb, 255:red, 245; green, 166; blue, 35 }  ,fill opacity=1 ][line width=1.5]  (393.32,194.65) -- (394.81,194.65) -- (394.81,197.84) -- (393.32,197.84) -- cycle ;
        
        \draw  [color={rgb, 255:red, 0; green, 0; blue, 0 }  ,draw opacity=1 ][fill={rgb, 255:red, 245; green, 166; blue, 35 }  ,fill opacity=1 ][line width=1.5]  (391.83,221.06) .. controls (391.83,220.23) and (392.5,219.57) .. (393.32,219.57) -- (397.79,219.57) .. controls (398.61,219.57) and (399.28,220.23) .. (399.28,221.06) -- (399.28,232.22) .. controls (399.28,233.04) and (398.61,233.71) .. (397.79,233.71) -- (393.32,233.71) .. controls (392.5,233.71) and (391.83,233.04) .. (391.83,232.22) -- cycle ;
        \draw  [color={rgb, 255:red, 0; green, 0; blue, 0 }  ,draw opacity=1 ][fill={rgb, 255:red, 245; green, 166; blue, 35 }  ,fill opacity=1 ][line width=1.5]  (393.32,216.38) -- (394.81,216.38) -- (394.81,219.57) -- (393.32,219.57) -- cycle ;
        
        \draw  [color={rgb, 255:red, 0; green, 0; blue, 0 }  ,draw opacity=1 ][fill={rgb, 255:red, 245; green, 166; blue, 35 }  ,fill opacity=1 ][line width=1.5]  (391.83,176.89) .. controls (391.83,176.07) and (392.5,175.4) .. (393.32,175.4) -- (397.79,175.4) .. controls (398.61,175.4) and (399.28,176.07) .. (399.28,176.89) -- (399.28,188.06) .. controls (399.28,188.88) and (398.61,189.55) .. (397.79,189.55) -- (393.32,189.55) .. controls (392.5,189.55) and (391.83,188.88) .. (391.83,188.06) -- cycle ;
        \draw  [color={rgb, 255:red, 0; green, 0; blue, 0 }  ,draw opacity=1 ][fill={rgb, 255:red, 245; green, 166; blue, 35 }  ,fill opacity=1 ][line width=1.5]  (393.32,172.22) -- (394.81,172.22) -- (394.81,175.4) -- (393.32,175.4) -- cycle ;
        
        \draw  [color={rgb, 255:red, 245; green, 166; blue, 35 }  ,draw opacity=1 ][dash pattern={on 5.63pt off 4.5pt}][line width=1.5]  (213.13,200.61) .. controls (213.13,167.51) and (267.82,140.68) .. (335.27,140.68) .. controls (402.73,140.68) and (457.42,167.51) .. (457.42,200.61) .. controls (457.42,233.71) and (402.73,260.55) .. (335.27,260.55) .. controls (267.82,260.55) and (213.13,233.71) .. (213.13,200.61) -- cycle ;
        \draw [color={rgb, 255:red, 245; green, 166; blue, 35 }  ,draw opacity=1 ][line width=1.5]  [dash pattern={on 5.63pt off 4.5pt}]  (282.36,194.65) -- (377.23,196.69) ;
        \draw [shift={(380.23,196.76)}, rotate = 181.23] [color={rgb, 255:red, 245; green, 166; blue, 35 }  ,draw opacity=1 ][line width=1.5]    (14.21,-4.28) .. controls (9.04,-1.82) and (4.3,-0.39) .. (0,0) .. controls (4.3,0.39) and (9.04,1.82) .. (14.21,4.28)   ;
        \draw [color={rgb, 255:red, 208; green, 2; blue, 27 }  ,draw opacity=1 ][line width=1.5]  [dash pattern={on 5.63pt off 4.5pt}]  (282.36,194.65) -- (371.98,84.12) ;
        \draw [shift={(373.87,81.79)}, rotate = 129.03] [color={rgb, 255:red, 208; green, 2; blue, 27 }  ,draw opacity=1 ][line width=1.5]    (14.21,-4.28) .. controls (9.04,-1.82) and (4.3,-0.39) .. (0,0) .. controls (4.3,0.39) and (9.04,1.82) .. (14.21,4.28)   ;
        \draw [color={rgb, 255:red, 80; green, 227; blue, 194 }  ,draw opacity=1 ][line width=1.5]    (237.8,82.49) -- (377.89,194.88) ;
        \draw [shift={(380.23,196.76)}, rotate = 218.74] [color={rgb, 255:red, 80; green, 227; blue, 194 }  ,draw opacity=1 ][line width=1.5]    (14.21,-4.28) .. controls (9.04,-1.82) and (4.3,-0.39) .. (0,0) .. controls (4.3,0.39) and (9.04,1.82) .. (14.21,4.28)   ;
        \draw  [color={rgb, 255:red, 0; green, 0; blue, 0 }  ,draw opacity=1 ][fill={rgb, 255:red, 189; green, 16; blue, 224 }  ,fill opacity=1 ][line width=1.5]  (391.83,80.15) .. controls (391.83,79.33) and (392.5,78.66) .. (393.32,78.66) -- (397.79,78.66) .. controls (398.61,78.66) and (399.28,79.33) .. (399.28,80.15) -- (399.28,91.32) .. controls (399.28,92.14) and (398.61,92.81) .. (397.79,92.81) -- (393.32,92.81) .. controls (392.5,92.81) and (391.83,92.14) .. (391.83,91.32) -- cycle ;
        \draw  [color={rgb, 255:red, 0; green, 0; blue, 0 }  ,draw opacity=1 ][fill={rgb, 255:red, 189; green, 16; blue, 224 }  ,fill opacity=1 ][line width=1.5]  (393.32,75.48) -- (394.81,75.48) -- (394.81,78.66) -- (393.32,78.66) -- cycle ;
        
        \draw  [color={rgb, 255:red, 0; green, 0; blue, 0 }  ,draw opacity=1 ][fill={rgb, 255:red, 189; green, 16; blue, 224 }  ,fill opacity=1 ][line width=1.5]  (391.83,103.29) .. controls (391.83,102.46) and (392.5,101.8) .. (393.32,101.8) -- (397.79,101.8) .. controls (398.61,101.8) and (399.28,102.46) .. (399.28,103.29) -- (399.28,114.45) .. controls (399.28,115.27) and (398.61,115.94) .. (397.79,115.94) -- (393.32,115.94) .. controls (392.5,115.94) and (391.83,115.27) .. (391.83,114.45) -- cycle ;
        \draw  [color={rgb, 255:red, 0; green, 0; blue, 0 }  ,draw opacity=1 ][fill={rgb, 255:red, 189; green, 16; blue, 224 }  ,fill opacity=1 ][line width=1.5]  (393.32,98.62) -- (394.81,98.62) -- (394.81,101.8) -- (393.32,101.8) -- cycle ;
        
        \draw   (379.44,57.05) .. controls (379.44,53.54) and (382.29,50.69) .. (385.8,50.69) -- (404.9,50.69) .. controls (408.42,50.69) and (411.27,53.54) .. (411.27,57.05) -- (411.27,118.89) .. controls (411.27,122.4) and (408.42,125.25) .. (404.9,125.25) -- (385.8,125.25) .. controls (382.29,125.25) and (379.44,122.4) .. (379.44,118.89) -- cycle ;
        \draw   (379.64,173.99) .. controls (379.64,170.47) and (382.49,167.62) .. (386.01,167.62) -- (405.11,167.62) .. controls (408.62,167.62) and (411.47,170.47) .. (411.47,173.99) -- (411.47,235.83) .. controls (411.47,239.34) and (408.62,242.19) .. (405.11,242.19) -- (386.01,242.19) .. controls (382.49,242.19) and (379.64,239.34) .. (379.64,235.83) -- cycle ;
        \draw   (124.12,139.72) .. controls (126.6,68.72) and (217.27,14.26) .. (326.63,18.08) .. controls (435.99,21.9) and (522.64,82.55) .. (520.16,153.55) .. controls (517.68,224.55) and (427.01,279.01) .. (317.65,275.19) .. controls (208.29,271.37) and (121.64,210.72) .. (124.12,139.72) -- cycle ;
        \draw [color={rgb, 255:red, 189; green, 16; blue, 224 }  ,draw opacity=1 ][line width=1.5]    (217.33,74.78) -- (196.22,129.94) ;
        \draw [color={rgb, 255:red, 189; green, 16; blue, 224 }  ,draw opacity=1 ][line width=1.5]    (217.33,74.78) -- (239.26,129.94) ;
        \draw [color={rgb, 255:red, 189; green, 16; blue, 224 }  ,draw opacity=1 ][line width=1.5]    (217.33,74.78) -- (218.57,135.24) ;
        \draw  [color={rgb, 255:red, 189; green, 16; blue, 224 }  ,draw opacity=1 ][line width=1.5]  (217.74,92.82) -- (224.78,95.64) -- (217.74,98.47) -- (210.71,95.64) -- cycle ;
        \draw  [color={rgb, 255:red, 189; green, 16; blue, 224 }  ,draw opacity=1 ][line width=1.5]  (217.54,101.3) -- (227.68,104.48) -- (217.54,107.66) -- (207.4,104.48) -- cycle ;
        \draw  [color={rgb, 255:red, 189; green, 16; blue, 224 }  ,draw opacity=1 ][line width=1.5]  (217.51,110.85) -- (231.24,113.5) -- (217.51,116.15) -- (203.78,113.5) -- cycle ;
        \draw  [color={rgb, 255:red, 189; green, 16; blue, 224 }  ,draw opacity=1 ][line width=1.5]  (218.11,120.75) -- (235.21,123.4) -- (218.11,126.05) -- (201,123.4) -- cycle ;
        \draw  [color={rgb, 255:red, 189; green, 16; blue, 224 }  ,draw opacity=1 ][fill={rgb, 255:red, 0; green, 0; blue, 0 }  ,fill opacity=1 ][line width=1.5]  (214.85,72.84) .. controls (214.85,71.76) and (215.91,70.89) .. (217.23,70.89) .. controls (218.54,70.89) and (219.61,71.76) .. (219.61,72.84) .. controls (219.61,73.91) and (218.54,74.78) .. (217.23,74.78) .. controls (215.91,74.78) and (214.85,73.91) .. (214.85,72.84) -- cycle ;
        \draw [color={rgb, 255:red, 189; green, 16; blue, 224 }  ,draw opacity=1 ][line width=1.5]    (217.23,72.84) -- (217.11,65.67) ;
        \draw [color={rgb, 255:red, 189; green, 16; blue, 224 }  ,draw opacity=1 ][line width=1.5]    (201.99,74.78) -- (231.43,57.26) ;
        \draw [color={rgb, 255:red, 189; green, 16; blue, 224 }  ,draw opacity=1 ][line width=1.5]    (231.43,57.26) -- (231.32,50.09) ;
        \draw [color={rgb, 255:red, 189; green, 16; blue, 224 }  ,draw opacity=1 ][line width=1.5]    (201.99,74.78) -- (201.88,67.61) ;
        \draw [color={rgb, 255:red, 189; green, 16; blue, 224 }  ,draw opacity=1 ][line width=1.5]    (208.47,70.74) -- (208.36,63.57) ;
        \draw [color={rgb, 255:red, 189; green, 16; blue, 224 }  ,draw opacity=1 ][line width=1.5]    (217.11,65.67) -- (217,58.5) ;
        \draw [color={rgb, 255:red, 189; green, 16; blue, 224 }  ,draw opacity=1 ][line width=1.5]    (223.59,62.32) -- (223.48,55.15) ;
        
        \draw [color={rgb, 255:red, 245; green, 166; blue, 35 }  ,draw opacity=1 ][line width=1.5]    (268.44,193.3) -- (253.52,232.91) ;
        \draw [color={rgb, 255:red, 245; green, 166; blue, 35 }  ,draw opacity=1 ][line width=1.5]    (268.44,193.3) -- (283.95,232.91) ;
        \draw [color={rgb, 255:red, 245; green, 166; blue, 35 }  ,draw opacity=1 ][line width=1.5]    (268.44,193.3) -- (269.32,236.71) ;
        \draw  [color={rgb, 255:red, 245; green, 166; blue, 35 }  ,draw opacity=1 ][line width=1.5]  (268.73,206.25) -- (273.71,208.28) -- (268.73,210.31) -- (263.76,208.28) -- cycle ;
        \draw  [color={rgb, 255:red, 245; green, 166; blue, 35 }  ,draw opacity=1 ][line width=1.5]  (268.59,212.34) -- (275.76,214.63) -- (268.59,216.91) -- (261.42,214.63) -- cycle ;
        \draw  [color={rgb, 255:red, 245; green, 166; blue, 35 }  ,draw opacity=1 ][line width=1.5]  (268.57,219.2) -- (278.27,221.1) -- (268.57,223) -- (258.86,221.1) -- cycle ;
        \draw  [color={rgb, 255:red, 245; green, 166; blue, 35 }  ,draw opacity=1 ][line width=1.5]  (268.99,226.3) -- (281.09,228.21) -- (268.99,230.11) -- (256.89,228.21) -- cycle ;
        \draw  [color={rgb, 255:red, 245; green, 166; blue, 35 }  ,draw opacity=1 ][fill={rgb, 255:red, 0; green, 0; blue, 0 }  ,fill opacity=1 ][line width=1.5]  (266.68,191.9) .. controls (266.68,191.13) and (267.44,190.51) .. (268.37,190.51) .. controls (269.3,190.51) and (270.05,191.13) .. (270.05,191.9) .. controls (270.05,192.67) and (269.3,193.3) .. (268.37,193.3) .. controls (267.44,193.3) and (266.68,192.67) .. (266.68,191.9) -- cycle ;
        \draw [color={rgb, 255:red, 245; green, 166; blue, 35 }  ,draw opacity=1 ][line width=1.5]    (268.37,191.9) -- (268.29,186.75) ;
        \draw [color={rgb, 255:red, 245; green, 166; blue, 35 }  ,draw opacity=1 ][line width=1.5]    (257.6,193.3) -- (278.41,180.71) ;
        \draw [color={rgb, 255:red, 245; green, 166; blue, 35 }  ,draw opacity=1 ][line width=1.5]    (278.41,180.71) -- (278.33,175.57) ;
        \draw [color={rgb, 255:red, 245; green, 166; blue, 35 }  ,draw opacity=1 ][line width=1.5]    (257.6,193.3) -- (257.51,188.15) ;
        \draw [color={rgb, 255:red, 245; green, 166; blue, 35 }  ,draw opacity=1 ][line width=1.5]    (262.18,190.39) -- (262.1,185.24) ;
        \draw [color={rgb, 255:red, 245; green, 166; blue, 35 }  ,draw opacity=1 ][line width=1.5]    (268.29,186.75) -- (268.2,181.61) ;
        \draw [color={rgb, 255:red, 245; green, 166; blue, 35 }  ,draw opacity=1 ][line width=1.5]    (272.87,184.35) -- (272.79,179.2) ;
        
        \draw   (17,177) -- (155,177) -- (155,234) -- (17,234) -- cycle ;
        \draw   (519,183.26) -- (650,183.26) -- (650,265) -- (519,265) -- cycle ;
        \draw  [color={rgb, 255:red, 0; green, 0; blue, 0 }  ,draw opacity=1 ][fill={rgb, 255:red, 195; green, 221; blue, 166 }  ,fill opacity=1 ] (155.75,207.63) -- (162.88,201.49) -- (163.66,203.39) -- (252.02,179.72) -- (253.58,183.53) -- (165.22,207.2) -- (166,209.11) -- cycle ;
        \draw  [color={rgb, 255:red, 0; green, 0; blue, 0 }  ,draw opacity=1 ][fill={rgb, 255:red, 195; green, 221; blue, 166 }  ,fill opacity=1 ] (521,213.22) -- (509.98,217.12) -- (510.48,214.73) -- (408.62,205.82) -- (409.63,201.04) -- (511.48,209.95) -- (511.98,207.56) -- cycle ;
        
        \draw (379.07,27.88) node [anchor=north west][inner sep=0.75pt]  [color={rgb, 255:red, 189; green, 16; blue, 224 }  ,opacity=1 ] [align=left] {\textbf{PR}};
        \draw (291.97,60.13) node [anchor=north west][inner sep=0.75pt]   [align=left] {g};
        \draw (273.46,93.78) node [anchor=north west][inner sep=0.75pt]   [align=left] {g'};
        \draw (320.61,175.09) node [anchor=north west][inner sep=0.75pt]   [align=left] {h};
        \draw (291.76,149.86) node [anchor=north west][inner sep=0.75pt]   [align=left] {h'};
        \draw (201.83,39.1) node [anchor=north west][inner sep=0.75pt]  [color={rgb, 255:red, 189; green, 16; blue, 224 }  ,opacity=1 ] [align=left] {\textbf{PT}};
        \draw (383.84,149.15) node [anchor=north west][inner sep=0.75pt]  [color={rgb, 255:red, 245; green, 166; blue, 35 }  ,opacity=1 ] [align=left] {\textbf{SR}};
        \draw (246.39,160.37) node [anchor=north west][inner sep=0.75pt]  [color={rgb, 255:red, 245; green, 166; blue, 35 }  ,opacity=1 ] [align=left] {\textbf{ST}};
        \draw (22.87,178.95) node [anchor=north west][inner sep=0.75pt]   [align=left] {{\footnotesize \textbf{ST antenna selection:}}\\{\footnotesize \textbf{Max SNR antenna}}};
        \draw (525.51,189.33) node [anchor=north west][inner sep=0.75pt]  [font=\footnotesize] [align=left] {\textbf{1. Distribution of \ }\\\textbf{SNR at SR}\\2. \textbf{E2E SNR }\\\textbf{distribution}};
    \end{tikzpicture}
    \caption{Spectrum-sharing system with PT, PR, ST, and SR. \ac{EVT} can be exploited in finding the distribution of SNR/SIR at SR, taking the decision on transmit antenna selection, and finding the E2E SNR distribution.}
    \label{fig:ss}
    \end{figure*}   
                
    \subsubsection*{Summary/Takeaway}
       By using \ac{EVT}, one can obtain the maximum \ac{SNR}/\ac{SIR} statistic in the form of Gumbel/Frechet \ac{CDF}s, while the exact distribution is significantly more complicated. For \ac{NG} applications, where one expects to use generalized fading models such as $\kappa-\mu$ \cite{subhash2019asymptotic} or $\alpha-\mu$ \cite{al2017asymptotic}, evaluating the exact distribution for large $k$ is intractable and one has no choice but to resort to \ac{EVT}. Furthermore, the works \cite{al2017asymptotic,xia2008opportunistic,song2006asymptotic,al2020throughput} use the limiting distribution in throughput analysis. When the maximum-order statistic is not available, $k$-th order statistics are used in \cite{ikki2010performance}-\cite{al2018secrecy} to derive the closed-form expressions for the average throughput, effective throughput, the average \ac{BEP}, \ac{ESC}, and outage probability. 

\subsection{Spectrum Sharing}
\begin{figure*}[b]
\hrule
        \begin{equation}\label{comp}
            F_{\hat{\gamma}_i}\left ( x \right )=\sum_{m=1}^{n_r} B_{m-1}^{n_r-1}\left ( \sigma_{g}^{2} \right )\left \{ 1-\frac{1}{\Gamma \left ( m \right )} \left [ \Gamma \left ( m,\frac{x}{\mu_g S_P} \right )-\left ( \frac{x}{\mu_g S_Q+x} \right )^{m}\Gamma \left ( m,\frac{x}{\mu_g S_P}+\frac{S_Q}{S_P} \right ) \right  ]\right \},
        \end{equation}
    \end{figure*}
    The proliferation of mobile users is creating a scarcity of radio spectrum, leading to spectrum sharing to improve spectral efficiency. Depending upon the applications, various spectrum-sharing techniques exist \cite{baby2016comparative}. These spectrum-sharing techniques can be classified as interweave, underlay, and overlay schemes. Spectrum will be shared between primary user and \ac{SU}, where primary users will be licensed to access the spectrum. When the primary users are not using the spectrum (interweave), it can be allocated to \ac{SU}s. In the case of underlay, spectrum sharing happens cooperatively without interfering with each other. In this case,  the \ac{SU}s must limit their transmit power to avoid interference with the primary users. In the overlay, \ac{SU}s can share the spectrum in time, spatial, and frequency domains.    
    
\subsubsection*{System Description}
     Consider a \ac{CR} scenario, where we have a \ac{PT},  \ac{PR}, \ac{ST}, and \ac{SR} as shown in Fig.~\ref{fig:ss}. Each \ac{PT}, \ac{ST}, \ac{PR}, and \ac{SR} can have either a single antenna or multiple antennas depending upon the scenario considered. The channel gains between \ac{PT}-\ac{PR} and \ac{PT}-\ac{SR} are represented by $g$ and ${g}'$, respectively, while the channel gains between \ac{ST}-\ac{SR} and \ac{ST}-\ac{PR} are denoted by $h$ and ${h}'$, respectively. EVT can be used in scenarios like selecting the maximum-\ac{SNR} antenna at \ac{ST}/\ac{SR}, selecting the minimum \ac{SIR} user in a multicast system, and deriving the PDF of \ac{SINR}/\ac{SNR}/\ac{SIR} at the receiver.  Various system models, e.g., multi-hop relaying \cite{xia2013spectrum}, spectrum sharing in the indoor network with a micro base station at \ac{ST} and multiple \ac{SR}s \cite{haider2015spectral}, as well as multicast \ac{CR} networks\cite{subhash2020transmit}, have been explored in the context of \ac{EVT}. 
\subsubsection*{Advantages of \ac{EVT}}
    Performance metrics such as the average rate, outage rate, and average symbol error rate can be computed easily at the \ac{SR}s under the constraints of maximal transmit power at \ac{ST} and maximum interference power at \ac{PR} with the help of \ac{EVT}. Furthermore, scaling laws can be derived with the use of performance metrics. As given in \cite[Eq.(10)]{duan2019asymptotic}, we can observe that the post-processed \ac{SNR} at \ac{SR} has a complicated expression (\ref{comp}). Here, $ B_{i}^{k}\left ( t \right )$ denotes  Bernstein basis polynomial of degree $k$, $S_{P}=\frac{P_{m}}{\sigma ^{2}}$, and $S_{Q}=\frac{Q}{\mu _{h}\sigma ^{2}}$. Here,  $\mu _{h}$ denotes the average path loss between the $i^ {th}$ ST antenna and the receive antenna of the PR, $P_m$ is the maximum transmission power of ST, and $Q$ denotes interference power threshold.
    
    So, finding the exact limiting distribution in such a scenario is more complicated ($F_{\gamma_{max}}\left ( x \right )=\left ( F_{\hat{\gamma}_i} \right )^{n_t}$). Deriving the performance metrics would also be more complex. Further, it is shown by \cite{duan2019asymptotic} that using \ac{EVT} can simplify the expressions of the limiting distributions and subsequent performance metrics ($F_{\gamma_{max}}\left ( x \right )=\exp\left ( \frac{-b_n}{x} \right )$, where $b_n=\frac{\mu _g S_Q}{\left [ \frac{n_t}{n_t -1} \right ]^{\frac{\zeta }{n_r}}-1}$). Even for spectrum-sharing aided multi-hop \ac{AF} relaying systems \cite{xia2013spectrum}, deriving the closed-form expressions involves complex mathematics. \ac{EVT} is used to simplify the expressions so that the performance analysis can be done conveniently. Generalized fading models will have complicated expressions of \ac{CDF} for \ac{SNR}/\ac{SINR}/\ac{SIR}. Thus, calculating the maxima/minima of \ac{RVs} in such a scenario is very complicated. \ac{EVT} is very helpful for solving such complicated expressions \cite{subhash2020transmit}. 
    \begin{table*}[b]
        \centering
        \caption{Applications of \ac{EVT} in \ac{CR} networks. }\label{Table: CR}
         \resizebox{\textwidth}{!}{%
         \begin{tabular}{|l|l|l|l|l|}
             \hline
              Ref & Application~ & \ac{SNR}/\ac{SINR}  & Min/Max   & Asymptotic Distribution  \\
               \hline
               \cite{xia2013spectrum} & Multi-hop relaying in \ac{CR} networks      & \ac{SNR}            & Maximum   & Weibull                  \\
             \cite{hong2011throughput} & Interweave \ac{CR} networks       & \ac{SNR}             & Maximum  & -                  \\
             \cite{haider2015spectral} & Underlay \ac{CR} networks       & \ac{SINR}            & Maximum  & -                  \\
               \cite{duan2019asymptotic} & Transmit antenna selection in \ac{CR} networks      & \ac{SNR}, \ac{SIR}            & Maximum    & Gumbel,Frechet                   \\
               \cite{subhash2020transmit} & Multicast \ac{CR} networks      & \ac{SIR}             & Minimum   & Weibull                   \\
             \hline
         \end{tabular}%
         }
    \end{table*}
\subsubsection*{Applicability Illustration}
    Different modes of spectrum-sharing methods will lead to different \ac{SNR}/\ac{SIR} distributions. Hong {\it et al.} in \cite{hong2011throughput}  considered the channel gains of both the desired and interfering signals as exponential distributions. Hence \ac{SNR}/\ac{SIR} distribution would be a ratio of exponential \ac{RVs}, while the limiting distributions in these cases would be Frechet. The multi-hop relaying systems in \cite{xia2013spectrum} used simple Rayleigh fading, leading to Pareto distribution as \ac{E2E} \ac{SNR}. Further, even more complicated channel gains like log-normal \cite{haider2015spectral}, chi-square \cite{duan2019asymptotic}, and $\kappa-\mu$ shadowed fading \cite{subhash2020transmit} were explored in getting the \ac{SNR}/\ac{SIR} distributions at \ac{SR}.\\
    \textit{Example}:
    \par Xia {\it et al.} \cite{xia2013spectrum} considered the multi-hop relaying link between \ac{ST} as well as \ac{SR} (K relaying paths)  and studied the limiting distribution functions of the lower and the upper bounds on the \ac{E2E} \ac{SNR} $\gamma_{e2e}$ of the relaying path. Let $\gamma_k$ be the received \ac{SNR} at $k$-th secondary node. Then, the \ac{E2E} \ac{SNR} \cite{xia2013spectrum} is given by $\gamma_{e 2 e}=\left(\sum_{k=1}^K \frac{1}{\gamma_k}\right)^{-1}$. The authors have derived the exact distribution functions and \ac{MGF} for $\gamma_{e2e}$. As they are quite complex for analysis, the authors have proposed a pair of lower and upper bounds on $\gamma_{e2e}$, i.e.,
    \begin{equation}
        \gamma_{e 2 e} \leq\left(\max _{k=1, \cdots, K} \frac{1}{\gamma_k}\right)^{-1}=\min _{k=1, \cdots, K} \gamma_k 
         = \gamma_{e 2 e}^{\text {upper }},
    \end{equation}
    \begin{equation}
        \gamma_{e 2 e} \geq\left(K \max _{k=1, \cdots, K} \frac{1}{\gamma_k}\right)^{-1}=\frac{1}{K} \min _{k=1, \cdots, K} \gamma_k 
     = \gamma_{e 2 e}^{\text {lower }}.
    \end{equation}
    The limiting distributions of the upper and lower bounds of the \ac{E2E} \ac{SNR} are derived using Gnedenko’s sufficient and necessary conditions\cite{galambos1978asymptotic}. Assuming that each $\gamma_k$ follows the Pareto distribution, the authors have proved that limiting distribution of $\gamma_{e 2 e}^{\text {upper }}$ belongs to the Weibull family.
\subsubsection*{State of the Art}
    The authors of \cite{hong2011throughput}  considered an interweave system with \ac{PT}, \ac{PR}, \ac{SR}, and multiple \ac{ST}s with the assumption of only utilizing the primary user spectrum when it is not occupied. The \ac{SR} selects the \ac{ST} with the highest \ac{SNR}, and the authors proved that the limiting distribution follows the Frechet CDF in this case. Furthermore, it is used to derive the secondary networks' asymptotic throughput.
  
     In \cite{haider2015spectral}, the trade-off between spectral and energy efficiency is analyzed by Haider and Wang at both the link and system levels in \ac{CR} networks. The system model has a micro-base station as \ac{ST} with multiple SRs and shares the spectrum with indoor primary networks with multiple \ac{PR}s. At the system level, \ac{EVT} is used to derive the \ac{PDF} of \ac{SINR} at the secondary receiver. The expression of spectral efficiency is derived for the optimal power allocation. Subsequently, it is used to derive the average energy efficiency in \ac{CR} networks.
    \par Ruan {\it et al.} \cite{duan2019asymptotic} focuses their attention on the asymptotic analysis of spectrum sharing systems assuming a large number of antennas as \ac{ST}. The system model relies on a single antenna at each \ac{PT} and \ac{PR}, but multiple antennas at \ac{SR} using the \ac{MRC} technique. \ac{TAS} is employed at \ac{ST} with the help of both perfect and imperfect \ac{CSI} from \ac{SR}. Assuming that \ac{TAS} will select the antenna having the highest \ac{SNR}, the authors derived the asymptotic distribution of \ac{SNR} at \ac{SR}, considering a large number of transmit antennas. It is proved that the limiting \ac{SNR} distribution of the \ac{SU} follows Frechet distribution, if the maximum transmit power dominates over the maximum interference power. Otherwise, it follows the Gumbel distribution. Furthermore, the performance analysis is carried out with the help of these distributions.
    
     \ac{EVT} is used by Subhash and Srinivasan \cite{subhash2020transmit} in finding the limiting distribution of minimum of \ac{i.i.d.} \ac{SIR} \ac{RVs} under $\kappa-\mu$ shadowed fading in a multicast (using a single channel for the same service by a group of users) \ac{CR} environment. The system model consists of a single \ac{PT} covering $M$ multicast \ac{PR}s and a single \ac{ST} covering $L$ multicast \ac{SR}s. Here the authors assumed that both the desired and interfering signals undergo $\kappa-\mu$ shadowed fading with non-identical parameters. Considering such \ac{i.i.d.} \ac{SIR} \ac{RVs}, it is proved that the limiting distribution of minimum follows Weibull distribution. Furthermore, it analyzes the power allocation at \ac{ST} and the ergodic multicast rate of \ac{SU}s under different \ac{QoS} constraints.  \textcolor{black}{Table \ref{Table: CR}  outlines the various applications of \ac{EVT} within the framework of \ac{CR} networks}.
\subsubsection*{Summary/Takeaway}
    \ac{CR} networks can encounter maximum or minimum of \ac{SNR}/\ac{SIR} \ac{RVs} in situations like the maximum-\ac{SNR} antenna selection at \ac{ST}/\ac{SR} or minimum \ac{SIR} user in a multicast system. In such scenarios, \ac{EVT} would be of benefit for simplified performance analysis. The analysis of performance metrics at \ac{SR} in \ac{CR} networks is simplified with the help of \ac{EVT}. Various system models like multi-hop relaying\cite{xia2013spectrum}, interweave \cite{hong2011throughput}, underlay\cite{haider2015spectral}, and multicast\cite{subhash2020transmit} systems explored the \ac{RVs}, from simple Rayleigh fading to more complicated channels like $\kappa-\mu$ shadowed fading. Deriving the maximum order statistics of \ac{SNR}/\ac{SIR} would have been more complex without \ac{EVT}. Once the limiting distribution of \ac{SNR}/\ac{SIR} is derived at the \ac{SR}, performance metrics like average rate, outage rate, average symbol error rate, and scaling laws are derived using the limiting distribution of \ac{SNR}/\ac{SIR}.    

\subsection{Relays}
\begin{figure*}
    \centering
\tikzset{every picture/.style={line width=0.75pt}} 

\begin{tikzpicture}[x=0.75pt,y=0.75pt,yscale=-1,xscale=1]

\draw  [fill={rgb, 255:red, 128; green, 128; blue, 128 }  ,fill opacity=1 ] (49,142.88) .. controls (49,140.18) and (51.18,138) .. (53.88,138) -- (68.52,138) .. controls (71.22,138) and (73.4,140.18) .. (73.4,142.88) -- (73.4,169.32) .. controls (73.4,172.02) and (71.22,174.2) .. (68.52,174.2) -- (53.88,174.2) .. controls (51.18,174.2) and (49,172.02) .. (49,169.32) -- cycle ;
\draw  [fill={rgb, 255:red, 0; green, 0; blue, 0 }  ,fill opacity=1 ] (53.88,126.2) -- (60.4,126.2) -- (60.4,138) -- (53.88,138) -- cycle ;
\draw  [fill={rgb, 255:red, 128; green, 128; blue, 128 }  ,fill opacity=1 ] (368,139.88) .. controls (368,137.18) and (370.18,135) .. (372.88,135) -- (387.52,135) .. controls (390.22,135) and (392.4,137.18) .. (392.4,139.88) -- (392.4,166.32) .. controls (392.4,169.02) and (390.22,171.2) .. (387.52,171.2) -- (372.88,171.2) .. controls (370.18,171.2) and (368,169.02) .. (368,166.32) -- cycle ;
\draw  [fill={rgb, 255:red, 0; green, 0; blue, 0 }  ,fill opacity=1 ] (372.88,123.2) -- (379.4,123.2) -- (379.4,135) -- (372.88,135) -- cycle ;
\draw  [fill={rgb, 255:red, 220; green, 214; blue, 214 }  ,fill opacity=1 ] (194,41.6) .. controls (194,34.64) and (202.6,29) .. (213.2,29) .. controls (223.8,29) and (232.4,34.64) .. (232.4,41.6) .. controls (232.4,48.56) and (223.8,54.2) .. (213.2,54.2) .. controls (202.6,54.2) and (194,48.56) .. (194,41.6) -- cycle ;
\draw  [fill={rgb, 255:red, 220; green, 214; blue, 214 }  ,fill opacity=1 ] (194,90.6) .. controls (194,83.64) and (202.6,78) .. (213.2,78) .. controls (223.8,78) and (232.4,83.64) .. (232.4,90.6) .. controls (232.4,97.56) and (223.8,103.2) .. (213.2,103.2) .. controls (202.6,103.2) and (194,97.56) .. (194,90.6) -- cycle ;
\draw  [fill={rgb, 255:red, 220; green, 214; blue, 214 }  ,fill opacity=1 ] (193,167.6) .. controls (193,160.64) and (201.6,155) .. (212.2,155) .. controls (222.8,155) and (231.4,160.64) .. (231.4,167.6) .. controls (231.4,174.56) and (222.8,180.2) .. (212.2,180.2) .. controls (201.6,180.2) and (193,174.56) .. (193,167.6) -- cycle ;
\draw  [fill={rgb, 255:red, 220; green, 214; blue, 214 }  ,fill opacity=1 ] (193,239.6) .. controls (193,232.64) and (201.6,227) .. (212.2,227) .. controls (222.8,227) and (231.4,232.64) .. (231.4,239.6) .. controls (231.4,246.56) and (222.8,252.2) .. (212.2,252.2) .. controls (201.6,252.2) and (193,246.56) .. (193,239.6) -- cycle ;
\draw   (210,191.36) .. controls (210,189.51) and (211.21,188) .. (212.7,188) .. controls (214.19,188) and (215.4,189.51) .. (215.4,191.36) .. controls (215.4,193.22) and (214.19,194.73) .. (212.7,194.73) .. controls (211.21,194.73) and (210,193.22) .. (210,191.36) -- cycle ;
\draw   (210,203.1) .. controls (210,201.24) and (211.21,199.74) .. (212.7,199.74) .. controls (214.19,199.74) and (215.4,201.24) .. (215.4,203.1) .. controls (215.4,204.96) and (214.19,206.46) .. (212.7,206.46) .. controls (211.21,206.46) and (210,204.96) .. (210,203.1) -- cycle ;
\draw   (210,215.84) .. controls (210,213.98) and (211.21,212.47) .. (212.7,212.47) .. controls (214.19,212.47) and (215.4,213.98) .. (215.4,215.84) .. controls (215.4,217.69) and (214.19,219.2) .. (212.7,219.2) .. controls (211.21,219.2) and (210,217.69) .. (210,215.84) -- cycle ;
\draw [color={rgb, 255:red, 74; green, 144; blue, 226 }  ,draw opacity=1 ]   (82.4,161.2) -- (193.7,230.15) ;
\draw [shift={(195.4,231.2)}, rotate = 211.78] [color={rgb, 255:red, 74; green, 144; blue, 226 }  ,draw opacity=1 ][line width=0.75]    (10.93,-3.29) .. controls (6.95,-1.4) and (3.31,-0.3) .. (0,0) .. controls (3.31,0.3) and (6.95,1.4) .. (10.93,3.29)   ;
\draw [color={rgb, 255:red, 80; green, 227; blue, 194 }  ,draw opacity=1 ]   (82.4,161.2) -- (191,167.48) ;
\draw [shift={(193,167.6)}, rotate = 183.31] [color={rgb, 255:red, 80; green, 227; blue, 194 }  ,draw opacity=1 ][line width=0.75]    (10.93,-3.29) .. controls (6.95,-1.4) and (3.31,-0.3) .. (0,0) .. controls (3.31,0.3) and (6.95,1.4) .. (10.93,3.29)   ;
\draw [color={rgb, 255:red, 245; green, 166; blue, 35 }  ,draw opacity=1 ]   (82.4,161.2) -- (188.69,96.24) ;
\draw [shift={(190.4,95.2)}, rotate = 148.57] [color={rgb, 255:red, 245; green, 166; blue, 35 }  ,draw opacity=1 ][line width=0.75]    (10.93,-3.29) .. controls (6.95,-1.4) and (3.31,-0.3) .. (0,0) .. controls (3.31,0.3) and (6.95,1.4) .. (10.93,3.29)   ;
\draw [color={rgb, 255:red, 208; green, 2; blue, 27 }  ,draw opacity=1 ]   (82.4,161.2) -- (192.64,43.06) ;
\draw [shift={(194,41.6)}, rotate = 133.02] [color={rgb, 255:red, 208; green, 2; blue, 27 }  ,draw opacity=1 ][line width=0.75]    (10.93,-3.29) .. controls (6.95,-1.4) and (3.31,-0.3) .. (0,0) .. controls (3.31,0.3) and (6.95,1.4) .. (10.93,3.29)   ;
\draw [color={rgb, 255:red, 208; green, 2; blue, 27 }  ,draw opacity=1 ] [dash pattern={on 4.5pt off 4.5pt}]  (232.4,41.6) -- (362.88,153.9) ;
\draw [shift={(364.4,155.2)}, rotate = 220.72] [color={rgb, 255:red, 208; green, 2; blue, 27 }  ,draw opacity=1 ][line width=0.75]    (10.93,-3.29) .. controls (6.95,-1.4) and (3.31,-0.3) .. (0,0) .. controls (3.31,0.3) and (6.95,1.4) .. (10.93,3.29)   ;
\draw [color={rgb, 255:red, 245; green, 166; blue, 35 }  ,draw opacity=1 ] [dash pattern={on 4.5pt off 4.5pt}]  (232.4,90.6) -- (362.6,154.32) ;
\draw [shift={(364.4,155.2)}, rotate = 206.08] [color={rgb, 255:red, 245; green, 166; blue, 35 }  ,draw opacity=1 ][line width=0.75]    (10.93,-3.29) .. controls (6.95,-1.4) and (3.31,-0.3) .. (0,0) .. controls (3.31,0.3) and (6.95,1.4) .. (10.93,3.29)   ;
\draw [color={rgb, 255:red, 80; green, 227; blue, 194 }  ,draw opacity=1 ] [dash pattern={on 4.5pt off 4.5pt}]  (231.4,167.6) -- (362.41,155.39) ;
\draw [shift={(364.4,155.2)}, rotate = 174.67] [color={rgb, 255:red, 80; green, 227; blue, 194 }  ,draw opacity=1 ][line width=0.75]    (10.93,-3.29) .. controls (6.95,-1.4) and (3.31,-0.3) .. (0,0) .. controls (3.31,0.3) and (6.95,1.4) .. (10.93,3.29)   ;
\draw [color={rgb, 255:red, 74; green, 144; blue, 226 }  ,draw opacity=1 ] [dash pattern={on 4.5pt off 4.5pt}]  (231.4,239.6) -- (362.71,156.27) ;
\draw [shift={(364.4,155.2)}, rotate = 147.6] [color={rgb, 255:red, 74; green, 144; blue, 226 }  ,draw opacity=1 ][line width=0.75]    (10.93,-3.29) .. controls (6.95,-1.4) and (3.31,-0.3) .. (0,0) .. controls (3.31,0.3) and (6.95,1.4) .. (10.93,3.29)   ;
\draw   (210,115.36) .. controls (210,113.51) and (211.21,112) .. (212.7,112) .. controls (214.19,112) and (215.4,113.51) .. (215.4,115.36) .. controls (215.4,117.22) and (214.19,118.73) .. (212.7,118.73) .. controls (211.21,118.73) and (210,117.22) .. (210,115.36) -- cycle ;
\draw   (210,127.1) .. controls (210,125.24) and (211.21,123.74) .. (212.7,123.74) .. controls (214.19,123.74) and (215.4,125.24) .. (215.4,127.1) .. controls (215.4,128.96) and (214.19,130.46) .. (212.7,130.46) .. controls (211.21,130.46) and (210,128.96) .. (210,127.1) -- cycle ;
\draw   (210,139.84) .. controls (210,137.98) and (211.21,136.47) .. (212.7,136.47) .. controls (214.19,136.47) and (215.4,137.98) .. (215.4,139.84) .. controls (215.4,141.69) and (214.19,143.2) .. (212.7,143.2) .. controls (211.21,143.2) and (210,141.69) .. (210,139.84) -- cycle ;

\draw (22,103) node [anchor=north west][inner sep=0.75pt]   [align=left] {Tranmsitter};
\draw (368,98) node [anchor=north west][inner sep=0.75pt]   [align=left] {Receiver};
\draw (192,2) node [anchor=north west][inner sep=0.75pt]   [align=left] {Relays};
\draw (207,31) node [anchor=north west][inner sep=0.75pt]   [align=left] {1};
\draw (208,81) node [anchor=north west][inner sep=0.75pt]   [align=left] {2};
\draw (207,157) node [anchor=north west][inner sep=0.75pt]   [align=left] {k};
\draw (204,230) node [anchor=north west][inner sep=0.75pt]   [align=left] {K};
\draw (270,54) node [anchor=north west][inner sep=0.75pt]    {$\textcolor[rgb]{0.82,0.01,0.11}{\gamma _{1}}$};
\draw (267,87) node [anchor=north west][inner sep=0.75pt]  [color={rgb, 255:red, 245; green, 166; blue, 35 }  ,opacity=1 ]  {$\textcolor[rgb]{0.96,0.65,0.14}{\gamma _{2}}$};
\draw (265,137) node [anchor=north west][inner sep=0.75pt]  [color={rgb, 255:red, 245; green, 166; blue, 35 }  ,opacity=1 ]  {$\textcolor[rgb]{0.31,0.89,0.76}{\gamma _{k}}$};
\draw (264,181) node [anchor=north west][inner sep=0.75pt]  [color={rgb, 255:red, 245; green, 166; blue, 35 }  ,opacity=1 ]  {$\textcolor[rgb]{0.29,0.56,0.89}{\gamma _{K}}$};
\draw (358,183) node [anchor=north west][inner sep=0.75pt]    {$\hat{k} =\arg \max_{i=1,\cdots ,K}(\textcolor[rgb]{0.82,0.01,0.11}{\gamma }\textcolor[rgb]{0.82,0.01,0.11}{_{1} ,}\textcolor[rgb]{0.96,0.65,0.14}{\gamma }\textcolor[rgb]{0.96,0.65,0.14}{_{2} ,} \cdots ,\textcolor[rgb]{0.29,0.56,0.89}{\gamma }\textcolor[rgb]{0.29,0.56,0.89}{_{K}})$};

\end{tikzpicture}
   \caption{DF relaying system, where relay that maximizes the E2E SNR transmits information to the receiver.}
    \label{fig_relay_sys}
\end{figure*}
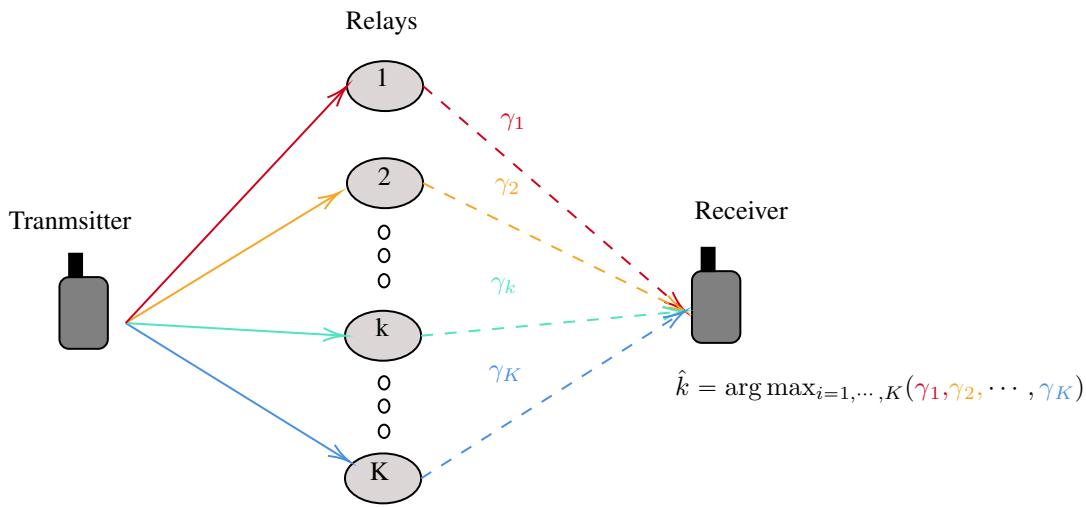
Relaying is a simple strategy capable of combating common impediments in wireless channels like fading, poor coverage, large power consumption, etc. When we consider node-to-node communication with single or multi-antenna nodes, there can be scenarios, where the signal power received at the destination may not achieve the desired \ac{QoS} constraints. In such scenarios, we may resort to using a relay node capable of decoding and forwarding or amplifying and forwarding the received signal to a destination node. The presence of more than one relay, with relay selection offers the benefits of diversity as well at the cost of extra hardware, computational complexity, and throughput reduction. Hence, relays have been demonstrated to improve the quality of wireless links \cite{marchenko2014experimental,madan2008energy,jing2009single,zhou2008energy}. \par The selection/combining of signals received over multiple relay links decides the performance of multi-relay systems. The ideal solution here depends on a trade-off between the allowable computational complexity and \ac{QoS} requirements. The simplest solution would be selecting the relay, which achieves the maximum performance in terms of the desired \ac{QoS} metric at the destination. This scheme would not require the aid of a signal combining at the destination and would give the best performance at the destination with a single relay node. Such single relay selection schemes are widely studied in the literature \cite{ikki2010performance,liu2014relay,belbase2018coverage,onalan2020relay}. Note that selecting the best or the $k$-th best relay in terms of received power/\ac{SNR}/\ac{SINR} requires the order statistics to characterize the system's performance. Several multi-hop relaying systems are also studied in literature \cite{shen2009multi,oyman2007multihop,cho2004throughput,karagiannidis2006bounds}, where a long distance low-quality link is broken into more than two high-quality links . 

\subsubsection*{System Description}Let us consider a single-hop relaying system with a single-antenna transmitter sending information to a single-antenna receiver via $K$ relays. We assume that the direct link between the transmitter and receiver is in outage. Hence, the relays are responsible for forwarding the information to the receiver either by amplifying and forwarding the received signal (referred to as the \ac{AF} network) or by decoding and re-transmitting the signal to the receiver (referred to as the \ac{DF} network). Furthermore, we assume that the specific relay which maximizes the desired performance metric of the \ac{E2E} link is selected for transmitting the information to the receiver. Hence, the maximum order statistic of the \ac{QoS} metric is essential for characterizing the system performance. We illustrate such a system in Fig. \ref{fig_relay_sys}, where $\gamma_k$ is the performance metric of interest, and the relay $\hat{k}$ with ${\gamma_{\hat{k}}}=\max\{\gamma_1,\cdots,\gamma_K\}$ is chosen for sending information to the receiver in the second hop of the relay network. Depending upon the particular relaying strategy (\ac{DF} or \ac{AF}, multi-hop or single-hop), the fading channel gain of the \ac{E2E} network will have a different distribution. For example, in a single-hop \ac{SISO} Rayleigh fading channel, a simple \ac{DF} relay network will result in a Rayleigh channel, whereas an \ac{AF} network will result in a double Rayleigh channel. Depending upon the channel's fading characteristics, the transmitter as well as the receiver nodes, and the relay selection strategy, we will have different distributions for the received signal power \cite{zhang2018performance,xia2012cooperative,manoj2018performance}.
\subsubsection*{Advantages of \ac{EVT}} The resultant expressions for the \ac{QoS} metrics are also different. Note that most of these are not available as simple expressions, not even for single relay channels. Hence, the exact order statistics of these metrics have complicated forms that are not amenable to analyses. \ac{EVT} allows us to characterize the statistics of these metrics in a simple form.
For example, Fig. \ref{relay_result} shows both the simulated \ac{CDF} and the asymptotic \ac{CDF} of the maximum received \ac{SNR} in a \ac{DF} relaying system.\footnote{The simulations consider the system model in \cite{subhash2021cooperative} and considers energy harvesting relays in a Rayleigh fading environment.} Observe that the asymptotic results are in good agreement with the exact CDF and hence can be used for simplifying the analysis of relaying systems.

\begin{figure}
    \centering
\includegraphics[scale=0.6]{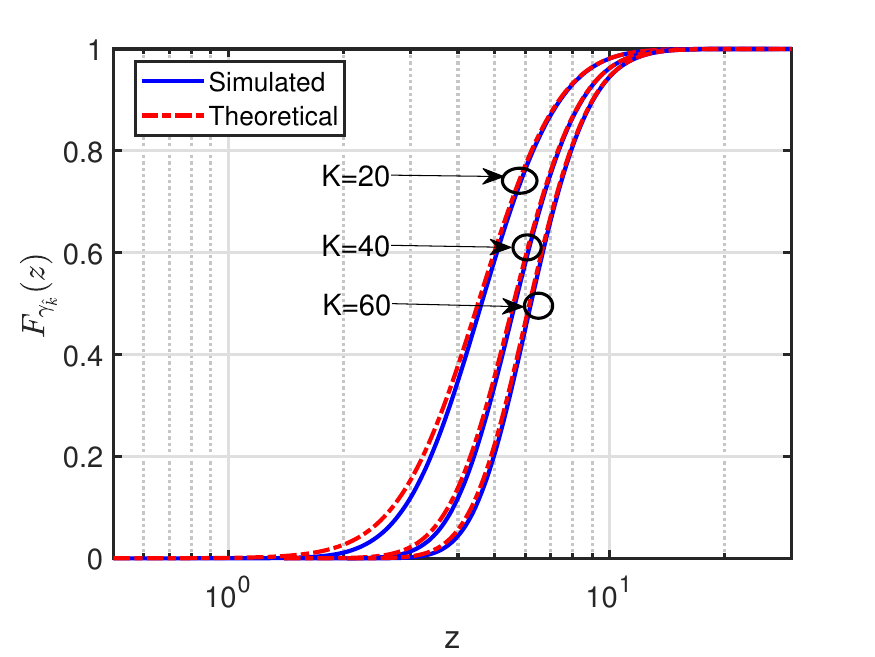}
    \caption{Simulated and theoretical CDF of $\gamma_{\hat{k}}=\max \left(\gamma_1,\cdots,\gamma_K \right)$}
    \label{relay_result}
\end{figure}
\subsubsection*{State of the Art}
\begin{table*}[b]
\centering
\caption{Applications of EVT in relay networks}\label{Table: relay}
 \resizebox{\textwidth}{!}{%
 \begin{tabular}{|l|l|l|l|l|l|}
 \hline
    Ref & Application~ & \ac{SNR}/\ac{SINR} &  Min/Max   & Asymptotic Distribution  \\
   \hline
    \cite{oyman2007opportunism} & Opportunistic Scheduling with relays       & Channel gain          & Maximum   & Gumbel                   \\
    \cite{oyman2008power} & Multi-hop relays       & Channel gain           & Maximum   & Gumbel                   \\
    \cite{biswas2015performance}  & millimeter wave with relays       & \ac{SNR}              & Maximum   & Gumbel                   \\
     \cite{liu2019cooperative} & Energy harvesting based relays       & \ac{SNR}         & Maximum   & Gumbel                   \\  
     \cite{xue2010mutual} & Relays       & Channel gain           & Maximum   & Gumbel                   \\
      \cite{xu2016cooperative} & Relays       & \ac{SNR}           & Maximum   & Gumbel                   \\
     
       \cite{xue2015performance} & Relays       & Channel gain           & Maximum   & Gumbel                   \\
    \hline
 \end{tabular}%
 }
\end{table*}
The tail of the received signal distribution will decide the asymptotic distribution of the received signal power. For example, Xia {\it et al.} \cite{xia2013spectrum
} study the limiting distribution functions of the lower and upper
bounds on the \ac{E2E} \ac{SNR} of a relaying path using \ac{EVT}. The upper and the lower bounds of \ac{SNR} are shown to follow the Gumbel and the Weibull distribution, respectively. \ac{EVT} is used to derive the 
scaling law of the system
throughput in a hybrid energy-harvesting aided relay system by Liu {\it et al.}  \cite{liu2019cooperative}. The maximum \ac{E2E} \ac{SNR} is proved to converge to the Gumbel distribution. This result is extended to the case of \ac{i.n.i.d.} relay links by Subhash {\it et al.} \cite{subhash2021cooperative}. Here also, the limiting distribution is demonstrated to be the Gumbel distribution for various special cases of the system model. Again, in contrast to the case of \ac{i.i.d.} \ac{RVs}, the limiting distribution need not be one of the three extreme value distributions for the most general case in this scenario. 

Oyman {\it et al.} \cite{oyman2007opportunism} analyze the spectral efficiency of opportunistic scheduling algorithms for fading multi-user relay channels. Using \ac{EVT}, the asymptotic channel power and the maximum spectral efficiency are proved to follow the Gumbel distribution. Using this result, Oyman {\it et al.}  \cite{oyman2008power} characterized the power-bandwidth efficiency tradeoff in multi-hop relaying systems.
The coverage probability in millimeter wave networks using relays is studied by Biswas {\it et al.} \cite{biswas2015performance}. Here, the authors use results from \ac{EVT} to derive the statistics of the maximum \ac{E2E} \ac{SNR}, which is in turn used for deriving a simple expression for the coverage probability. The asymptotic performance of a wireless-powered communication network following the $k$-th best device selection strategy is studied by Dimitropoulou {\it et al.} \cite{dimitropoulou2020k}. The best relay might not always be available for communication, and in many systems, it makes sense to assume that the $k$-th best relay node is used for communication. \ac{EVT} allows a convenient characterization of the $k$-th order statistics once the maximum/minimum order statistics have been established. \textcolor{black}{Table \ref{Table: relay}  highlights the diverse ways \ac{EVT} contributes to strengthening the robustness and efficiency of relay network operations}.

\subsubsection*{Summary/Takeaway}
All of the above studies demonstrate how \ac{EVT} has aided the analysis of various relay-based communication systems. Note that \ac{EVT} has been used in very different system models for different \ac{QoS} metrics. This explains the generality of the results from \ac{EVT} and how they can be applied in various applications. Furthermore, the asymptotic results established by using \ac{EVT} have very simple forms in all of the above treatises. They can be easily used for deriving various scaling laws as well to impose simple \ac{QoS} constraints for resource allocation problems. 

\subsection{Massive \ac{MIMO} and Cell-Free Massive \ac{MIMO}}
   \ac{MIMO} systems rely on multiple antennas at the \ac{BS} as well as \ac{MS} and exploit the signal's multipath propagation. As multiple antennas are available for communication, more than one signal can be transmitted simultaneously through the \ac{MIMO} system, increasing the system's data rate or sum rate. With the advent of \ac{MIMO}  technology, the sum-rate is boosted in the wireless communication system and \ac{EVT} may be used for analyzing it and its scaling laws.    
\begin{figure*} 
\centering
    \tikzset{every picture/.style={line width=0.75pt}} 
    \begin{tikzpicture}[x=0.75pt,y=0.75pt,yscale=-1,xscale=1]
    
    \draw [color={rgb, 255:red, 0; green, 0; blue, 0 }  ,draw opacity=1 ][line width=1.5]    (159.79,145.11) -- (117,238.97) ;
    \draw [color={rgb, 255:red, 0; green, 0; blue, 0 }  ,draw opacity=1 ][line width=1.5]    (159.79,145.11) -- (204.26,238.97) ;
    \draw [color={rgb, 255:red, 0; green, 0; blue, 0 }  ,draw opacity=1 ][line width=1.5]    (159.79,145.11) -- (162.31,248) ;
    \draw  [color={rgb, 255:red, 0; green, 0; blue, 0 }  ,draw opacity=1 ][line width=1.5]  (160.63,175.8) -- (174.9,180.61) -- (160.63,185.43) -- (146.37,180.61) -- cycle ;
    \draw  [color={rgb, 255:red, 0; green, 0; blue, 0 }  ,draw opacity=1 ][line width=1.5]  (160.21,190.24) -- (180.77,195.65) -- (160.21,201.07) -- (139.65,195.65) -- cycle ;
    \draw  [color={rgb, 255:red, 0; green, 0; blue, 0 }  ,draw opacity=1 ][line width=1.5]  (160.16,206.48) -- (187.99,211) -- (160.16,215.51) -- (132.33,211) -- cycle ;
    \draw  [color={rgb, 255:red, 0; green, 0; blue, 0 }  ,draw opacity=1 ][line width=1.5]  (161.68,223.33) -- (195.5,226.67) -- (160.68,230) -- (126.86,226.67) -- cycle ;
    \draw  [color={rgb, 255:red, 0; green, 0; blue, 0 }  ,draw opacity=1 ][fill={rgb, 255:red, 0; green, 0; blue, 0 }  ,fill opacity=1 ][line width=1.5]  (154.97,146.42) .. controls (154.97,144.6) and (157.13,143.11) .. (159.79,143.11) .. controls (162.46,143.11) and (164.62,144.6) .. (164.62,146.42) .. controls (164.62,148.25) and (162.46,149.73) .. (159.79,149.73) .. controls (157.13,149.73) and (154.97,148.25) .. (154.97,146.42) -- cycle ;
    \draw [color={rgb, 255:red, 0; green, 0; blue, 0 }  ,draw opacity=1 ][line width=1.5]    (159.79,154.11) -- (160.44,102.6) ;
    \draw [color={rgb, 255:red, 0; green, 0; blue, 0 }  ,draw opacity=1 ][line width=1.5]    (131.33,118.14) -- (188.04,88.25) ;
    \draw [color={rgb, 255:red, 0; green, 0; blue, 0 }  ,draw opacity=1 ][line width=1.5]    (188.04,88.25) -- (187.02,76.05) ;
    \draw [color={rgb, 255:red, 0; green, 0; blue, 0 }  ,draw opacity=1 ][line width=1.5]    (131.33,118.14) -- (130.31,105.94) ;
    \draw  [color={rgb, 255:red, 0; green, 0; blue, 0 }  ,draw opacity=1 ][line width=1.5]  (130.31,105.94) -- (124.54,95.02) -- (133.56,94.55) -- cycle ;
    \draw  [color={rgb, 255:red, 0; green, 0; blue, 0 }  ,draw opacity=1 ][line width=1.5]  (150.31,93.94) -- (144.54,83.02) -- (153.56,82.55) -- cycle ;
    \draw  [color={rgb, 255:red, 0; green, 0; blue, 0 }  ,draw opacity=1 ][line width=1.5]  (187.02,76.05) -- (181.25,65.13) -- (190.27,64.66) -- cycle ;
    \draw  [color={rgb, 255:red, 0; green, 0; blue, 0 }  ,draw opacity=1 ][fill={rgb, 255:red, 255; green, 255; blue, 255 }  ,fill opacity=1 ][line width=1.5]  (359.27,87.23) .. controls (359.27,85.45) and (360.71,84) .. (362.5,84) -- (372.2,84) .. controls (373.99,84) and (375.43,85.45) .. (375.43,87.23) -- (375.43,112.23) .. controls (375.43,114.02) and (373.99,115.47) .. (372.2,115.47) -- (362.5,115.47) .. controls (360.71,115.47) and (359.27,114.02) .. (359.27,112.23) -- cycle ;
    \draw [color={rgb, 255:red, 0; green, 0; blue, 0 }  ,draw opacity=1 ][line width=1.5]    (378.85,55.6) ;
    \draw [color={rgb, 255:red, 0; green, 0; blue, 0 }  ,draw opacity=1 ][line width=1.5]    (348.2,71.11) -- (378.85,55.6) ;
    \draw [color={rgb, 255:red, 0; green, 0; blue, 0 }  ,draw opacity=1 ][line width=1.5]    (348.2,71.11) -- (347.96,58.91) ;
    \draw [color={rgb, 255:red, 0; green, 0; blue, 0 }  ,draw opacity=1 ][line width=1.5]    (362.5,84) -- (361.34,64.23) ;
    \draw [color={rgb, 255:red, 0; green, 0; blue, 0 }  ,draw opacity=1 ][line width=1.5]    (378.85,55.6) -- (378.62,43.4) ;
    \draw  [color={rgb, 255:red, 0; green, 0; blue, 0 }  ,draw opacity=1 ][line width=1.5]  (347.96,58.91) -- (342.8,47.98) -- (352.02,47.52) -- cycle ;
    \draw  [color={rgb, 255:red, 0; green, 0; blue, 0 }  ,draw opacity=1 ][line width=1.5]  (361.1,52.02) -- (355.94,41.09) -- (365.16,40.64) -- cycle ;
    \draw  [color={rgb, 255:red, 0; green, 0; blue, 0 }  ,draw opacity=1 ][line width=1.5]  (378.62,43.4) -- (373.46,32.47) -- (382.67,32.02) -- cycle ;
    \draw [color={rgb, 255:red, 0; green, 0; blue, 0 }  ,draw opacity=1 ][line width=1.5]    (361.34,64.23) -- (361.1,52.02) ;
    \draw    (363.5,75) -- (385.5,75) ;
    
    \draw [color={rgb, 255:red, 0; green, 0; blue, 0 }  ,draw opacity=1 ][line width=1.5]    (151.33,106.14) -- (150.31,93.94) ;
    \draw    (139,135) -- (161,135) ;
    \draw  [color={rgb, 255:red, 0; green, 0; blue, 0 }  ,draw opacity=1 ][fill={rgb, 255:red, 255; green, 255; blue, 255 }  ,fill opacity=1 ][line width=1.5]  (358.77,233.23) .. controls (358.77,231.45) and (360.21,230) .. (362,230) -- (371.7,230) .. controls (373.49,230) and (374.93,231.45) .. (374.93,233.23) -- (374.93,258.23) .. controls (374.93,260.02) and (373.49,261.47) .. (371.7,261.47) -- (362,261.47) .. controls (360.21,261.47) and (358.77,260.02) .. (358.77,258.23) -- cycle ;
    \draw [color={rgb, 255:red, 0; green, 0; blue, 0 }  ,draw opacity=1 ][line width=1.5]    (378.35,201.6) ;
    \draw [color={rgb, 255:red, 0; green, 0; blue, 0 }  ,draw opacity=1 ][line width=1.5]    (347.7,217.11) -- (378.35,201.6) ;
    \draw [color={rgb, 255:red, 0; green, 0; blue, 0 }  ,draw opacity=1 ][line width=1.5]    (347.7,217.11) -- (347.46,204.91) ;
    \draw [color={rgb, 255:red, 0; green, 0; blue, 0 }  ,draw opacity=1 ][line width=1.5]    (362,230) -- (360.84,210.23) ;
    \draw [color={rgb, 255:red, 0; green, 0; blue, 0 }  ,draw opacity=1 ][line width=1.5]    (378.35,201.6) -- (378.12,189.4) ;
    \draw  [color={rgb, 255:red, 0; green, 0; blue, 0 }  ,draw opacity=1 ][line width=1.5]  (347.46,204.91) -- (342.3,193.98) -- (351.52,193.52) -- cycle ;
    \draw  [color={rgb, 255:red, 0; green, 0; blue, 0 }  ,draw opacity=1 ][line width=1.5]  (360.6,198.02) -- (355.44,187.09) -- (364.66,186.64) -- cycle ;
    \draw  [color={rgb, 255:red, 0; green, 0; blue, 0 }  ,draw opacity=1 ][line width=1.5]  (378.12,189.4) -- (372.96,178.47) -- (382.17,178.02) -- cycle ;
    \draw [color={rgb, 255:red, 0; green, 0; blue, 0 }  ,draw opacity=1 ][line width=1.5]    (360.84,210.23) -- (360.6,198.02) ;
    \draw    (363,221) -- (385,221) ;
    
    \draw  [dash pattern={on 4.5pt off 4.5pt}] (328,15) -- (492,15) -- (492,293) -- (328,293) -- cycle ;
    \draw    (415.5,292.75) -- (415.5,321.25) ;
    \draw    (121,321.25) -- (120.52,261.75) ;
    \draw [shift={(120.5,259.75)}, rotate = 89.53] [color={rgb, 255:red, 0; green, 0; blue, 0 }  ][line width=0.75]    (10.93,-3.29) .. controls (6.95,-1.4) and (3.31,-0.3) .. (0,0) .. controls (3.31,0.3) and (6.95,1.4) .. (10.93,3.29)   ;
    \draw  [dash pattern={on 4.5pt off 4.5pt}] (33.5,51.25) -- (219,51.25) -- (219,259.75) -- (33.5,259.75) -- cycle ;
    \draw    (134.5,103.25) -- (332.66,188.96) ;
    \draw [shift={(334.5,189.75)}, rotate = 203.39] [color={rgb, 255:red, 0; green, 0; blue, 0 }  ][line width=0.75]    (10.93,-3.29) .. controls (6.95,-1.4) and (3.31,-0.3) .. (0,0) .. controls (3.31,0.3) and (6.95,1.4) .. (10.93,3.29)   ;
    \draw    (157,88.75) -- (332.76,188.76) ;
    \draw [shift={(334.5,189.75)}, rotate = 209.64] [color={rgb, 255:red, 0; green, 0; blue, 0 }  ][line width=0.75]    (10.93,-3.29) .. controls (6.95,-1.4) and (3.31,-0.3) .. (0,0) .. controls (3.31,0.3) and (6.95,1.4) .. (10.93,3.29)   ;
    \draw    (190.5,73.75) -- (332.94,188.5) ;
    \draw [shift={(334.5,189.75)}, rotate = 218.85] [color={rgb, 255:red, 0; green, 0; blue, 0 }  ][line width=0.75]    (10.93,-3.29) .. controls (6.95,-1.4) and (3.31,-0.3) .. (0,0) .. controls (3.31,0.3) and (6.95,1.4) .. (10.93,3.29)   ;
    \draw    (134.5,103.25) -- (332.57,47.79) ;
    \draw [shift={(334.5,47.25)}, rotate = 164.36] [color={rgb, 255:red, 0; green, 0; blue, 0 }  ][line width=0.75]    (10.93,-3.29) .. controls (6.95,-1.4) and (3.31,-0.3) .. (0,0) .. controls (3.31,0.3) and (6.95,1.4) .. (10.93,3.29)   ;
    \draw    (157,88.75) -- (332.55,47.71) ;
    \draw [shift={(334.5,47.25)}, rotate = 166.84] [color={rgb, 255:red, 0; green, 0; blue, 0 }  ][line width=0.75]    (10.93,-3.29) .. controls (6.95,-1.4) and (3.31,-0.3) .. (0,0) .. controls (3.31,0.3) and (6.95,1.4) .. (10.93,3.29)   ;
    \draw    (190.5,73.75) -- (332.53,47.61) ;
    \draw [shift={(334.5,47.25)}, rotate = 169.57] [color={rgb, 255:red, 0; green, 0; blue, 0 }  ][line width=0.75]    (10.93,-3.29) .. controls (6.95,-1.4) and (3.31,-0.3) .. (0,0) .. controls (3.31,0.3) and (6.95,1.4) .. (10.93,3.29)   ;
    \draw  [dash pattern={on 0.84pt off 2.51pt}] (233,70.25) .. controls (233,63.07) and (235.35,57.25) .. (238.25,57.25) .. controls (241.15,57.25) and (243.5,63.07) .. (243.5,70.25) .. controls (243.5,77.43) and (241.15,83.25) .. (238.25,83.25) .. controls (235.35,83.25) and (233,77.43) .. (233,70.25) -- cycle ;
    \draw  [dash pattern={on 0.84pt off 2.51pt}] (293,165.25) .. controls (293,158.07) and (295.35,152.25) .. (298.25,152.25) .. controls (301.15,152.25) and (303.5,158.07) .. (303.5,165.25) .. controls (303.5,172.43) and (301.15,178.25) .. (298.25,178.25) .. controls (295.35,178.25) and (293,172.43) .. (293,165.25) -- cycle ;
    \draw    (121,321.25) -- (415.5,321.25) ;

    \draw    (43,124) -- (138,124) -- (138,149) -- (43,149) -- cycle  ;
    \draw (44,125) node [anchor=north west][inner sep=0.75pt]   [align=left] {Beam former};
    \draw (182,238) node [anchor=north west][inner sep=0.75pt]   [align=left] {\textbf{{\scriptsize BS}}};
    \draw (376,252.5) node [anchor=north west][inner sep=0.75pt]   [align=left] {{\scriptsize \textbf{MS-K}}};
    \draw (380.5,101.5) node [anchor=north west][inner sep=0.75pt]   [align=left] {{\scriptsize \textbf{MS-1}}};
    \draw (232.5,298) node [anchor=north west][inner sep=0.75pt]   [align=left] {CSI};
    \draw    (386,209) -- (466,209) -- (466,234) -- (386,234) -- cycle  ;
    \draw (387,210) node [anchor=north west][inner sep=0.75pt]   [align=left] {Combining};
    \draw    (386.5,63) -- (466.5,63) -- (466.5,88) -- (386.5,88) -- cycle  ;
    \draw (387.5,64) node [anchor=north west][inner sep=0.75pt]   [align=left] {Combining};
    \end{tikzpicture}
\caption{MIMO}
\label{fig:mimo}
\end{figure*}
\subsubsection*{System Description}
    Let us consider a downlink system model having $M$ transmit antennas at the \ac{BS} and $K$ \ac{MS}s. The \ac{MS}'s can have either a single antenna or $N$ antennas depending upon whether a \ac{MISO} or \ac{MIMO} system is considered. Let $h_{k}$ represent the channel coefficient vector between the \ac{BS} and $k$-th \ac{MS} and $s$ denotes the signal vector for transmission. Given a beamforming vector set $\left \{ \mathbf{w}_{m}:m=1,., M \right \}$, the transmitted signal $x$ can be represented by $x=\mathbf{w}_{m}^{T}s$. In addition, we denote the \ac{SNR}/\ac{SINR}  at the $k^{th}$ \ac{MS} as $\gamma_{k}$. One can select the link having the highest \ac{SNR}/\ac{SINR} for multi-user scheduling. For identifying the best beamformer, select $w_{m}$ for which $\gamma_{k}$ is the highest \ac{SNR}/\ac{SINR}.
\subsubsection*{Advantages of \ac{EVT}}
    As \ac{MIMO} systems can have many transmit and receive antennas, the feedback of full \ac{CSI} at the transmitter is complicated. Partial \ac{CSI}, like the best \ac{SNR}/\ac{SINR} link information, strikes an attractive trade-off. Linear combining techniques like \ac{SC}/\ac{MRC}/\ac{OC} will be employed at the receiver to improve the effective \ac{SINR} at the \ac{MS}. \ac{EVT} will help derive and analyze metrics such as the average sum rate, ergodic sum rate, and scaling laws. As shown by Lee {\it et al.} \cite{lee2012zero}, finding the distribution of \ac{SNR} using \ac{ZFBF} under a per antenna power constraint involves a complicated joint distribution of diagonal elements. The usage of \ac{EVT} helped convert such a complicated problem to the simple problem of finding the minimum of chi-square \ac{RVs}. \ac{RBF}, with other-cell interference, involves a complicated \ac{SINR} expression, as shown in \cite{{moon2011sum}}. Since finding the exact sum-rate expression would be difficult, the asymptotic expression can be derived with the help of extreme order statistics.  
\subsubsection*{Applicability Illustration}
    As \ac{MIMO} schemes involves more than one transmit and receive antenna, linear combining techniques will generally be applied at the receiver.
    Pun {\it et al.} \cite{pun2010performance} explored \ac{MIMO} schemes with optimum combining, where the effective \ac{SIR} is considered as a Pareto distribution. The limiting distribution of the maximum \ac{SIR} is proved to be the Frechet type.
    Additionally, different beamforming techniques such as \ac{ZFBF} and \ac{RBF} are often used in \ac{MIMO}. \ac{RBF} with inter-cell interference \cite{{moon2011sum}} involve complicated \ac{SINR} expressions relying on the chi-square and gamma \ac{RVs} ratio. In the case of massive \ac{MIMO} \cite{gao2017massive}, many antennas are expected at the transmitter and receiver. In such a scenario, antenna selection involves the maximum and the top $k$ maxima \ac{RVs}. Hence, finding the exact distribution of the top $k$ \ac{RVs} would be difficult; instead, one can carry out an asymptotic analysis using extreme order statistics for sub-array/full-array antenna selection. In a multicasting scenario, knowing the minimum rate at which information needs to be transmitted is required. Hence, extreme order statistics are quite helpful in this context. The chi-square RV with $n$ degrees of freedom is more often encountered in \cite{park2010capacity},  \cite{park2008capacity},  \cite{park2009outage}.
     \begin{table*}[b]
    \centering
    \caption{Applications of \ac{EVT} in \ac{MIMO} systems. }\label{Table: MIMO}
     \resizebox{\textwidth}{!}{%
     \begin{tabular}{|l|l|l|l|l|l|}
         \hline
           & Application~ & SNR/SINR &  Min/Max   & Asymptotic Distribution  \\
           \hline
          \cite{pun2010performance}         & MIMO-SDMA with Linear combining          & SINR,SIR           & Maximum     & Gumbel,Frechet       
               \\ 
          \cite{moon2011sum}  & MISO with RBF      & SINR            & Maximum  & Gumbel                 
          \\
          \cite{lee2012zero}  & MISO with ZFBF      & SNR           & Minimum  & Weibull                
          \\  
          \cite{huang2013random}  & MISO with RBF      & Channel gain            & Maximum  & Gumbel    
          \\
          \cite{kazemi2020analysis}   & Massive-MIMO      & E2E-SNDR           & Maximum   & Weibull, Gumbel                
          \\
          \cite{park2010capacity}  & Multicasting (MIMO)     & SNR             & Minimum   & Weibull   
          \\
          \cite{gao2017massive}  & Massive-MIMO      & Channel gain            & Maximum   & Gumbel 
           \\
          \cite{park2008capacity}  & Multicasting (MIMO)     & SNR           & Minimum   & Weibull
          \\
           \cite{park2009outage}  & Multicasting (MIMO)      & SNR             & Minimum   & Weibull
          \\
         \hline 
     \end{tabular}%
     }
   \end{table*}
\subsubsection*{State of the Art}
    MIMO systems having many transmit and receive antennas, getting full \ac{CSI}, to improve the sum-rate may not be feasible. Instead, partial side information like the best \ac{SINR} link of all links is very much useful in this context. Partial feedback like the max \ac{SINR} link in capacity analysis is used by Sharif {\it et al.} \cite{sharif2005capacity}. In \cite{pun2010performance}, \ac{EVT} is used to derive the average sum rates and their scaling laws for joint opportunistic scheduling and receiver design for \ac{MIMO}-\ac{SDMA} systems. A \ac{MIMO} system consisting of a \ac{BS} with $M$ transmit antennas and $K$ \ac{MS}s with $N$ antennas each is considered. Hence, $M$ orthonormal beamforming vectors are considered at the \ac{BS}, and the best beamformer will be selected in a particular time slot. Linear rake receivers are studied along with different linear combining techniques (\ac{SC}, \ac{MRC}, \ac{OC}) at the MS. In each time slot, the \ac{MS} will calculate each beam's effective \ac{SINR} and select the highest effective \ac{SINR} beam. The limiting distribution for the effective \ac{SIR} is derived and shown to be a Frechet-type distribution, while the limiting distribution for the effective \ac{SINR} is shown to be of the Gumbel type.
    \par Moon {\it et al.} \cite{moon2011sum} considered the \ac{MISO} downlink in the presence of other-cell interference. A \ac{BS} with $M$ transmit antennas and $K$ \ac{MS}s with a single antenna each are considered. \ac{RBF} is used at BS with $M_s(\leq M)$ number of orthonormal beamforming vectors. All the \ac{MS}s will calculate the \ac{SINR} for each beam and select the specific beam with the highest \ac{SINR}. The concept of maximum order statistics is used to derive the asymptotic closed-form expression for the ergodic sum rate assuming a large number of users. The authors have proved that the asymptotic sum rate follows the Gumbel distribution. Lee {\it et al.} \cite{lee2012zero} considered the multi-user \ac{MISO} downlink with \ac{ZFBF} technique and also considered the per antenna power constraint. A \ac{MISO} system consisting of a \ac{BS} having $M$ transmit antennas and $K$ \ac{MS}s, each with single antennas, is considered.
    First, they have proved that the received \ac{SNR} in this context is characterized as the minimum of  $M$ chi-square \ac{RVs} with parameter $M-K+1$. \ac{EVT} is used for finding the asymptotic sum rate of multi-user-\ac{MIMO} systems with \ac{ZFBF}. The limiting distribution of received \ac{SNR} is found to be of Weibull distribution.
    \par In \cite{huang2013random}, the \ac{MISO} system with random beamforming and selective feedback is analyzed by selecting the best beam from the feedback of best-L resource blocks. \ac{EVT} is used to examine multi-user diversity's randomness for the selective feedback considered. Gao {\it et al.} \cite{gao2017massive} used \ac{EVT} to derive the asymptotic upper capacity bounds for a massive \ac{MIMO} system under the assumption of \ac{i.i.d.} Rayleigh flat fading for both full-array and sub-array systems. In\cite{kazemi2020analysis}, \ac{EVT} is exploited to derive the average \ac{SNDR} of a massive \ac{MIMO} system with relays. The authors have considered the following two scenarios: a massive number of antennas at the source and a single antenna at the destination, and a massive number of antennas at both the source and destination. As the number of relays increases, the \ac{E2E} \ac{SNDR} of the system converges to Gumbel and Weibull distributions, respectively.
   
    Multicasting in wireless communications refers to simultaneously transmitting information to multiple users. The information rate for multicasting transmission will be the lowest rate among all users. In \cite{park2008capacity}, Park {\it et al.} studied the capacity limits of multicast channels in multi-antenna systems with the help of \ac{EVT}. Here the authors selected some subset of base station antennas that maximizes the minimum \ac{SNR} among the users.
    In \cite{park2009outage}, Park {\it et al.} studied the outage performance of the multi-antenna system using \ac{EVT}. They derived closed-form expressions for the limiting distribution of multicast channels assuming a large number of users. Furthermore, they derived upper and lower bounds for the outage probability. In \cite{park2010capacity}, Park {\it et al.} studied the performance of multiple antenna multicasting systems for spatially correlated channels.  They also derived the exact and limiting distribution of \ac{SNR} for multicasting systems.  In the case of many users, it is proved that the power should be uniformly allocated to the eigenvectors of the transmit channel covariance with non-zero eigenvalues.  \textcolor{black}{Table \ref{Table: MIMO} showcases a range of applications of \ac{EVT} in \ac{MIMO} systems.}  
\subsubsection*{Summary/Takeaway}
    \ac{MIMO} systems generally have multiple transmit and receive antennas. As we have multiple antennas at the receiver, different combining techniques  \cite{pun2010performance} can be applied for improving the effective \ac{SNR}/\ac{SINR}. Moreover, beam-forming techniques like, \ac{ZFBF}  \cite{lee2012zero} and \ac{RBF}  \cite{moon2011sum}, \cite{huang2013random} can be applied at the transmitter to improve the sum-rate. All these techniques lead to complicated \ac{SNR}/\ac{SINR} expressions. In some situations, deriving the exact expressions are difficult, forcing us to go for asymptotic expressions with the aid of \ac{EVT}. In the case of massive \ac{MIMO} \cite{kazemi2020analysis, gao2017massive} or multicasting scenarios \cite{park2009outage}, \cite{park2010capacity, park2008capacity}, transmit antenna subset selection involves modeling the distribution of top-$k$ \ac{SNR}/\ac{SINR} \ac{RVs}. Although simple Rayleigh fading is considered, we encounter more  \ac{RVs} like chi-square and Pareto due to the specific structure of MIMO schemes. \ac{EVT} is beneficial in deriving the performance metrics like average sum rate, scaling laws, and channel capacity bounds.

\subsection{URLLC Applications}

\begin{figure*}[h]
\center
\includegraphics[scale=0.5]
{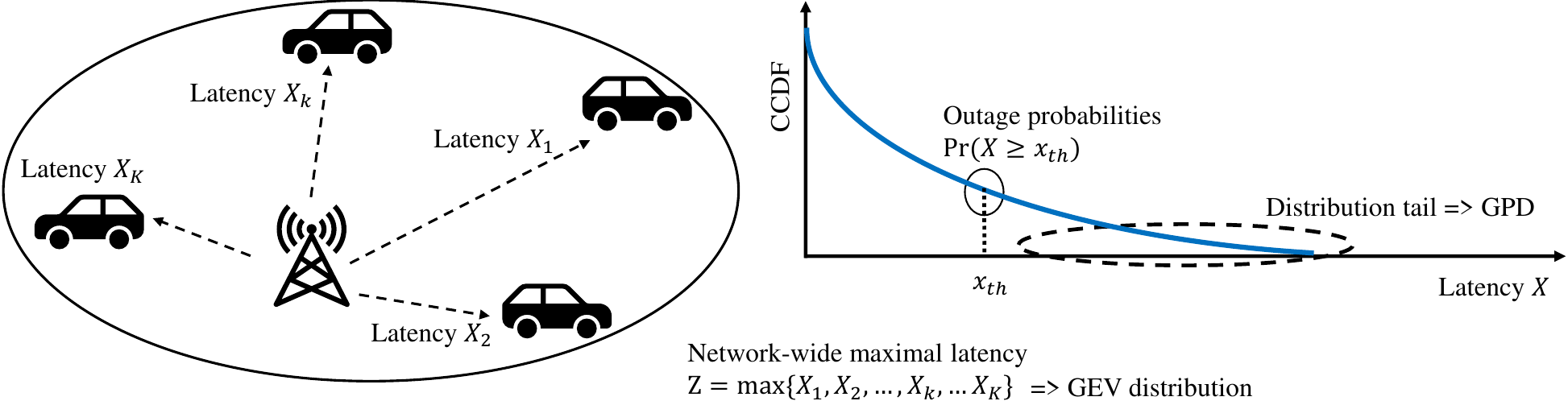}
\caption{URLLC illustration.}
\label{URLLC_model}
\end{figure*}

\begin{figure}[b]
\center
\includegraphics[scale=0.5]{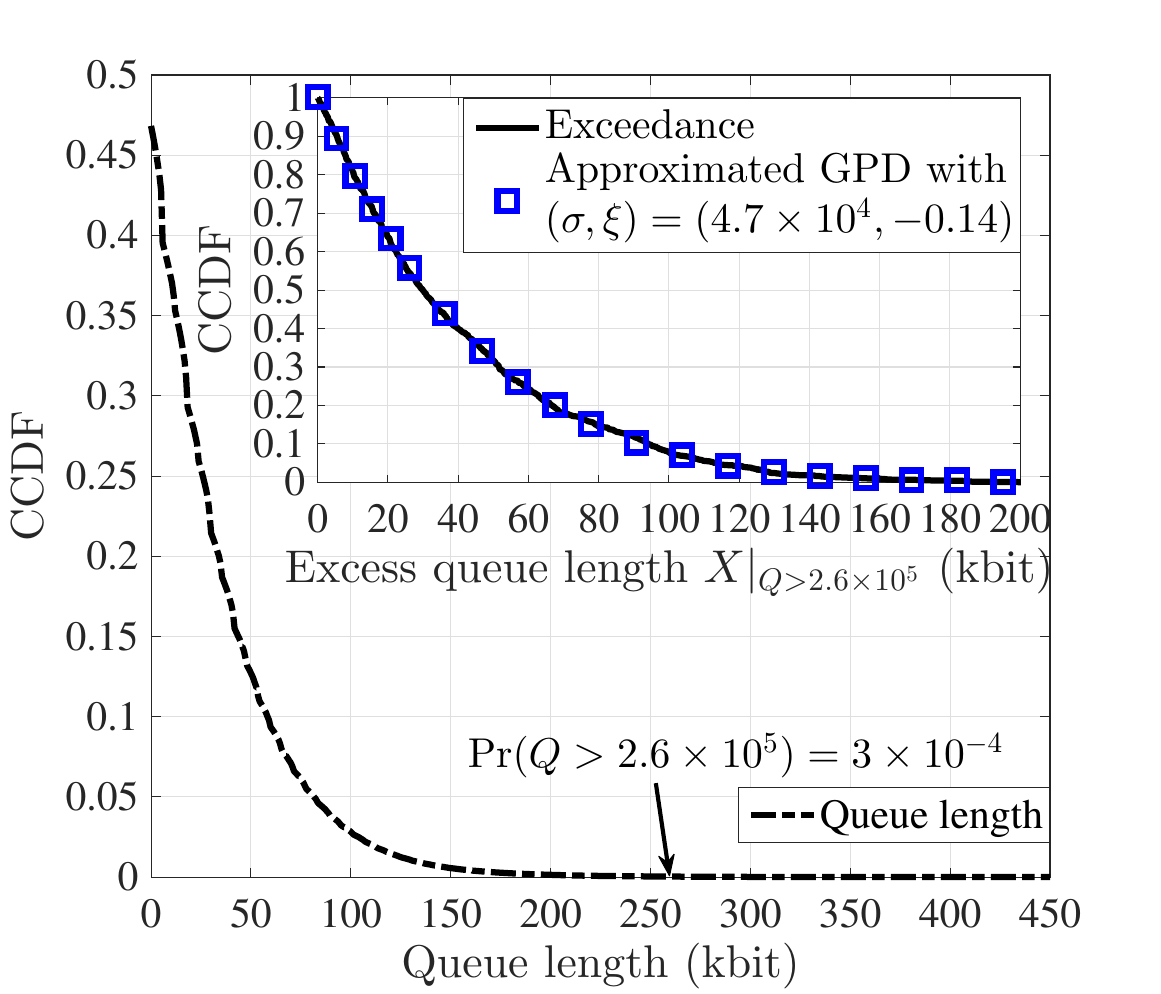}
\caption{Approximation of the delay distribution tail by the GPD \cite{LiuBenPoo17}}
\label{URLLC_model2}
\end{figure}
\begin{figure}
\center
\includegraphics[scale=0.5]
{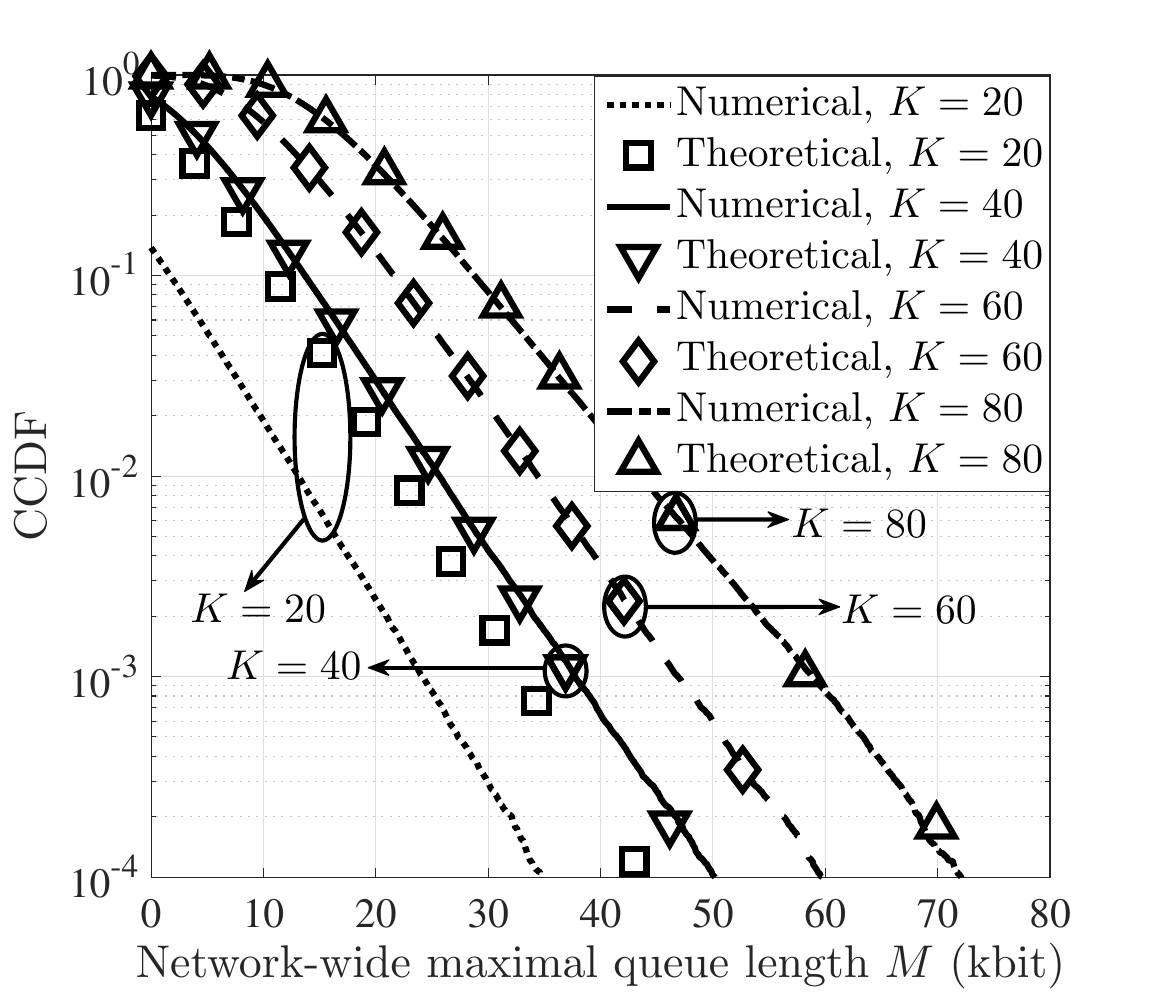}
\caption{Approximation of the maximal delay distribution by the GEV distribution \cite{liu2018ultra}.}
\label{URLLC_model3}
\end{figure}
As one of the pillars in 5G communication, \ac{URLLC} aims for ensuring timely information delivery for mission-critical and latency-sensitive applications, which have rather stringent \ac{E2E} delay requirements. In this regard, the \ac{3GPP} Release 15 specified the \ac{URLLC} requirement for the \ac{E2E} delay (measured at the ingress and egress points between the data link and network layers) of a small  packet, e.g., 32 bytes, as the 1 ms target with the outage probability less than $10^{-5}$ \cite{LiUusShaSin18}.
In addition to specifying the delay outage probability as the URLLC requirement, further analyzing the threshold-violating event, i.e., the delay outage part, and minimizing its effect help remedy the system performance degradation. For example, in immersive services, retaining the threshold-violating delayed data at the head-mounted display than plainly dropping them imposes lower degradation of the user's quality of experience.

The \ac{E2E} latency is also related to the maximal latency and can affect the \ac{AoI}. Its relationship with the maximal latency can be intuitively understood. In real life, the disasters such as car accidents and fire accidents usually occur in worst-case situations. Hence, investigating the performance degradation imposed by the maximal latency in the spatial (i.e., in the network) or temporal (i.e., over a time period) context can provide further insights for \ac{URLLC} system design.
In some applications of \ac{V2V} communication and of the industrial \ac{IoT}, up-to-date environment information is desired at the receiver regarding the safety concerns and production failures. In this situation, maintaining the freshness of the available information at the receiver is crucial for the transmitter even at the expense of using excessive communication resources. Explicitly, more frequent transmissions result in more up-to-date information, but dissipate more energy. We can evaluate the information freshness by the \ac{AoI} metric \cite{KauGruRaiKen11}, which is measured at the receiver and defined as the time elapsed since the data was generated until the current time instant. Briefly, when the receiver obtains the new data, the \ac{AoI} is equivalent to the \ac{E2E} latency. Then, before the next data is successfully received, the \ac{AoI} increases by the storage time. Since the \ac{E2E} latency is incorporated in the \ac{AoI}, the \ac{URLLC} performance directly affects the \ac{AoI}.

\subsubsection*{System Description}
Let us consider a system in which multiple transmitters send their own data to the respective receivers. For a single link, we denote the \ac{E2E} latency of a data packet as a random variable $X$. Then, the aforementioned \ac{URLLC} requirement (defined by \ac{3GPP}) can be mathematically expressed as
\begin{equation}\label{Eq: URLLC outage}
\mathbb{P}(X>x_{\rm th})\leq \epsilon,
\end{equation}
in which $x_{\rm th}$ and $\epsilon$ are the delay threshold and tolerable probability, respectively. Thus, with respect to the threshold $x_{\rm th}$, the threshold deviation\footnote{Threshold deviation refers to the difference between a measured value and a preset threshold.}  value of the \ac{E2E} delay $X$  is given by
\begin{equation}\label{Eq: threshold deviation}
Y|_{X>x_{\rm th}}=X-x_{\rm th}>0,
\end{equation}
whose mean, variance, and higher-order statistics require a dedicated focus as motivated above.
We further denote all transmitter-receiver links as a set $\mathcal{K}$ and index the delay of the $k$th transmitter-receiver link with the subscript $k$, i.e., $X_k$. The maximal latency over the network is accordingly defined as
\begin{equation}\label{Eq: URLLC maximum}
Z=\max _{k \in \mathcal{K}}\left\{X_k\right\}.
\end{equation}

Note that the \ac{AoI} (denoted by a function $A(t),\forall\,t\geq 0$) measures the age of the receiver's latest available information. Let us index the transmitter's sequentially-updated data by $n\in\mathbb{Z}^{+}$. When the successful reception of the $n^{th}$ data occurs at time instant $t^n$, the \ac{AoI} is the \ac{E2E} delay of the latest data, i.e., $A(t^n) = X^n$. Here, we add the superscript $n$ to the \ac{E2E} delay. Then before the next data is delivered, the \ac{AoI} increases linearly with time. The general \ac{AoI} function can be expressed as
\begin{equation}
A(t)=X^n+t-t^n,~\forall\,t\in[t^n,t^{n+1}),n\in\mathbb{Z}^{+}.
\end{equation}
Depending on the transmission and scheduling schemes of the network concerned, the \ac{E2E} delay $X^n$ may include the transmission delay, propagation delay, queuing delay, processing delay, computation delay, and/or other delay sources.

\subsubsection*{Advantages of \ac{EVT}}
The tolerable outage probability $\epsilon = 10^{-5}$ in the \ac{URLLC} regime is rather small, as specified previously. In some use cases, the \ac{URLLC} requirements are even more stringent, e.g., $\epsilon = 10^{-9}$ \cite{Yil16}. Given these values, we are mainly concerned about the distribution tail for the delay outage and threshold deviation. Additionally, the network density is expected to increase drastically. While studying the worst-case metric, we can focus on the asymptotic performance.
However, since the communication network becomes more complex, we may not be able to tractably find the tail distribution of the \ac{E2E} delay $X$ and to further derive the asymptotic distribution of the maximal delay $Z$ from $F_Z( z)=\prod_{k\in\mathcal{K}}F_{X_k}(z)$.
Since \ac{EVT} provides a powerful framework of characterizing the tail behavior of the generic distribution and of analyzing the asymptotic statistics of the maximum and minimum, we can resort to \ac{EVT} to address these dilemmas in the performance analysis and system design of \ac{URLLC}.

\subsubsection*{Applicability Illustration}

To characterize the tail distribution of an arbitrary random variable $X$ in \eqref{Eq: URLLC outage} by \ac{EVT}, we first stipulate a constant value $d\leq x_{\rm th}$ associated with $\mathbb{P}(X> d)= \delta\geq \epsilon$. Then we derive
\begin{align}\label{Eq: URLLC aux 1}
    \mathbb{P}(X>x_{\rm th})
    &=\frac{\mathbb{P}(X>d)\mathbb{P}(X>x_{\rm th},X>d)}{\mathbb{P}(X>d)}\nonumber\\
    &=\delta \mathbb{P}(X>x_{\rm th}|X>d).
\end{align}
Assume that for the selected $d$, the conditional value $Y|_{X>d}=X-d>0$ satisfies the Pickands\allowbreak–Balkema–de Haan theorem \cite{de2006extreme}, i.e.,
\begin{align}\label{Eq: URLLC aux 2}
 \bar{F}_{Y | X>d}(y)=\mathbb{P}(X-d>y |X> d)=\bar{G}(y;\sigma,\xi).
\end{align}
Note that the values of the characteristic parameters $\sigma$ and $\xi$ vary with the selected $d$. Thus, by incorporating \eqref{Eq: URLLC aux 1} and \eqref{Eq: URLLC aux 2}, the tail distribution of the delay can be approximated by a \ac{GPD} as follows:
\begin{align}
\mathbb{P}(X>x_{\rm th})\approx \delta \bar{G}(x_{\rm th}-d;\sigma,\xi),
\end{align}
even though the tractable derivation of the delay is not available. 
Moreover, threshold deviation \eqref{Eq: threshold deviation} can also be characterized by the statistics of the \ac{GPD} by setting $d$ as $x_{\rm th}$ in \eqref{Eq: URLLC aux 2}.
Regarding the maximal delay, for simplicity we assume that all links' delay statistics are identical. As the number of the network entities is ever increasing, the maximal delay $Z$ in \eqref{Eq: URLLC maximum} approaches the asymptotic regime, where the maximum can be characterized by a \ac{GEV} distribution by referring to the Fisher–Tippett–Gnedenko theorem \cite{de2006extreme}. 
Finally, the performance analysis and system design of \ac{URLLC} are based on the closed-form expressions of the statistics of the \ac{GPD} and \ac{GEV} distribution or their characteristic parameters.

\subsubsection*{State of the Art}

Mouradian \cite{mouradian2016extreme} studied the excess delays in wireless networks with the aid of \ac{EVT}. Specifically,  considering the block maxima and the peak over the threshold of the inter-beacon delay in vehicular ad hoc networks as a case study, the fitness between the \ac{GEV} distribution/\ac{GPD} and the empirical data samples was investigated.
Mehrnia and Coleri \cite{mehrnia2021wireless} proposed a new methodology for characterizing the extreme events of the wireless channel model in \ac{URLLC} communication based on \ac{EVT}. Given the sequence of received power samples, they will be converted to \ac{i.i.d.} samples by removing the dependency. Lower tail statistics are derived for the \ac{i.i.d.} received samples using the \ac{GPD}. The authors also derived the optimum threshold for characterizing the lower tail statistics and further specified the stopping conditions to find the minimum number of samples required for the tail estimation using \ac{GPD}.
In \cite{mehrnia2022extreme}, \ac{EVT} determines the transmit rate selection for \ac{URLLC}. Assuming that the receiver receives a packet in the presence of an unknown channel, the received power values at the receiver are used to fit the \ac{GPD} distribution to the channel's tail distribution. Using this, the optimum transmission rate is selected, which satisfies the \ac{URLLC} constraints.
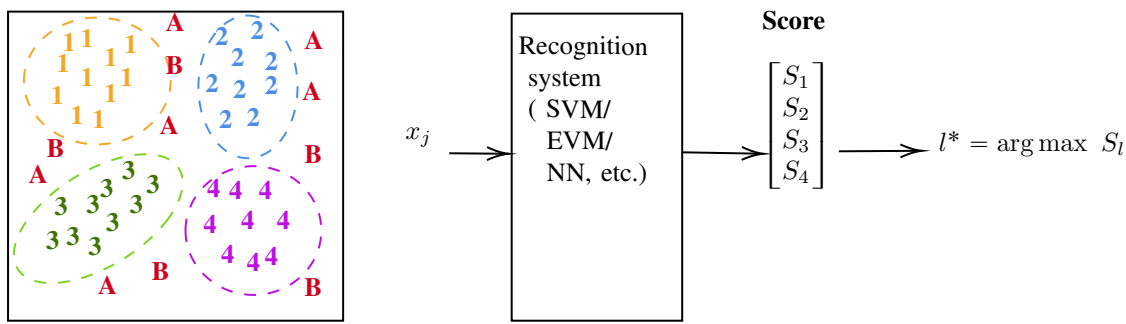
\begin{figure*}
    \centering
        \tikzset{every picture/.style={line width=0.75pt}} 
        \begin{tikzpicture}[x=0.75pt,y=0.75pt,yscale=-1,xscale=1]
        
        \draw   (26,89) -- (195,89) -- (195,245) -- (26,245) -- cycle ;
        \draw   (280,90) -- (366,90) -- (366,245) -- (280,245) -- cycle ;
        \draw    (366,160) -- (405,160.71) ;
        \draw [shift={(407,160.75)}, rotate = 181.05] [color={rgb, 255:red, 0; green, 0; blue, 0 }  ][line width=0.75]    (10.93,-3.29) .. controls (6.95,-1.4) and (3.31,-0.3) .. (0,0) .. controls (3.31,0.3) and (6.95,1.4) .. (10.93,3.29)   ;
        \draw    (445,159.5) -- (485,159.26) ;
        \draw [shift={(487,159.25)}, rotate = 179.66] [color={rgb, 255:red, 0; green, 0; blue, 0 }  ][line width=0.75]    (10.93,-3.29) .. controls (6.95,-1.4) and (3.31,-0.3) .. (0,0) .. controls (3.31,0.3) and (6.95,1.4) .. (10.93,3.29)   ;
        \draw    (248.5,160.5) -- (276.5,160.73) ;
        \draw [shift={(278.5,160.75)}, rotate = 180.48] [color={rgb, 255:red, 0; green, 0; blue, 0 }  ][line width=0.75]    (10.93,-3.29) .. controls (6.95,-1.4) and (3.31,-0.3) .. (0,0) .. controls (3.31,0.3) and (6.95,1.4) .. (10.93,3.29)   ;
        \draw  [color={rgb, 255:red, 74; green, 144; blue, 226 }  ,draw opacity=1 ][dash pattern={on 4.5pt off 4.5pt}] (147,91) .. controls (160.81,91) and (172,107.12) .. (172,127) .. controls (172,146.88) and (160.81,163) .. (147,163) .. controls (133.19,163) and (122,146.88) .. (122,127) .. controls (122,107.12) and (133.19,91) .. (147,91) -- cycle ;
        \draw  [color={rgb, 255:red, 245; green, 166; blue, 35 }  ,draw opacity=1 ][dash pattern={on 4.5pt off 4.5pt}] (34,124) .. controls (34,106.33) and (50.57,92) .. (71,92) .. controls (91.43,92) and (108,106.33) .. (108,124) .. controls (108,141.67) and (91.43,156) .. (71,156) .. controls (50.57,156) and (34,141.67) .. (34,124) -- cycle ;
        \draw  [color={rgb, 255:red, 126; green, 211; blue, 33 }  ,draw opacity=1 ][dash pattern={on 4.5pt off 4.5pt}] (106.8,193.5) .. controls (93.65,211.45) and (67.2,226) .. (47.7,226) .. controls (28.21,226) and (23.06,211.45) .. (36.2,193.5) .. controls (49.35,175.55) and (75.8,161) .. (95.3,161) .. controls (114.79,161) and (119.94,175.55) .. (106.8,193.5) -- cycle ;
        \draw  [color={rgb, 255:red, 189; green, 16; blue, 224 }  ,draw opacity=1 ][dash pattern={on 4.5pt off 4.5pt}] (115.5,199.5) .. controls (115.5,181.55) and (130.83,167) .. (149.75,167) .. controls (168.67,167) and (184,181.55) .. (184,199.5) .. controls (184,217.45) and (168.67,232) .. (149.75,232) .. controls (130.83,232) and (115.5,217.45) .. (115.5,199.5) -- cycle ;
        
        \draw (282,99.74) node [anchor=north west][inner sep=0.75pt]  [rotate=-0.34,xslant=-0.01] [align=left] {Recognition \\ \ system\\ \ ( SVM/\\ \ \ \ EVM/\\ \ \ \ NN, etc.)};
        \draw (405,113) node [anchor=north west][inner sep=0.75pt]   [align=left] {$\displaystyle \begin{bmatrix}
        S_{1}\\
        S_{2}\\
        S_{3}\\
        S_{4}
        \end{bmatrix}$};
        \draw (494,148) node [anchor=north west][inner sep=0.75pt]    {$l^{*} =\arg\max\ S_{l} \ $};
        \draw (225,146) node [anchor=north west][inner sep=0.75pt]    {$x_{j}$};
        \draw (49,109) node [anchor=north west][inner sep=0.75pt]   [align=left] {\textbf{\textcolor[rgb]{0.96,0.65,0.14}{1}}};
        \draw (129,95) node [anchor=north west][inner sep=0.75pt]   [align=left] {\textbf{\textcolor[rgb]{0.29,0.56,0.89}{2}}};
        \draw (124,119) node [anchor=north west][inner sep=0.75pt]   [align=left] {\textbf{\textcolor[rgb]{0.29,0.56,0.89}{2}}};
        \draw (154,120) node [anchor=north west][inner sep=0.75pt]   [align=left] {\textbf{\textcolor[rgb]{0.29,0.56,0.89}{2}}};
        \draw (145,136) node [anchor=north west][inner sep=0.75pt]   [align=left] {\textbf{\textcolor[rgb]{0.29,0.56,0.89}{2}}};
        \draw (131,137) node [anchor=north west][inner sep=0.75pt]   [align=left] {\textbf{\textcolor[rgb]{0.29,0.56,0.89}{2}}};
        \draw (137,106) node [anchor=north west][inner sep=0.75pt]   [align=left] {\textbf{\textcolor[rgb]{0.29,0.56,0.89}{2}}};
        \draw (154,108) node [anchor=north west][inner sep=0.75pt]   [align=left] {\textbf{\textcolor[rgb]{0.29,0.56,0.89}{2}}};
        \draw (147,92) node [anchor=north west][inner sep=0.75pt]   [align=left] {\textbf{\textcolor[rgb]{0.29,0.56,0.89}{2}}};
        \draw (138,122) node [anchor=north west][inner sep=0.75pt]   [align=left] {\textbf{\textcolor[rgb]{0.29,0.56,0.89}{2}}};
        \draw (65,201) node [anchor=north west][inner sep=0.75pt]   [align=left] {\textbf{\textcolor[rgb]{0.25,0.46,0.02}{3}}};
        \draw (43.8,198) node [anchor=north west][inner sep=0.75pt]   [align=left] {\textbf{\textcolor[rgb]{0.25,0.46,0.02}{3}}};
        \draw (71,172) node [anchor=north west][inner sep=0.75pt]   [align=left] {\textbf{\textcolor[rgb]{0.25,0.46,0.02}{3}}};
        \draw (54,196) node [anchor=north west][inner sep=0.75pt]   [align=left] {\textbf{\textcolor[rgb]{0.25,0.46,0.02}{3}}};
        \draw (82,163) node [anchor=north west][inner sep=0.75pt]   [align=left] {\textbf{\textcolor[rgb]{0.25,0.46,0.02}{3}}};
        \draw (74,189) node [anchor=north west][inner sep=0.75pt]   [align=left] {\textbf{\textcolor[rgb]{0.25,0.46,0.02}{3}}};
        \draw (64,180) node [anchor=north west][inner sep=0.75pt]   [align=left] {\textbf{\textcolor[rgb]{0.25,0.46,0.02}{3}}};
        \draw (48,181) node [anchor=north west][inner sep=0.75pt]   [align=left] {\textbf{\textcolor[rgb]{0.25,0.46,0.02}{3}}};
        \draw (86,180) node [anchor=north west][inner sep=0.75pt]   [align=left] {\textbf{\textcolor[rgb]{0.25,0.46,0.02}{3}}};
        \draw (94.37,170) node [anchor=north west][inner sep=0.75pt]   [align=left] {\textbf{\textcolor[rgb]{0.25,0.46,0.02}{3}}};
        \draw (125,172) node [anchor=north west][inner sep=0.75pt]   [align=left] {\textbf{\textcolor[rgb]{0.74,0.06,0.88}{4}}};
        \draw (142,188) node [anchor=north west][inner sep=0.75pt]   [align=left] {\textbf{\textcolor[rgb]{0.74,0.06,0.88}{4}}};
        \draw (136,173) node [anchor=north west][inner sep=0.75pt]   [align=left] {\textbf{\textcolor[rgb]{0.74,0.06,0.88}{4}}};
        \draw (145,209) node [anchor=north west][inner sep=0.75pt]   [align=left] {\textbf{\textcolor[rgb]{0.74,0.06,0.88}{4}}};
        \draw (160,188) node [anchor=north west][inner sep=0.75pt]   [align=left] {\textbf{\textcolor[rgb]{0.74,0.06,0.88}{4}}};
        \draw (154,206) node [anchor=north west][inner sep=0.75pt]   [align=left] {\textbf{\textcolor[rgb]{0.74,0.06,0.88}{4}}};
        \draw (132,205) node [anchor=north west][inner sep=0.75pt]   [align=left] {\textbf{\textcolor[rgb]{0.74,0.06,0.88}{4}}};
        \draw (123,190) node [anchor=north west][inner sep=0.75pt]   [align=left] {\textbf{\textcolor[rgb]{0.74,0.06,0.88}{4}}};
        \draw (151,173) node [anchor=north west][inner sep=0.75pt]   [align=left] {\textbf{\textcolor[rgb]{0.74,0.06,0.88}{4}}};
        \draw (73,124) node [anchor=north west][inner sep=0.75pt]   [align=left] {\textbf{\textcolor[rgb]{0.96,0.65,0.14}{1}}};
        \draw (56,135) node [anchor=north west][inner sep=0.75pt]   [align=left] {\textbf{\textcolor[rgb]{0.96,0.65,0.14}{1}}};
        \draw (52,98) node [anchor=north west][inner sep=0.75pt]   [align=left] {\textbf{\textcolor[rgb]{0.96,0.65,0.14}{1}}};
        \draw (81,116) node [anchor=north west][inner sep=0.75pt]   [align=left] {\textbf{\textcolor[rgb]{0.96,0.65,0.14}{1}}};
        \draw (46,125) node [anchor=north west][inner sep=0.75pt]   [align=left] {\textbf{\textcolor[rgb]{0.96,0.65,0.14}{1}}};
        \draw (61,116) node [anchor=north west][inner sep=0.75pt]   [align=left] {\textbf{\textcolor[rgb]{0.96,0.65,0.14}{1}}};
        \draw (73,106) node [anchor=north west][inner sep=0.75pt]   [align=left] {\textbf{\textcolor[rgb]{0.96,0.65,0.14}{1}}};
        \draw (61,96) node [anchor=north west][inner sep=0.75pt]   [align=left] {\textbf{\textcolor[rgb]{0.96,0.65,0.14}{1}}};
        \draw (67,137) node [anchor=north west][inner sep=0.75pt]   [align=left] {\textbf{\textcolor[rgb]{0.96,0.65,0.14}{1}}};
        \draw (83,99) node [anchor=north west][inner sep=0.75pt]   [align=left] {\textbf{\textcolor[rgb]{0.96,0.65,0.14}{1}}};
        \draw (35,165) node [anchor=north west][inner sep=0.75pt]   [align=left] {\textbf{\textcolor[rgb]{0.82,0.01,0.11}{A}}};
        \draw (44,152) node [anchor=north west][inner sep=0.75pt]   [align=left] {\textbf{\textcolor[rgb]{0.82,0.01,0.11}{B}}};
        \draw (104,88) node [anchor=north west][inner sep=0.75pt]   [align=left] {\textbf{\textcolor[rgb]{0.82,0.01,0.11}{A}}};
        \draw (101,140) node [anchor=north west][inner sep=0.75pt]   [align=left] {\textbf{\textcolor[rgb]{0.82,0.01,0.11}{A}}};
        \draw (174,97) node [anchor=north west][inner sep=0.75pt]   [align=left] {\textbf{\textcolor[rgb]{0.82,0.01,0.11}{A}}};
        \draw (173,123) node [anchor=north west][inner sep=0.75pt]   [align=left] {\textbf{\textcolor[rgb]{0.82,0.01,0.11}{A}}};
        \draw (174,222) node [anchor=north west][inner sep=0.75pt]   [align=left] {\textbf{\textcolor[rgb]{0.82,0.01,0.11}{B}}};
        \draw (104,110) node [anchor=north west][inner sep=0.75pt]   [align=left] {\textbf{\textcolor[rgb]{0.82,0.01,0.11}{B}}};
        \draw (174,155) node [anchor=north west][inner sep=0.75pt]   [align=left] {\textbf{\textcolor[rgb]{0.82,0.01,0.11}{B}}};
        \draw (97,214) node [anchor=north west][inner sep=0.75pt]   [align=left] {\textbf{\textcolor[rgb]{0.82,0.01,0.11}{B}}};
        \draw (405,87) node [anchor=north west][inner sep=0.75pt]   [align=left] {\textbf{Score}};
        \draw (70,221) node [anchor=north west][inner sep=0.75pt]   [align=left] {\textbf{\textcolor[rgb]{0.82,0.01,0.11}{A}}};
        \end{tikzpicture}
        \caption{OSR in Machine learning. 1,2,3 and 4 are known class samples for training the network. Along with the known class samples, Unknown classes like A and B will also appear in testing phase. OSR should identify these unknown classes apart from classifying the known classes. }
        \label{fig:OSR}
    \end{figure*}

Liu {\it et al.}~\cite{LiuBenPoo17,liu2019dynamic} focused their attention on task offloading as well as the transmit power and computation frequency allocation in a multi-user multi-server \ac{MEC} system, where the servers are deployed at the network edge to provide computation services for mobile users. Each user has two queuing buffers, one for local computation and one for offloading, respectively, while each server has the queuing buffers for the offloaded tasks from all users. By optimizing over each user's transmit power and computation frequency allocation along with each server's computation resource scheduling, the authors aimed for minimizing the users' transmit and computation power consumption subject to specific statistical constraints on each queue length. Specifically, the constraint on the queue length violation probability was first considered. Furthermore, taking into account the event of queue length violation, the authors characterized the violation events by the \ac{GPD} and imposed constraints on the \ac{GPD} parameters. The resource allocation and scheduling approaches relied on using Lyapunov optimization in \cite{LiuBenPoo17,liu2019dynamic}. Moreover, the authors considered the load balancing problem among servers and proposed a user-server association approach by resorting to matching theory \cite{liu2019dynamic}.
In \cite{zhou2020learning}, the task offloading problem is formulated by Zhou {\it et al.}~with the help of \ac{MAB} in the Internet of health things. To characterize the \ac{URLLC} constraint, the tail distribution of extreme events is approximated by \ac{GPD}.
In \cite{zhu2021reliability}, Zhu {\it et al.}~considered a multi-server \ac{MEC} network with the transmission/offloading of the information in the form of finite block-length codes. \ac{EVT} is applied in characterizing the extreme event of queue length violation.

Samarakoon {\it et al.}~\cite{samarakoon2019distributed} studied the transmit power allocation problem in \ac{V2V} communications in which the \ac{GPD} is used for characterizing the tail distribution of the queue length. The authors considered centralized learning and federated learning frameworks to characterize the \ac{GPD} parameters. In the centralized framework, all \ac{VUEs} send their local queue length samples to the \ac{RSU}. The \ac{RSU} learns the \ac{GPD} parameters using all samples and feeds back the parameters to all \ac{VUEs}. In the federated learning framework, each VUE learns the local \ac{GPD} parameters with the aid of its own queue length samples. Then, instead of the queue length samples, the \ac{VUEs} share their own learned parameters with the RSU. The RSU then finds the global \ac{GPD} parameters by aggregating all local parameters and sends the global parameters to all \ac{VUEs}.
Liu and Bennis \cite{LiuBen21} further considered the correlation of the data samples in federated learning for the \ac{GPD} parameters.

Considering the V2V network, Liu and Bennis \cite{liu2018ultra} took into account the network-wide maximal queue length among all \ac{VUEs} and studied the transmit power allocation problem. The objective of the optimization problem was to minimize the sum of all \ac{VUEs}' time-averaged transmit power subject to specific time-averaged constraints on the mean and variance of the maximal queue length. The authors proposed an \ac{EVT}-based power allocation approach in which first the \ac{GEV} distribution is leveraged for characterizing the maximal queue length. Then the estimation of the \ac{GEV} distribution parameters is aided by the Pickands–Balkema–de Haan theorem  \cite{de2006extreme}. Based on the estimated \ac{GEV} distribution, the \ac{VUEs} are allocated the transmit power if they have the chance to reach the maximal queue length.
Liu {\it et al.}~\cite{LiuWicSurBenDeb23} focused on the maximal transmission latency in the \ac{UAV} deployment problem. They characterized the maximal latency by the \ac{GEV} distribution and incorporated both \ac{EVT} and Gaussian process regression to predict the maximal latency statistics for any arbitrary position of the \ac{UAV}.

In\cite{AbdLiuSamBenSaa18,abdel2019optimized,HsuLiuSumWeiBen21, HsuLiuWeiBen22}, the authors minimized the transmit power, respectively, in \ac{V2V} communication and in industrial \ac{IoT}s subject to a probabilistic constraint on \ac{AoI} threshold violation events, which are further characterized by the \ac{GPD}.
Hsu {\it et al.}~\cite{HsuLiuSumWeiBen21,HsuLiuWeiBen22} further leveraged federated learning for \ac{GPD} characterization.
The maximal \ac{AoI} evaluated over a time period in industrial \ac{IoT} was considered in \cite{liu2019taming}. Therein, Liu and Bennis characterized the maximal \ac{AoI} as the \ac{GEV} distribution and focused their attention on taming the tail behavior of the maximal \ac{AoI}, while minimizing the sensor's transmit power and transmission block length.

\subsubsection*{Summary/Takeaway}

Since the tolerable delay outage probability is very small in the \ac{URLLC} regime, we focus on the tail distribution of the delay. Thus, the \ac{GPD} can be leveraged for characterizing the outage probability. Similarly, the \ac{GEV} distribution assists in characterizing the maximal delay within the network. Then we can design the \ac{URLLC} system around these characteristic parameters.

\subsection{Extreme Value Theory and Machine Learning}
    In traditional machine learning classification tasks, the focus is on closed-set recognition, where both training and testing samples come from known or seen classes. However, in real-world applications, classification tasks may involve samples from unknown or unseen classes, a scenario known as \ac{OSR}. The core challenge in \ac{OSR} is to correctly identify samples from these unknown/unseen classes during testing. \ac{EVT} is applied in the analysis of post-recognition scores. Moreover, the \ac{MAB} framework involves selecting the arm with the highest reward to maximize cumulative gain over multiple trials. \ac{EVT} is useful for analyzing the distribution of maxima and minima in the \ac{MAB} context.

\subsubsection*{System Description}
    The dataset provided to the recognition system includes both known and unknown classes. Only a subset of classes is used during training, while the remaining classes are reserved for testing. As illustrated in Fig.~\ref{fig:OSR}, classes 1, 2, 3, and 4 are treated as known classes for training, while classes A and B are considered unknown. In closed-set recognition, training and testing are performed using only the known classes. In contrast, \ac{OSR} systems are tested with both known and unknown classes. The recognition system can be implemented using various algorithms in \ac{ML} such as \ac{SVM}, \ac{EVM}, or \ac{NN}, which output scores or probabilities  \cite{scheirer2014probability}, \cite{yu2021deep} for a test sample relative to all the trained classes. The test sample is then assigned to the class with the highest score. Some models also account for reconstruction errors \cite{zhang2016sparse},  \cite{oza2019c2ae}, \cite{yang2021conditional} or distance distributions \cite{rudd2017extreme}, \cite{vignotto2020extreme} to identify unknown classes.
    
    
\subsubsection*{Advantages of \ac{EVT}}  
    In general, if the test sample is taken from the trained classes, the score of the class to which it belongs would be higher (match score), and the rest of the classes' scores (non-match score) would be lower. If an unknown class sample is taken for testing, the score for all the trained classes should be very low. Hence, an appropriate threshold on this score will differentiate the unknown from a known class. One can observe that in the distribution of scores, the unknown class sample either falls in the upper tail of the non-match score or the lower tail of the match score. Thus, the \ac{GPD} or \ac{GEV} distributions from \ac{EVT} help analyze the score distributions to identify the unknown classes in the \ac{OSR} context. As we can identify the unknown classes during the testing phase, one can update/retrain the models (incremental learning) with the unknown classes. In this case, retraining is carried out by considering all the samples of the new classes and only a few percentages of samples of previously trained classes. 
\subsubsection*{Applicability Illustration}
    The extreme values of a score distribution obtained by any recognition task can be modeled as the reverse Weibull distribution/Weibull distribution depending upon the data boundedness. If the data is bounded above, it can be modeled as a reverse Weibull distribution. If the data is bounded below, it can be modeled as a Weibull distribution. The lower tail of matched scores can be modeled as Weibull distribution, and the upper tail of non-match scores can be modeled as reverse Weibull distribution \cite{scheirer2014probability}. In \cite{rudd2017extreme}, \cite{vignotto2020extreme}, margin distances were modeled as Weibull/\ac{GEV} distributions in proposing new classifiers like \ac{EVM} and \ac{GEV}. As the \ac{GPD} can model the tail distribution, reconstruction errors in \cite{zhang2016sparse}, \cite{oza2019c2ae}, \cite{yang2021conditional}  were modeled as \ac{GPD}. In the \ac{MAB}s scenario, each arm should give the maximum. Therefore, assuming the Gumbel/\ac{GEV} distribution for each arm is more suitable as mentioned in \cite{cicirello2005max}, \cite{streeter2006asymptotically}.  
\subsubsection*{State of the Art}
   Initially, \ac{EVT} was applied to post-recognition score analysis \cite{scheirer2010robust}, \cite{scheirer2011meta}, where the lower tail of the match score was modeled as a Weibull distribution, and the upper tail of the non-match score as a Reverse Weibull distribution. However, \ac{OSR} was not considered in these studies. Scheirer \textit{et al.} proposed an \ac{OSR} model having a \ac{CAP}, where the probability of class membership decreases as we move from known to unknown classes. Using \ac{EVT}, they introduced a  Weibull-calibrated \ac{SVM} (W-SVM), where training samples are split into match samples (for class $y$, $x \in K^{+}$) and non-match samples (for class $\neq y$, $x \in K^{-}$). For a given sample $x_{i} \in K = K^{+} \cup K^{-}$, the SVM decision score $s_{i}$ is computed. Match and non-match sample scores are collected separately, with non-match scores modeled as a reverse Weibull distribution and match scores modeled as a Weibull distribution. During testing, probabilities      
   \begin{equation}
      \mathbb{P}_\psi(y \mid f(x))=e^{-\left(\frac{f(x)-v_\psi}{\lambda_\psi}\right)^{\kappa_\psi}},
    \end{equation}
    \begin{equation}
      \mathbb{P}_\eta(y \mid f(x))=1-e^{-\left(\frac{-f(x)-v_\eta}{\lambda_\eta}\right)^{\kappa_\eta}}, 
    \end{equation}
     are calculated for each test sample.
    They used these two estimates to determine the probability of class membership. The product $\mathbb{P}_\psi \times \mathbb{P}_\eta$ represents ``the probability that the input belongs to the match class and is NOT from any known non-match classes''.

    \par Zhang and Patel \cite{zhang2016sparse} introduced an \ac{OSR} method using \ac{SRC}, where reconstruction errors for test samples are smaller for matched classes than for non-matched ones. They applied \ac{EVT} to model the tail distribution of reconstruction errors for matched classes. The authors also found that sparse coefficients for open-set examples differ significantly from closed-set ones. Analyzing the non-matched reconstruction error sums revealed distinct distributions for closed-set and open-set samples. To enhance \ac{OSR} performance, they combined the right tail of the matched class distribution and the left tail of the non-matched distribution using \ac{GPD} to identify unknown classes.       
    \begin{table*}[h]
        \centering
        \caption{Applications of \ac{EVT} in \ac{ML}. }\label{Table: ML}
         \resizebox{\textwidth}{!}{%
         \begin{tabular}{|l|l|l|l|l|}
         \hline
          Ref & Application~ & Score/Reconstruction error   & Weibull/GEV/GPD  \\
           \hline
         \cite{scheirer2014probability}  & Weibull- calibrated SVM (W-SVM)      & Score     & Weibull/Reverse Weibull   \\
        \cite{zhang2016sparse} & Sparse-representation       & Reconstruction error           & GPD                 \\
        \cite{bendale2016towards} & Open-Max algorithm    & Distance            &  Weibull \\
          \cite{oza2019c2ae},\cite{yang2021conditional} & Class-conditioned autoencoder      & Reconstruction error          &  GPD \\
         \cite{ni2021open}  & Deep discriminative representation network       & Distance              & Reverse Weibull \\
         \cite{yu2021deep} & CNN        & Score          & Weibull \\
          \cite{rudd2017extreme} & EVM       & Distance          & Weibull \\
          \cite{vignotto2020extreme} & GPD and GEV Classifiers       & Distance          & GPD and GEV \\
          \hline
         \end{tabular}%
         }
    \end{table*}
    
    \par As a further advance, Bendale and Boult \cite{bendale2016towards} introduced the Open Max algorithm to reject open-set samples. They used values from the penultimate layer of deep networks as \ac{AVs} and proposed an \ac{EVT}-based meta-recognition algorithm to detect outliers. Each class is represented by a \ac{MAV}, computed as the mean of correctly classified samples. For an input image, the distances between the \ac{MAV} and correctly classified samples are calculated. Weibull fitting is applied to the largest distance, which is then used to estimate the probability of outlier detection.    
    
    
    \par Oza and Patel \cite{oza2019c2ae} proposed an \ac{OSR} method using a class-conditioned autoencoder with a two-stage training process: closed-set and open-set. Closed-set training involves a general encoder and classifier with traditional loss functions, while open-set training focuses on training the decoder by freezing the autoencoder's parameters. The decoder minimizes reconstruction errors for matched labels and poorly reconstructs non-matched labels. Reconstruction errors are modeled with \ac{EVT}, where matched and non-matched errors follow a \ac{GPD}. During testing, if the reconstruction error is below a threshold, the sample is considered a known class; otherwise, it is labeled as unknown.
    

    \par  Bridge {\it et al.} \cite{bridge2020introducing} proposed the \ac{GEV} activation function, based on \ac{EVT}, to improve performance on unbalanced data. Unlike the sigmoid function, the \ac{GEV} activation, implemented as a trainable layer, enhanced results in tasks like COVID-19 image classification.
    

    \par Ni and Huang \cite{ni2021open} used a \ac{DDRN} to embed human gait features, applying \ac{EVT} to find the \ac{CIP}. Known identities cluster in the embedding space, while unknown ones are distant. They modeled class-specific distances with reversed Weibull distribution and used a threshold to identify unknown samples.
    

    Yang {\it et al.} \cite{yang2021conditional} developed an intrusion detection system using class-conditioned autoencoders, where reconstruction errors for known attacks are smaller than those for unknown ones. They applied \ac{EVT} to model the distribution of these errors, identifying unknown attacks by comparing the test sample’s reconstruction error to the tail of the training error distribution, modeled by \ac{GPD}.
    
    \par  As a further development, Yu {\it et al.} \cite{yu2021deep} proposed a framework for \ac{OSFD}, addressing unpredictable fault modes during testing. \ac{OSFD} is split into \ac{SOSFD} and \ac{COSFD}, where \ac{SOSFD} uses a 1-D \ac{CNN} trained on identical training and test distributions, while \ac{COSFD} employs a bilateral weighted adversarial network for different distributions. Recognition scores from these models are fitted to a Weibull distribution using \ac{EVT}. The probability of a test sample belonging to a known class is calculated using the Weibull \ac{CDF}, with samples surpassing a threshold classified as unknown.   
    

    \par Gong {\it et al.} \cite{gong2021multi} proposed a deep neural network for open-set wireless signal identification, consisting of three parts: \ac{MC}, an improved counterfactual \ac{GAN}, and an \ac{OSC}. The \ac{MC}, with a \ac{CNN} feature extractor, learns general input representations. After training, the classification error is modeled with \ac{EVT} to identify boundary samples, which are then fed into the counterfactual \ac{GAN} to generate synthetic open-set samples. The \ac{OSC} uses both synthetic and original samples to perform \ac{OSR}.    
    

    \par Rudd {\it et al.} \cite{rudd2017extreme} introduced the EVM classifier, leveraging EVT concepts. For each training sample $(x_{i}, y_{i})$, the margin distance to the closest sample from a different class is calculated as $m_{ij} = ||x_i - x_j|| / 2$. The authors focus on finding the maximum margin distance and its distribution. Using the margin distribution theorem, they showed that the distribution of minimal margin distances follows a Weibull distribution. Based on this, the probability that an unknown sample $x'$ belongs to the boundary of $x_i$ is:
   \begin{equation}
        \Psi (x_i,x':k_i,\lambda_i) =e^{-\left (\frac{\left \| x_i-x' \right \|}{\lambda_i}  \right )^{k_i}}.
    \end{equation}
   The probability that sample $x'$ belongs to class $\mathcal{C}_l$ is $\hat{\mathbb{P}}(\mathcal{C}_l|x')=\arg\max_{i:y_i=\mathcal{C}_l}\Psi (x_i,x';k_i,\lambda_i)$. The final classification decision is given by
    \begin{equation}
        y^* = \begin{cases}
      \arg\max_{l\in\{1,\dots,M\}}\hat{\mathbb{P}}(\mathcal{C}_l|x'), & \hat{\mathbb{P}}(\mathcal{C}_l|x') \geq \delta, \\
      \text{Unknown}, & \text{otherwise},
    \end{cases}
    \end{equation}
    where $\delta$ is a threshold for classifying known versus unknown classes.

    \par Vignotto and Engelke \cite{vignotto2020extreme} introduced the GPD and GEV classifiers for anomaly detection based on \ac{EVT}. In the \ac{GPD} classifier, distances $D_i = |x_i - x_0|$ between a new sample and training data are fitted to a distribution. If the lower tail matches the normal class, $x_0$ is normal; otherwise, it's abnormal. The shape parameter $\xi$ is estimated using Hill's estimator and compared to a threshold.
    The nearest neighbor distance $D_i^{\min}$ is computed for each training sample in the \ac{GEV} classifier. A negative \ac{GEV} distribution is fitted, and if ${\mathbb{P}}(-D^{\min} < -d_0^{\min})$ is small, $x_0$ is abnormal.

    \begin{figure}[h]
        \centering
        \includegraphics[scale=0.6]{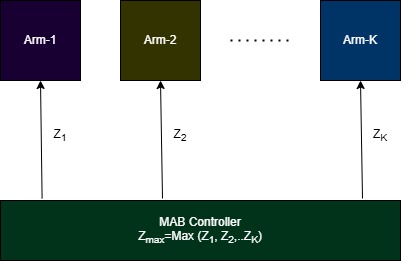}
        \caption{\ac{MAB} with \ac{EVT} }
        \label{fig:mab}
    \end{figure}
    The authors proposed a new approach for outlier/anomaly detection in time series data with the help of \ac{EVT} in \cite{siffer2017anomaly}. Dang \textit{et al.} \cite{dang2019open} introduced the \ac{OSmIL} framework, which consists of two parts. In the first part, an exemplar selection method is used to identify the boundary samples for each class, representing the class with these edge samples. In the second part, these edge samples are used for training and testing with the \ac{EVM} algorithm proposed by Rudd {\it et al.} \cite{rudd2017extreme}.
    Wilson {\it et al.} \cite{wilson2022deepgpd} focused on modeling Geospatio-Temporal data to predict extreme events using deep learning. They applied the \ac{GPD} to model excesses over a threshold, which fits well with the data. The parameters of the \ac{GPD} are learned through deep learning techniques, while \ac{GAN} are used to generate artificial data samples that match a known distribution.    
    Authors \cite{bhatia2021exgan} proposed ExGAN, known as Extreme-GAN, for generating extreme samples of known distribution with the help of \ac{EVT}.  \textcolor{black}{Table \ref{Table: ML} presents different ways in which \ac{EVT} is applied in \ac{ML} scenarios.}  
     \par The framework of \ac{MAB} allocates limited resources to multiple competing options to maximize the cumulative gain. In a one-armed bandit, pulling the lever gives a specific reward. Consider a slot machine with $n$ arms, where each arm has its own unknown prior distribution. Each arm's selection has its merits/demerits, so the challenge here is to select the arms sequentially, one by one, to maximize the total cumulative gain. Hence, the best arm may be selected from the set of  available arms in each time slot. Here the concept of exploration and exploitation is used, where the first few trials are used to explore which arms are giving the maximum rewards. The arms selected after the exploration will be exploited in the subsequent trials to maximize the cumulative gain.

     \ac{MAB} problems involve selecting arms to maximize cumulative rewards over time, with each arm having an unknown reward distribution. \ac{EVT} is applied in \ac{MAB} to model and analyze the tail behavior of reward distributions, focusing on identifying the maximum rewards across multiple trials. By incorporating \ac{EVT}, \ac{MAB} algorithms can better estimate the extreme outcomes, optimizing decision-making in uncertain environments. \textcolor{black}{As illustrated in Fig.~\ref{fig:mab}, the \ac{MAB} problem involves selecting from multiple arms to maximize cumulative rewards, balancing exploration (trying new arms) and exploitation (choosing the best-known arm).}
    \par Cicirello and Smith \cite{cicirello2005max} proposed a new variation of \ac{MAB} called Max K-armed bandit with the idea to allocate the trials to maximize the expected best single sample reward. Here the authors assumed that each arm draws a sample from the Gumbel distribution. To maximize the expected max single sample reward in $N$ trials, the authors proved that the best arm observed should be sampled at a rate increasing double exponentially relative to the number of samples given to the other $k-1$ arms.
    Streeter and Smith \cite{streeter2006asymptotically} presented an asymptotically optimal algorithm for the max-k armed bandit problem. Here, the authors considered the issue of allocating the trials to a set of $k$ slot machines, where each gives its reward from an unknown distribution. They considered each arm reward to be drawn from a \ac{GEV} distribution with unknown parameters. Consider the $i^{th}$ arm and pull it $n$ times, furthermore let $m_{n}^{i} $ denote the expected maximum reward for $i^{th}$ arm. 
    The authors explored the samples of the $i^{th}$ arm to identify the estimate of $m_{n}^{i} $ as $\widehat{m}_{n}^{i} $ along with property that 
    \begin{equation}
       \mathbb{P}\left [ \left | \widehat{m}_{n}^{i}-m_{n}^{i} \right |< \epsilon  \right ]\geq 1-\delta .
    \end{equation}
    In exploitation $\widehat{i}$ will be selected for the remaining trials, where we have    
    \begin{equation}       
         \widehat{i}= \underset{1\leq i\leq k}{\operatorname{argmax}} \,     
        \widehat{m}_{n}^{i}.
    \end{equation}
    The authors also derived the so called \ac{PAC} bounds on the sample complexity of estimating $m_{n}$.
    Carpentier and Valke \cite{carpentier2014extreme} proposed the Extreme Hunter algorithm for extreme bandits and provided their analysis. The classical bandits focus on the arm with the highest mean, but the extreme bandits focus on the arm with the heaviest tail. Given k-armed bandits, in each timestamp $t$, a sample $X_{k,t}$  is emitted from the $k^{th}$ arm with an unknown probability distribution $P_{k}$. The learner selects the arm $I_{t}$, and the corresponding sample is $X_{I_{t},t}$. The reward for the learner is defined as 
    \begin{equation}
      G_{n}^{L}=\max_{t\leq n}X_{I_{t},t}.
    \end{equation}
    The optimal strategy is to select the arm with the heaviest tail, and the expected reward is 
    \begin{equation}
    \mathbb{E}\left[G_n^*\right]=\max _{k \leq K} \mathbb{E}\left[\max _{t \leq n} X_{k, t}\right].
    \end{equation}
    They have considered  a second-order Pareto distribution for  $P$  and defined the terminology extreme regret as
    \begin{align}    \mathbb{E}\left[R_n^L\right]=&\mathbb{E}\left[G_n^*\right]-\mathbb{E}\left[G_n^L\right]\nonumber \\
        =&\max _{k \leq K} \mathbb{E}\left[\max _{t \leq n} X_{k, t}\right]-\mathbb{E}\left[\max _{t \leq n} X_{I_t, t}\right].
    \end{align}
    So the minimization of extreme regret is achieved by pulling the specific arm with the heaviest tail.
\subsubsection*{Summary/Takeaway}
    The task of a recognition system is to identify the class belongingness of unknown/known test samples. In general, closed-set/open-set recognition systems will specify the score/probability of the test sample with respect to all the trained class samples. Post-recognition scores of the match and non-match samples are helpful for identifying unknown classes. Various recognition systems were analyzed in the context of \ac{OSR}, starting from \ac{SVM} to more complex deep neural networks. Either score  \cite{scheirer2014probability}, \cite{yu2021deep} or reconstruction errors \cite{zhang2016sparse},  \cite{oza2019c2ae}, \cite{yang2021conditional} were considered in fitting the Weibull/reverse Weibull distributions to analyze the tail behavior in deciding upon unknown samples. The authors of \cite{zhang2016sparse}, \cite{oza2019c2ae},\cite{yang2021conditional} considered the \ac{GPD} for tail analysis. The key takeaway is one can analyze the tail behavior of the post recognition system score in identifying the unknown/known test sample with the help of \ac{EVT}. Moreover, the \ac{GEV}/Gumbel distribution helps in getting the maximum cumulative gain for \ac{MAB}s. \\

 \subsection{Various other applications}

\subsubsection{Energy Efficiency}
The number of mobile users escalates, leading to substantial energy consumption in wireless systems. The number of bits transmitted per unit of power consumption is referred to as energy efficiency. One must use energy efficiently to cater for the ever-increasing wireless traffic demands. The following literature survey will present how \ac{EVT} is used for improving energy efficiency.
Considering the double Nakagami fading channel of wireless sensor networks \cite{bahl2015asymptotic}, Bahl {\it et al.}~focused on the energy efficiency of multicast systems. The asymptotic expression of the \ac{CDF} of the minimum \ac{SNR} was derived with the aid of \ac{EVT}. The authors also proved that the minimum \ac{SNR} distribution converges to a Weibull distribution, and they analyzed the throughput. Furthermore, the network energy efficiency is given by the ratio of the total throughput divided by the total power consumption. 
Considering the uncertain nature of wireless channels and the network queue size, the authors in \cite{ji2022energy} focused on the energy efficiency of \ac{MEC} systems. Therein, \ac{EVT} was used for deriving the probability of uncertain events.

\subsubsection{Signal Detection}
Milstein {\it et al.}~\cite{milstein1969robust} studied the detection problem of noisy binary signal. Furthermore, the detection error probability was estimated by using \ac{EVT}. Moreover, in \cite{jeruchim1984techniques}, \ac{BER} estimation was carried out by using \ac{EVT}. In \cite{guida1988comparative}, the performance of \ac{EVT} and \ac{GEV} estimators considering factors like the estimation methods, sample size, and initial distribution is carried out with the aid of Monte Carlo simulations.

\subsubsection{{OFDM}}
By exploiting that the band-limited \ac{OFDM} signal converges weakly to a specific Gaussian random process, the expression for the distribution of the \ac{PMEPR} was derived by Wei {\it et al.}~in \cite{wei2002modern}. Then an approximation of the \ac{PMEPR} distribution of the \ac{OFDM} signal is expressed by applying \ac{EVT}. Additionally, considering an \ac{OFDM} system suffering for high fade rates, Kalyani and Giridhar \cite{kalyani2006extreme} proposed an \ac{EVT}-based Huber-M estimator for decision-directed channel tracking. 
In \cite{jiang2008derivation}, Jiang {\it et al.}~proposed an analytical expression for \ac{PAPR} distribution in \ac{OFDM} systems under the constraint of unequal power allocation to the individual subcarriers using \ac{EVT}. In \cite{hung2014papr}, \ac{EVT} is used to derive the complementary \ac{CDF} of the \ac{PAPR} for the \ac{MRT} and \ac{EGT} techniques conceived for \ac{OFDM} systems.
Considering the transmit antenna correlation in \cite{park2008performance}, the authors derived the \ac{CDF} of the received \ac{SNR}  for the transmit diversity, antenna selection, and spatial multiplexing modes of multi-antenna transmission schemes. Furthermore, \ac{EVT} was used in the asymptotic analysis of maximum throughput transmission.

\subsubsection{Decision Directed Channel Tracking}
\ac{EVT} has been effectively employed for channel estimation using decision directed channel tracking \cite{4151173, kalyani2006extreme, kalyani2006leverage, kalyani2007mitigation}. It has also been used for mitigating narrow-band interference in \ac{OFDM} systems \cite{4151173}. The presence of multiple narrow-band interferences in \ac{OFDM} systems leads to an outlier-contaminated Gaussian regression problem, and Kalyani and Giridhar \cite{4151173} used EVT-based weights to update their M estimator-aided solution. The proposed estimator outperforms both the Gaussian \ac{PDF}-based \ac{ML} estimator and an M estimator based only on Huber’s cost function \cite{4151173}. The authors of \cite{kalyani2006leverage, kalyani2007mitigation} also demonstrated how EVT can be used for reducing error propagation induced by wrong symbol decisions during decision directed channel tracking. While conventional estimators like the 2D-minimum mean square error channel estimator and the expectation maximization based Kalman channel estimators cannot handle such scenarios, the EVT-based solution of \cite{kalyani2006leverage, kalyani2007mitigation} has a significantly better error rate performance. 
\color{black}
\subsubsection{Network Throughput}
In \cite{gesbert2010rate}, \ac{EVT} is used by Gesbert and Kountouris to derive the scaling laws of the network's sum rates by assuming that the number of users per cell increases.
Considering the cases of both perfect and imperfect interference \ac{CSI}, Ji {\it et al.~}\cite{ji2010capacity} derived the asymptotic capacity of a multicast network in a spectrum-sharing system by using \ac{EVT}.

\subsubsection{User Scheduling}
A resource redistributive opportunistic scheduler was designed in \cite{cho2009resource} to enhance resource utilization relying on weighted fairness. The authors used \ac{EVT} to prove that as the number of users becomes large, the throughput loss is linear by proportional to the degree of weighted fairness.
In \cite{low2010optimized}, the opportunistic multicast scheduling scheme was proposed to strike a trade-off between multi-user diversity and multicast gain in cellular networks. Briefly, a subset of users are selected at an optimal rate for delay minimization in each time slot. \ac{EVT} was used for their system delay analysis in both homogeneous and heterogeneous networks. 
In \cite{al2012asymptotic}, Al-Ahmadi studied the asymptotic capacity of opportunistic scheduling in the face of shadowed Nakagami fading using \ac{EVT}. The author showed that the maximum of \ac{i.i.d.} generalized-$K$ \ac{RVs} converges to the Gumbel distribution.
In \cite{kalyani2012analysis}, Kalyani and Karthik utilized \ac{EVT} in deriving the equations of scheduling gain and spectral efficiency in \ac{OFDM} systems for both the proportional fair and maximum rate scheduling algorithms

\subsubsection{Energy Harvesting}

Considering the \ac{WPS} systems in \cite{xia2015efficiency}, Xia and Aissa focused on the power transfer efficiency of the users on the average output \ac{DC} power. Furthermore, they derived the limiting distribution of maximum instantaneous output \ac{DC} power among users in two scenarios namely for both symmetric and asymmetric users. Their result showed Gumbel-type and Frechet-type distributions for the symmetric and asymmetric cases.
In \cite{ding2021harvesting}, Ding proposed a pair of efficient transmission strategies,  i.e., wireless power transfer-assisted non-orthogonal multiple
access and \ac{BAC-NOMA} for \ac{IoT} devices. The asymptotic analysis of the outage probability error floor of \ac{BAC-NOMA} is considered with the aid of \ac{EVT}.
 \subsubsection{Threshold Selection}
 In \cite{m2021efficient}, Abdelmoniem {\it et al.}~focused on designing an efficient gradient compression technique of minimal overhead for distributed training systems. Here, the sparsity of gradients is exploited to model the gradients with the aid of some sparsity-inducing distributions. They used \ac{EVT} to propose a threshold-based sparsification method termed as, Sparsity-Inducing Distribution based Compression.

\subsubsection{Target Detection}
In most target detection problem, \ac{CFAR} detectors assume that the interference probability distribution belongs to a distribution family, e.g.,  generalized gamma or K-distribution. However, without any specific assumptions conceiving the interference, Piotrkowski \cite{Piotrkowski06} proposed a \ac{CFAR} detector by using \ac{EVT}. As a further development,
Broadwater and Chellappa \cite{BroChe10} incorporated the features of the \ac{GPD} and the Kolmogorov–Smirnov statistical test in their algorithm for determining the detection threshold. The proposed approach can adaptively maintain low false alarm rates in the presence of targets.

\subsubsection{Localization}
An indoor localization problem was studied in \cite{YanGuoGuoZhaZha20} by Yang {\it et al.}~based on the received signal
strength indication. The authors proposed a trilateration algorithm for localization by using \ac{EVT}.

\subsubsection{Spectrum Utilization}
Rocke and Wyglinski \cite{RocWyg11} investigated the modeling of the geographically varying spectrum utilization. Specifically, the \ac{GPD} was used to characterize the statistics of the overall spectrum occupancy.
Rocke \cite{Rocke17} focused the attention on the harm claim threshold, which is a regulatory approach conceived for controlling the interference. Accordingly, the receivers are designed to operate within acceptable performance levels controlled by thresholds. The author used \ac{EVT} to characterize the receiver environment.

\subsubsection{Packet Error Rate analysis} \ac{EVT} was also used to approximate the \ac{PER} of uncoded schemes in \ac{AWGN} channels in \cite{mahmood2016packet}.  The approximation holds for any finite \ac{SNR} values and packet lengths instead of the asymptotic case. A generic \ac{BER} approximation expression is presented first by using \ac{EVT}, and then the \ac{PER} expressions are derived. Here, the \ac{PER} is asymptotically approximated by the Gumbel distribution, as the packet length tends to infinity.
\subsubsection{Antenna Array Performance} In \cite{krishnamurthy2019peak}, \ac{EVT} was used by Krishnamurthy {\it et al.}~to find the maximal side lobe value of an antenna array. First, at each angle of the array field, the beampattern distribution expression was approximated by an exponential distribution, and the authors proved that the maximal sidelobe follows the Gumbel distribution.
\subsubsection{Backscatter Communication} \ac{EVT} was used in \cite{li2019capacity} by Li {\it et al.}~to derive an approximate expression for the ergodic capacity of a backscatter communication system, where the tags, e.g., of radio frequency identification (RFID), harvest energy from externally-generated carriers for tag circuit operation and signal backscattering. 
Therein, the tag with the maximum \ac{SNR} is selected for combating double fading.
In \cite{al2020performance}, Al-Badarneh {\it et al.}~studied an RFID backscatter communication system consisting of a monostatic RFID reader and multiple tags. The authors used order statistics to analyze system performance by selecting the $k^{th}$ best \ac{SNR} value among all read-tag links and derived asymptotic expressions for average and effective throughputs.
\subsubsection{Applications for \ac{i.n.i.d.} RVs} Kalyani and Karthik \cite{kalyani2012analysis} used \ac{EVT} to characterize the scheduling gain and spectral efficiency of opportunistic scheduling algorithms in \ac{OFDM} systems. In case of opportunistic scheduling, the scheduling gain and the spectral efficiency are functions of the maximum SNR across users and the paper demonstrates how  we can handle the non-identical nature of the users and harness EVT to study the asymptotic maximum order statistics, depending upon the applications. In \cite{subhash2022asymptotic}, Subhash {\it et al.}~used order statistics to derive the $k^{th}$ maximum of a sequence of non-central chi-square random variables with two degrees of freedom assuming non-identical centrality parameters. Furthermore, the results were used in applications like \ac{UAV}-assisted \ac{IoT} systems and selection combining receivers. EVT was also used to derive the maximum SNR statistics over non-identical relay links in \cite{subhash2021cooperative}. Sagar and Kalyani \cite{sagar2023multi} utilized \ac{EVT} in the asymptotic analysis of best \ac{RIS} selection in a Multi-\ac{RIS} communication system having non-identical links across the RISs.

\section{Lessons and Future Directions}

In the previous sections, we discussed how EVT has been used for simplifying the analysis of different communication systems. While substantial advances have been made, it is also important to keep in mind the limitations of this theory.
\begin{itemize}
    \item First of all, one has to understand that like any asymptotic solutions, these results are only accurate for very high system dimensions and the accuracy of these approximations in finite regimes depends on the statistical nature of the RV sequences considered. For many of the applications, it can be observed that the asymptotic distributions are very close to the exact distributions of the extremes evaluated over finite-length sequences, but this does not guarantee that the same rate of convergence can be expected for all the distributions.
    \item Secondly, as discussed in Section II, the underlying uniformity assumptions are very important and one should carefully evaluate their validity to make use of EVT for the analysis of an extreme RV.
    \item Finally, in the case of extremes over sequences of i.n.i.d. RVs, the limiting distributions may assume more general forms than the EVDs and hence one has to be careful with the assumptions/approximations in their analysis. 
\end{itemize}

\par All the results in the previous sections assume that the extremes are evaluated over sequences of independent RVs. However, in communication theory, we also encounter practical scenarios, where we are interested in extreme order statistics over sequences of correlated RVs. \ac{EVT} also has recently received attention for correlated RVs but has not yet been used for communication theoretic applications. Majumdar \textit{et al.} \cite{majumdar2020extreme} discuss the extreme order statistics over both weakly and strongly correlated sequences of RVs. In applications like best antenna selection or best \ac{UAV} link selection in non-terrestrial networks, the links may be correlated, and hence it is important to take this into account for system design. Studies involving correlated RVs constitute an exciting direction for more applications using \ac{EVT} in communications.  

\par To analyze and optimize the performance of wireless communications technologies such as \ac{URLLC} and massive \ac{MIMO} using \ac{EVT}, the parameter characterizations of the \ac{GEV} distribution and \ac{GPD} are inefficient in terms of samples utilization. In the former, a maximal value data is obtained have a set of $M$ data, where $M$ is the  block size. In other words, the fraction of $\frac{M-1}{M}$ of the data is not utilized. In the latter, if we consider the value $\bar{F}^{-1}(1/N)$ as the threshold over which the associated \ac{GPD} is efficient characterization, $\frac{N-1}{N}$ fraction of the data will not be beneficially leveraged. Specifically, $N=100$ with $\frac{N-1}{N}=99\%$ was considered in some works.  
Higher $M$ and $N$ values give more accurate characterization, but they incur higher delay before reaching convergence.
Moreover, the higher training overheads will significantly degrade data transmission efficiency in massive \ac{MIMO} schemes and the delay of \ac{URLLC}, respectively. Hence, the tradeoff between sample complexity and performance convergence/guarantee deserves dedicated studies for effectively applying \ac{EVT} in \ac{NG} networks.
\begin{figure*}
    \centering

\tikzset{every picture/.style={line width=0.75pt}} 

\begin{tikzpicture}[x=0.75pt,y=0.75pt,yscale=-1,xscale=1]

\draw   (248.27,148.8) .. controls (248.27,126.71) and (291.02,108.8) .. (343.77,108.8) .. controls (396.51,108.8) and (439.27,126.71) .. (439.27,148.8) .. controls (439.27,170.89) and (396.51,188.8) .. (343.77,188.8) .. controls (291.02,188.8) and (248.27,170.89) .. (248.27,148.8) -- cycle ;
\draw    (341.77,108.8) -- (342.25,54.8) ;
\draw [shift={(342.27,52.8)}, rotate = 90.51] [color={rgb, 255:red, 0; green, 0; blue, 0 }  ][line width=0.75]    (10.93,-3.29) .. controls (6.95,-1.4) and (3.31,-0.3) .. (0,0) .. controls (3.31,0.3) and (6.95,1.4) .. (10.93,3.29)   ;
\draw    (434.27,137.07) -- (479.53,111.06) ;
\draw [shift={(481.27,110.07)}, rotate = 150.12] [color={rgb, 255:red, 0; green, 0; blue, 0 }  ][line width=0.75]    (10.93,-3.29) .. controls (6.95,-1.4) and (3.31,-0.3) .. (0,0) .. controls (3.31,0.3) and (6.95,1.4) .. (10.93,3.29)   ;
\draw    (341.77,188.8) -- (341.29,229.8) ;
\draw [shift={(341.27,231.8)}, rotate = 270.67] [color={rgb, 255:red, 0; green, 0; blue, 0 }  ][line width=0.75]    (10.93,-3.29) .. controls (6.95,-1.4) and (3.31,-0.3) .. (0,0) .. controls (3.31,0.3) and (6.95,1.4) .. (10.93,3.29)   ;
\draw    (248.27,148.8) -- (191.27,149.77) ;
\draw [shift={(189.27,149.8)}, rotate = 359.03] [color={rgb, 255:red, 0; green, 0; blue, 0 }  ][line width=0.75]    (10.93,-3.29) .. controls (6.95,-1.4) and (3.31,-0.3) .. (0,0) .. controls (3.31,0.3) and (6.95,1.4) .. (10.93,3.29)   ;
\draw    (404.77,116.8) -- (431.27,70.8) ;
\draw [shift={(432.27,69.07)}, rotate = 119.95] [color={rgb, 255:red, 0; green, 0; blue, 0 }  ][line width=0.75]    (10.93,-3.29) .. controls (6.95,-1.4) and (3.31,-0.3) .. (0,0) .. controls (3.31,0.3) and (6.95,1.4) .. (10.93,3.29)   ;
\draw    (401.77,179.8) -- (442.76,215.49) ;
\draw [shift={(444.27,216.8)}, rotate = 221.04] [color={rgb, 255:red, 0; green, 0; blue, 0 }  ][line width=0.75]    (10.93,-3.29) .. controls (6.95,-1.4) and (3.31,-0.3) .. (0,0) .. controls (3.31,0.3) and (6.95,1.4) .. (10.93,3.29)   ;
\draw    (293.77,115.8) -- (246.87,80.99) ;
\draw [shift={(245.27,79.8)}, rotate = 36.59] [color={rgb, 255:red, 0; green, 0; blue, 0 }  ][line width=0.75]    (10.93,-3.29) .. controls (6.95,-1.4) and (3.31,-0.3) .. (0,0) .. controls (3.31,0.3) and (6.95,1.4) .. (10.93,3.29)   ;
\draw    (279.77,177.8) -- (236.86,210.59) ;
\draw [shift={(235.27,211.8)}, rotate = 322.62] [color={rgb, 255:red, 0; green, 0; blue, 0 }  ][line width=0.75]    (10.93,-3.29) .. controls (6.95,-1.4) and (3.31,-0.3) .. (0,0) .. controls (3.31,0.3) and (6.95,1.4) .. (10.93,3.29)   ;
\draw   (290,28.07) -- (407.27,28.07) -- (407.27,54.07) -- (290,54.07) -- cycle ;
\draw   (423,38.07) -- (546.27,38.07) -- (546.27,69) -- (423,69) -- cycle ;
\draw   (482,99.07) -- (530.27,99.07) -- (530.27,123.07) -- (482,123.07) -- cycle ;
\draw   (422,219) -- (522.27,219) -- (522.27,265.07) -- (422,265.07) -- cycle ;
\draw   (131.27,52.07) -- (246.27,52.07) -- (246.27,98.07) -- (131.27,98.07) -- cycle ;
\draw   (117.27,126.07) -- (194.27,126.07) -- (194.27,175.07) -- (117.27,175.07) -- cycle ;
\draw   (145,211.07) -- (290.27,211.07) -- (290.27,238.07) -- (145,238.07) -- cycle ;
\draw   (315.27,230) -- (375.27,230) -- (375.27,257.07) -- (315.27,257.07) -- cycle ;
\draw    (433.27,162.07) -- (502.29,172.76) ;
\draw [shift={(504.27,173.07)}, rotate = 188.81] [color={rgb, 255:red, 0; green, 0; blue, 0 }  ][line width=0.75]    (10.93,-3.29) .. controls (6.95,-1.4) and (3.31,-0.3) .. (0,0) .. controls (3.31,0.3) and (6.95,1.4) .. (10.93,3.29)   ;
\draw   (504,152) -- (558.27,152) -- (558.27,186.07) -- (504,186.07) -- cycle ;

\draw (271,122) node [anchor=north west][inner sep=0.75pt]   [align=left] {{\Large  Possible future}\\{\Large  \ \ applications}};
\draw (293,29) node [anchor=north west][inner sep=0.75pt]   [align=left] {Green networks};
\draw (425,44) node [anchor=north west][inner sep=0.75pt]   [align=left] {Cell-free systems};
\draw (141,55) node [anchor=north west][inner sep=0.75pt]   [align=left] {Non-terrestrial\\ \ \ \ networks};
\draw (322,234) node [anchor=north west][inner sep=0.75pt]   [align=left] {AI/ML};
\draw (486,100) node [anchor=north west][inner sep=0.75pt]   [align=left] {V2X};
\draw (421,219) node [anchor=north west][inner sep=0.75pt]   [align=left] {Fluid antenna\\ \ \ \ \ systems};
\draw (151,214) node [anchor=north west][inner sep=0.75pt]   [align=left] {Holographic MIMO};
\draw (125,129) node [anchor=north west][inner sep=0.75pt]   [align=left] {Wireless\\security};
\draw (510,159) node [anchor=north west][inner sep=0.75pt]   [align=left] {ISAC};

\end{tikzpicture}

\caption{Possible future directions for applying EVT}
    \label{future_applications}
\end{figure*}
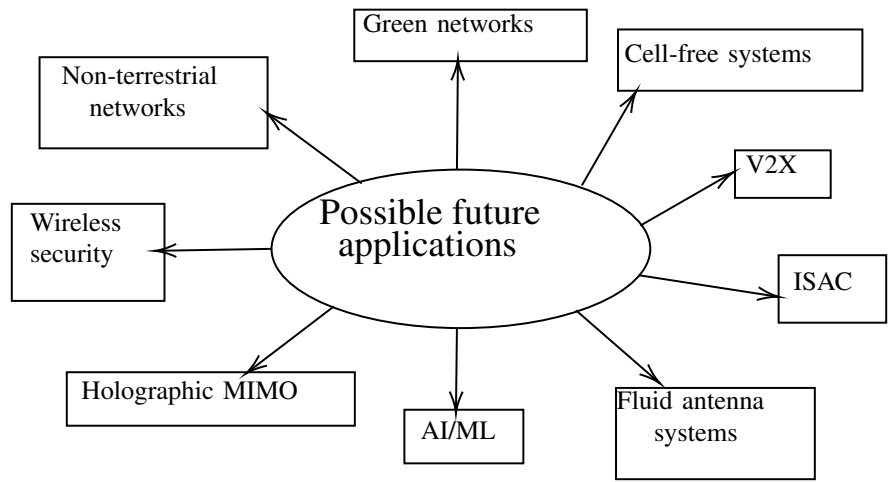
\par EVT can also be used both for the performance analysis and for the design of \ac{NG} wireless communication systems. In the following paragraphs we discuss some further topics of interest for \ac{NG} wireless networks, where EVT can be a potent tool for both design and analysis. 
\subsubsection*{Green Networks}
Considering environmental concerns, activating hundreds of antennas simultaneously in massive MIMO systems is inefficient due to the considerable energy consumption escalation caused by the active antennas' power thirsty radio frequency (RF) chains \cite{marinello2020antenna}. Selecting the $K$ best antennas in terms of the received signal power is an efficient choice in such scenarios. The statistical analysis of such systems involves characterizing the joint distribution of the SNR of the selected antennas. Investigating multivariate EVT as a potential research area could provide solutions in such scenarios, making it an interesting proposition for future exploration. 
\subsubsection*{Cell-Free MIMO Systems} Cell-free massive MIMO breaks traditional cell boundaries using distributed antennas for seamless connectivity and improved performance. It consists of numerous \ac{AP}s and a smaller number of active users spread across a large area \cite{biswas2021ap}. Each user is assigned a subset of the available APs for service and the specific AP having the strongest link is selected to serve that particular user. In \cite{duc2023power}, the authors employed multiple aerial relay stations (ARSs) of a cell-free system, where each ARS is equipped with multiple antennas. Each user is connected to a subset of ARSs, and the ARS having the strongest link is selected for communication. Hence, EVT can be applied in the performance analysis of such systems.
\subsubsection*{Non-Terrestrial Networks} Enhanced coverage is one of the pressing demands of the future wireless networks and integrated ground-air-space wireless networks are expected to play a major role in accomplishing this need. These new heterogeneous networks bring in a new set of challenges and characterizing these systems is extremely important \cite{liu2022operation, wu2023energy,azari2022evolution,rinaldi2020non}. EVT can be useful for studying such large systems supporting a large number of nodes/links determining the system performance. 
\par Similarly, EVT can be used for analyzing the performance along the strongest user link selection in non-terrestrial networks. Multiple UAV systems are deployed for tasks such as surveillance, delivery, and communication, providing improved coverage, flexibility, and resilience compared to single UAV systems. In \cite{singh2020uav}, Singh and Agrawal employed multiple UAVs as relays in a UAV-assisted cooperative communication network between ground users. They proposed two UAV selection strategies, namely the Best Harmonic Mean (HM) and the Best Downlink SNR (BDS), based schemes. In a different study Bansal and Agrawal \cite{bansal2023ris} explored the integration of \ac{RIS} in UAV-based multiuser downlink communications. In this scenario, a flying UAV provides service to several single-antenna users through multiple RISs installed on different buildings. In all these endeavors, selecting the optimal UAV is paramount, and EVT can offer useful insights into the performance of such systems. 
\subsubsection*{Wireless Security} Security in wireless networks is a problem of growing concern and the security of even the weakest user is critical. EVT can be used in the performance study of such scenarios. Such studies for the case of UAV assisted systems have been explored in \cite{LeiWanParAnsJiaPanAlo20, SubKalBadAlo22}.  
\subsubsection*{Vehicle-to-Everything communication} Autonomous driving and \ac{V2V} communication critically rely on \ac{NG} wireless networks \cite{noor20226g,zhu2022intelligent}. EVT has already been explored for optimizing vehicular communication networks \cite{liu2018ultra,mehrnia2021wireless,mehrnia2024multivariate} and EVT can be instrumental in designing these future vehicular networks. 
\subsubsection*{Holographic MIMO/Active \ac{RIS}} Large intelligent reflecting surfaces capable of controlling the scattering, reflection, and refraction characteristics of the radio waves are studied as an enabler for improving the quality of \ac{NG} wireless networks. The performance tradeoffs in different systems relying on a \ac{RIS} have gained much attention recently. Furthermore, there is a growing interest in exploring the feasibility of tunable and active metasurfaces to replace passive \ac{RIS} systems. The benefits of such active systems over traditional wireless networks as well as over networks enhanced by \ac{RIS}s, inspire substantial research in the community. EVT can be used to derive scaling laws and characterize the asymptotic performance to examine the viability of \ac{RIS} as a conducive technology for future networks. 
\subsubsection*{Integrated Sensing and Communication}
{EVT can be used to analyze and design \ac{ISAC} systems. \ac{ISAC} systems integrate sensing and communication into a unified framework, enhancing resource efficiency and enabling advanced applications such as autonomous driving, and environmental monitoring. For instance, in multistatic sensing and communication scenarios, selecting a subset of base stations having the best collective sensing and communication capabilities is crucial for both robust and efficient target detection, as well as for high-quality user service. EVT supports the analysis and design of such selection processes by modeling the distribution of extreme performance metrics, such as the highest \ac{SNR}, detection probabilities, or data rates among the BSs available. Hence, by applying EVT, \ac{ISAC} systems can optimize their configuration to achieve superior sensing and communication performance, even in dynamic and extreme propagation environments.}
\subsubsection*{Artificial Intelligence/ML} With \ac{AI} permeating all fields of engineering and sciences, telecommunication engineers are also exploring the benefits of using large \ac{AI} models for \ac{NG} communication systems. These new models have to ensure that the each worst-case system performance still complies with the minimum quality of service requirements. We believe that EVT can be extremly useful for the performance analysis of these systems. Furthermore, EVT has recently been harnessed for improving machine learning algorithms \cite{rudd2017extreme,geng2020recent}, and for characterizing the performance of machine learning algorithms \cite{tizpaz2024worst,li2024generalized}.  
\par In most of the literature on \ac{ML}-aided solutions \ac{EVT} was used in post-recognition score analysis \cite{scheirer2014probability, zhang2016sparse, oza2019c2ae,yang2021conditional}. Incorporating \ac{EVT} into deep learning frameworks would be a more challenging task. Integrating \ac{EVT} techniques into deep learning models will enhance the model's robustness and can be an interesting area of future research. It would be a very interesting problem to study how the knowledge of the extremes can be incorporated into the training of \ac{ML} models for improving their robustness to outliers and unknown classes. This would also involve extending univariate \ac{EVT} analysis to multivariate \ac{EVT}.
\par Apart from the applications mentioned above, EVT can be useful in any scenario, where the asymptotic order statistics of a sequence of RVs are of interest. All the literature discussed provides great inspiration for future exploration of this compelling theory. 

\bibliographystyle{IEEEtran}

 \bibliography{IEEEabrv,Elzanaty_bibliography.bib,LISBIB.bib,references.bib}
\end{document}

%% file: Definitions.tex
\newcommand{\Nb} {N_\mathsf{B}}
\newcommand{\Nm} {N_\mathsf{M}}
\newcommand{\Nl} {N_\mathsf{R}}
\newcommand{\Ns} {N}

\newcommand{\Hm} {\mathbf{H}}
\newcommand{\y} {\mathbf{y}}
\newcommand{\x} {\mathbf{x}}
\newcommand{\n} {\boldsymbol{\omega}}
\newcommand{\Hbm} {\mathbf{H}_\mathsf{BM}}
\newcommand{\Hbl} {\mathbf{H}_\mathsf{BR}}
\newcommand{\Hlm} {\mathbf{H}_\mathsf{RM}}
\newcommand{\bmOmega} {\bm{\Omega}}
\newcommand{\alphar} {\bm{\alpha}_\mathsf{r}}
\newcommand{\alphat} {\bm{\alpha}_\mathsf{t}}

\newcommand{\rhobm} {\rho_\mathsf{BM}}
\newcommand{\rhobl} {\rho_\mathsf{BR}}
\newcommand{\rholm} {\rho_\mathsf{RM}}
\newcommand{\rhoblm} {\rho_\mathsf{BRM}}

\newcommand{\tbm} {\tau_{bm}}
\newcommand{\tbl} {\tau_{br}}
\newcommand{\tlm} {\tau_{rm}}
\newcommand{\tdm} {\tau_{sm}}

\newcommand{\taubm} {\tau_\mathsf{BM}}
\newcommand{\taubl} {\tau_\mathsf{BR}}
\newcommand{\taulm} {\tau_\mathsf{RM}}
\newcommand{\taudm} {\tau_\mathsf{DM}}

\newcommand{\dbmo} {d_\mathsf{BM}}
\newcommand{\ddmo} {d_\mathsf{SM}}
\newcommand{\dblo} {d_\mathsf{BR}}
\newcommand{\dlmo} {d_\mathsf{RM}}

\newcommand{\xilm} {\xi_{\mathsf{RM}}}
\newcommand{\etal} {\eta_r}
\newcommand{\etam} {\eta_m}

\newcommand{\xibm} {\xi_{\mathsf{BM}}}

\newcommand{\chibm}{\chi_{\mathsf{B}\mathsf{M}}}
\newcommand{\chiblm}{\chi_{\mathsf{B}\mathsf{R}\mathsf{M}}}

\newcommand{\phibm} {\phi_\mathsf{BM}}
\newcommand{\phibl} {\phi_\mathsf{BR}}
\newcommand{\philm} {\phi_\mathsf{RM}}
\newcommand{\phiL} {\phi_{\mathsf{R}}}
\newcommand{\phidm} {\phi_\mathsf{SM}}

\newcommand{\thetadm} {\theta_\mathsf{SM}}
\newcommand{\thetabm} {\theta_\mathsf{BM}}
\newcommand{\thetabl} {\theta_\mathsf{BR}}
\newcommand{\thetalm} {\theta_\mathsf{RM}}

\newcommand{\thetam} {\theta_m}
\newcommand{\thetab} {\theta_b}
\newcommand{\thetal} {\theta_l}

\newcommand{\xms} {x_{\mathsf{M}}}
\newcommand{\yms} {y_{\mathsf{M}}}
\newcommand{\zms} {z_{\mathsf{M}}}

\newcommand{\phim} {\phi_m}
\newcommand{\phib} {\phi_b}
\newcommand{\phil} {\phi_r}

\newcommand{\thetaL} {\theta_{\mathsf{R}}}

\newcommand{\dm} {d_{m}}
\newcommand{\dl} {d_{r}}
\newcommand{\db} {d_{b}}
\newcommand{\dbm} {d_{bm}}
\newcommand{\ddm} {d_{sm}}
\newcommand{\dbl} {d_{br}}
\newcommand{\dlm} {d_{rm}}
\newcommand{\fn} {f_n}

\newcommand{\xb}{x_b}
\newcommand{\yb}{y_b}
\newcommand{\nb}{w}
\newcommand{\zb}{z_b}

\newcommand{\xm}{x_m}
\newcommand{\ym}{y_m}
\newcommand{\zm}{z_m}

\newcommand{\xd}{x_s}
\newcommand{\yd}{y_s}
\newcommand{\zd}{z_s}

\newcommand{\Cbrm}{C_{brm}}
\newcommand{\BCr}{\bar{C}_{r}}
\newcommand{\BCk}{\bar{C}_{k}}
\newcommand{\Cbkm}{C_{bkm}}
\newcommand{\betam}{\beta_{m}}
\newcommand{\fc}{f_{0}}
\newcommand{\thetar}{\theta_r}
\newcommand{\thetak}{\theta_k}

\newcommand{\xdi}{x_d^{(0)}}
\newcommand{\ydi}{y_d^{(0)}}
\newcommand{\zdi}{z_d^{(0)}}

\newcommand{\xmi}{x_m^{(0)}}
\newcommand{\ymi}{y_m^{(0)}}
\newcommand{\zmi}{z_m^{(0)}}

\newcommand{\cosam}{c_{\alpha_{\mathsf{M}}}}
\newcommand{\sinam}{s_{\alpha_{\mathsf{M}}}}
\newcommand{\cosbm}{c_{\beta_{\mathsf{M}}}}
\newcommand{\sinbm}{s_{\beta_{\mathsf{M}}}}

\newcommand{\cosgm}{c_{\gamma_{\mathsf{M}}}}
\newcommand{\singm}{s_{\gamma_{\mathsf{M}}}}

\newcommand{\xl}{x_r}
\newcommand{\yl}{y_r}
\newcommand{\zl}{z_r}
\newcommand{\godm}{\mathsf{G}^{(1)}_{sm}}
\newcommand{\gobm}{\mathsf{G}^{(1)}_{bm}}
\newcommand{\gobl}{\mathsf{G}^{(1)}_{bl}}
\newcommand{\gtbm}{\mathsf{G}^{(2)}_{bm}}
\newcommand{\gtdm}{\mathsf{G}^{(2)}_{sm}}
\newcommand{\gtBM}{G_{\mathsf{B,M}}^{(b,m)}}
\newcommand{\gtDM}{G_{\mathsf{D,M}}^{(d,m)}}

\newcommand{\gtbl}{\mathsf{G}^{(2)}_{br}}
\newcommand{\gtlm}{\mathsf{G}^{(2)}_{rm}}
\newcommand{\gtLM}{G_{\mathsf{R,M}}^{(r,m)}}

\newcommand{\gtBL}{G_{\mathsf{B,R}}^{(b,r)}}

\newcommand{\xmo}{x_{\mathsf{M}}}
\newcommand{\ymo}{y_{\mathsf{M}}} 
\newcommand{\zmo}{z_{\mathsf{M}}} 
\newcommand{\pmo}{\mathbf{p}_{\mathsf{M}}}

\newcommand{\pdo}{\mathbf{p}_{\mathsf{D}}}
\newcommand{\ado}{\boldsymbol{\Phi}_{\mathsf{D}}}
\newcommand{\amo}{\boldsymbol{\phi}_{\mathsf{M}}}
\newcommand{\abo}{\boldsymbol{\Phi}_{\mathsf{B}}}
\newcommand{\alo}{\boldsymbol{\Phi}_{\mathsf{R}}}

\newcommand{\am}{{a}_{\mathsf{M}}}

\newcommand{\xbo}{x_{\mathsf{B}}}
\newcommand{\ybo}{y_{\mathsf{B}}} 
\newcommand{\zbo}{z_{\mathsf{B}}} 
\newcommand{\pbo}{\mathbf{p}_{\mathsf{B}}}

\newcommand{\xlo}{x_{\mathsf{R}}}
\newcommand{\ylo}{y_{\mathsf{R}}} 
\newcommand{\zlo}{z_{\mathsf{R}}} 
\newcommand{\plo}{\mathbf{p}_{\mathsf{R}}}

\newcommand{\hbm}{h_{bm}}
\newcommand{\hbl}{h_{br}}
\newcommand{\hlm}{h_{rm}}
\newcommand{\omegal}{\omega_{r}}

\newcommand{\xbs}{{x}_{\mathsf{B}}}
\newcommand{\ybs}{{y}_{\mathsf{B}}}
\newcommand{\zbs}{{z}_{\mathsf{B}}}
\newcommand{\xds}{{x}_{\mathsf{D}}}
\newcommand{\yds}{{y}_{\mathsf{D}}}
\newcommand{\zds}{{z}_{\mathsf{D}}}
\newcommand{\xls}{{x}_{\mathsf{R}}}
\newcommand{\yls}{{y}_{\mathsf{R}}}
\newcommand{\zls}{{z}_{\mathsf{R}}}
\newcommand{\dant}{{d}_{\mathsf{ant}}}
\newcommand{\muBM} {\mu_{b,\mathsf{BM}}[n]}
\newcommand{\muBLM} {\mu_{b,\mathsf{BRM}}[n]}
\newcommand{\muBMc} {\mu^{*}_\mathsf{BM}[n]}
\newcommand{\muBLMc} {\mu^{*}_\mathsf{BRM}[n]}

\newcommand{\mubm} {\mu_{bm}[n]}
\newcommand{\cmubmp} {\mu_\mathsf{bm'}^*[n]}

\newcommand{\mublm} {\mu_{brm}[n]}
\newcommand{\mubmc} {\mu^{*}_\mathsf{bm}[n]}
\newcommand{\mublmc} {\mu^{*}_\mathsf{brm}[n]}

\newcommand{\alphamo}{\alpha_{\mathsf{M}}}
\newcommand{\betamo}{\beta_{\mathsf{M}}}
\newcommand{\gammamo}{\gamma_{\mathsf{M}}}

\newcommand{\Phim}{\Phi_{\mathsf M}}

\newcommand{\PEB}{\mathsf{PEB}}
\newcommand{\OEB}{\mathsf{OEB}}

%% file: Acronyms.tex
\begin{acronym}[WPT-NOMA]
\acro{CLT}{central limit theorem}
\acro{RVs}{random variables}
\acro{EVT}{extreme value theory}
\acro{EVDs}{extreme value distributions}
\acro{CR}{cognitive radio}
\acro{RIS}{reconfigurable intelligent surface}
\acro{MIMO}{multiple-input multiple-output}
\acro{EMF}{electromagnetic fields}
\acro{SNR}{signal-to-noise ratio}
\acro{URLLC}{ultra-reliable and low-latency communication}
\acro{GPD}{generalized Pareto distribution}
\acro{ML}{machine learning}
\acro{OFDM}{orthogonal frequency-division multiplexing}
\acro{i.i.d.}{independent and identically distributed}
\acro{i.n.i.d.}{independent and non-identically distributed}
\acro{CDF}{cumulative distribution function}
\acro{AP}{access point}
\acro{FSO}{free-space optical communication}
\acro{SIR}{signal-to-interference ratio}
\acro{WPS}{wireless powered communication system}
\acro{BEP}{bit error probability}
\acro{BER}{bit error rate}
\acro{SU}{secondary user}
\acro{PT}{primary transmitter}
\acro{PR}{primary receiver}
\acro{ST}{secondary transmitter}
\acro{SR}{secondary receiver}
\acro{ESC}{ergodic secrecy capacity}
\acro{SINR}{signal-to-interference-plus-noise ratio}
\acro{MGF}{moment-generating function}
\acro{PDF}{probability density function}
\acro{MRC}{maximum-ratio combining}
\acro{TAS}{transmit antenna selection}
\acro{CSI}{channel state information}
\acro{QoS}{quality-of-service}
\acro{AF}{amplify-and-forward}
\acro{DF}{decode-and-forward}
\acro{SISO}{single-input single-output}
\acro{BS}{base station}
\acro{MS}{mobile stations}
\acro{MISO}{multiple-input single-output}
\acro{SC}{selection combining}
\acro{OC}{optimal combining}
\acro{ZFBF}{zero-forcing beamforming}
\acro{RBF}{random beamforming}
\acro{SDMA}{space division multiple access}
\acro{SNDR}{signal-to-noise-plus-distortion ratio}
\acro{E2E}{end-to-end}
\acro{AoI}{age of information}
\acro{V2V}{vehicle-to-vehicle}
\acro{IoT}{Internet of Things}
\acro{GEV}{generalized extreme value}
\acro{MEC}{multi-access edge computing}
\acro{RSU}{roadside unit}
\acro{VUEs}{vehicular user equipments}
\acro{UAV}{unmanned aerial vehicle}
\acro{OSR}{open set recognition}
\acro{MAB}{multi-armed bandits}
\acro{SVM}{support vector machine}
\acro{EVM}{extreme value machine}
\acro{NN}{neural network}
\acro{CAP}{compact abating probability}
\acro{SRC}{sparse representation-based classification}
\acro{AVs}{activation vectors}
\acro{MAV}{mean activation vector}
\acro{DDRN}{deep discriminative representation network}
\acro{CIP}{class-inclusion probability}
\acro{OSFD}{open set fault diagnosis}
\acro{SOSFD}{shared-domain OSFD}
\acro{COSFD}{cross-domain OSFD}
\acro{MC}{multi-task based feature extractor and classifiers}
\acro{GAN}{generative adversarial network}
\acro{CG}{counterfactual generative adversarial network}
\acro{OSC}{open-set classifier}
\acro{CNN}{convolutional neural network}
\acro{OSmIL}{open set model with incremental learning}
\acro{PAC}{probably approximately correct}
\acro{WSN}{wireless sensor networks}
\acro{PMEPR}{peak-to-mean envelope power ratio}
\acro{PAPR}{peak-to-average power ratio}
\acro{MRT}{maximum ratio transmission}
\acro{EGT}{equal gain transmission}
\acro{RRO}{resource redistributive opportunistic}
\acro{OMS}{opportunistic multicast scheduling}
\acro{DC}{direct current}
\acro{WPT-NOMA}{WPT-assisted non-orthogonal multiple access}
\acro{BAC-NOMA}{backscatter communication-assisted non-orthogonal multiple access}
\acro{CFAR}{constant false alarm rate}
\acro{PER}{packet error rate}
\acro{AWGN}{additive white Gaussian noise}
\acro{SWIPT}{simultaneous wireless information and power transfer}
\acro{NG}{next-generation}
\acro{ISAC}{integrated sensing and communication}
\acro{AI}{artificial intelligence}
\acro{3GPP}{3rd Generation Partnership Project}
\end{acronym}